\documentclass[11pt]{article}

\usepackage{graphicx} 
\usepackage{subfigure}
\usepackage{geometry}
\usepackage{amsmath}
\usepackage{amssymb}
\usepackage{breqn}
\usepackage{multirow}
\usepackage{setspace}
\usepackage{amsthm}
\usepackage{tabularx}
\usepackage{algorithm}
\usepackage{algpseudocode}
\usepackage{standalone}
\usepackage{float}
\usepackage{authblk} 
\usepackage{natbib}
\usepackage[colorlinks=true, linkcolor=blue, citecolor=blue, urlcolor=blue]{hyperref}

\usepackage{xcolor}

\usepackage{siunitx}
\usepackage{longtable}
\usepackage{array}
\usepackage{enumitem}
\usepackage{booktabs}
\usepackage{comment}
\usepackage{tcolorbox}

\newtheorem{theorem}{Theorem}

\newtheorem{remark}{Remark}
\newtheorem{lemma}{Lemma}

\title{Likelihood-Based Regression for Weibull Accelerated Life Testing Model Under Censored Data}

\author[1]{Rahul Konar}
\author[2]{Ramnivas Jat}
\author[3]{Neeraj Joshi\thanks{\textbf{Corresponding Author:} Neeraj Joshi, Department of Mathematics, Indian Institute of Technology Delhi, Hauz Khas, New Delhi, 110016, India. Email: \href{mailto:njoshi@maths.iitd.ac.in}{njoshi@maths.iitd.ac.in}}\thanks{ORCID: \href{https://orcid.org/0000-0001-9017-4255}{0000-0001-9017-4255}}}
\author[4]{Raghu Nandan Sengupta\thanks{ORCID: \href{https://orcid.org/0000-0002-2864-660X}{0000-0002-2864-660X}}}

\affil[1]{Indian Statistical Institute, Kolkata, India}
\affil[2]{Moody’s Ratings, Bengaluru, India}
\affil[3]{Department of Mathematics, Indian Institute of Technology Delhi, New Delhi, India}
\affil[4]{Department of Management Sciences, Indian Institute of Technology Kanpur, Kanpur, India}

\date{\today}

\begin{document}
\maketitle
\renewcommand{\thefootnote}{\fnsymbol{footnote}}
\setcounter{footnote}{0}
\footnotetext{This work originates from the defended M.Tech thesis work of Ramnivas Jat (Roll No. 21114016), and SURGE program-supported research of Rahul Konar (Roll No. S250104) carried out at the Indian Institute of Technology Kanpur, Kanpur, India.}
\renewcommand{\thefootnote}{\arabic{footnote}}
\onehalfspacing

\noindent\textbf{Abstract:} In this paper, we investigate accelerated life testing (ALT) models based on the Weibull distribution with stress-dependent shape and scale parameters. Temperature and voltage are treated as stress variables influencing the lifetime distribution. Data are assumed to be collected under Progressive Hybrid Censoring (PHC) and Adaptive Progressive Hybrid Censoring (APHC). A two-step estimation framework is developed. First, the Weibull parameters are estimated via maximum likelihood, and the consistency and asymptotic normality of the estimators are established under both censoring schemes. Second, the resulting parameter estimates are linked to the stress variables through a regression model to quantify the stress–lifetime relationship. Extensive simulations are conducted to examine finite-sample performance under a range of parameter settings, and a data illustration is also presented to showcase practical relevance. The proposed framework provides a flexible approach for modeling stress-dependent reliability behavior in ALT studies under complex censoring schemes.
\section{Introduction and Background}\label{sec1}
\noindent ALT models are essential frameworks in reliability and life-testing studies. In ALT experiments, products are subjected to elevated stress conditions to induce failures more quickly than under normal operating environments, enabling the collection of sufficient life data within a shorter time. Conventional life testing is time consuming and highly dependent on undesirable environmental conditions. In such situations, ALT provides an efficient and practical alternative. Common accelerating factors include temperature, voltage, pressure, humidity, etc. Using appropriate statistical methods, lifetime data obtained under accelerated conditions are extrapolated to assess system behavior and lifetime under normal operating conditions. ALT is widely used in industries to maintain the long-term reliability for product design, maintenance planning, and quality assurance.\\
\indent Exploring life-stress relationship is critical for data analysis and ALT planning. The type of stress used is determined by the type of component being tested as well as the stresses that occur during normal operation. To obtain the early failures, mechanical parts are frequently subjected to vibration, dampness, and random shock. On the other hand, temperature, humidity, vibration, electrical current, voltage, etc., increase the degradation of electronic components. These are the most widely employed stressors in general, and they can be combined depending on the actual working circumstances and the mechanisms underlying the failure process. The stress translation function (STF) determines the impact of each stress variable on the lifetime of the system.\\
\indent The most frequent life-stress relationships under ALT are described by the \textit{Arrhenius Life Stress Model} (ALSM) and \textit{Inverse Power Law Model} (IPLM). These models are useful to exhibit the thermal and non-thermal accelerated variables, respectively. The ALSM employs the \textit{Arrhenius reaction rate equation}, which in general form can be written as, $L(T)=ce^{b/T},$
where $L$ denotes a measurable quantity that characterizes the lifetime of a system such as mean lifetime, median lifetime, etc. and $T$ represents the level of stress applied to the system such as temperature. In this context, $c$ and $b$ are model parameters that need to be estimated, where $c$ is a positive constant. Under IPLM, the relationship between the life-measure $L$ and the stress level $U$ is expressed as, $L(U)=KU^{-n},$ where $K(> 0)$ and $n$ are model parameters that need to be estimated. Here $U$ represents the applied stress level, which can be usage rate, voltage, humidity, pressure, or any other appropriate stress variable. The ALSM and IPLM can be combined to obtain a more comprehensive life-stress relationship, which can be expressed as, $L(T,U)=Ce^{b/T}U^{-n}.$ A more simplified logarithmic transformation of this relation is,
$L(T,U)=\log C+b/T-n\log U.$\\
\indent Another useful concept in life testing experiments is \textit{Censoring}. The type of data used in statistical models determines the conclusions that can be drawn, so it is critical to take this into account when analyzing lifetime data. Complete data sets, however, might not always be accessible because of a potential system failure. Complete data describes circumstances in which the life cycle of each test subject or sample is known or recorded. On the other hand, censored data is incomplete
or truncated, in which the failure times of some sample units are unknown or occur after the test period. It is typical for units to be lost or removed from reliability and life testing trials prior to failure, resulting in insufficient data on failure times for all units. Censored data refers to this kind of information. Censoring may also occur unintentionally, as when a unit is accidentally damaged, or it may be deliberate, as when units are purposefully removed to save time and money. Some commonly used censoring schemes are Type-I, Type-II, Progressive, Hybrid, Interval, and Random censoring schemes. Under a given censoring scheme, the experiment gets terminated after attaining a pre-fixed number of failures or pre-fixed time or some other stopping criteria based on the need of a particular censoring.\\
\indent The literature on ALT is extremely vast and still being explored. \cite{nelson1982life, nelson1990accelerated} discussed the relevance of optimal planning for performing ALT on a single component system and underlined the need of selecting the optimum stress level and stress variables during the testing process to produce the most precise estimates. \cite{meeker1984comparison} have provided a comparison of ALT plans for Weibull and lognormal distributions under Type-I censoring. \cite{bai1992optimal} have proposed the optimal design of partial ALT for an exponential distribution under Type-I Censoring. \cite{meeker_luvalle_1995} have introduced an ALT model based on reliability kinetics. \cite{meeker1998reliability} have discussed several techniques of performing ALT and identified three main approaches. These approaches include increasing the amount of stress, increasing the rate of usage, and combining the effects of increasing both stress and usage levels. They also discussed about combinations of these methods in ALT. \cite{tang1999planning} have discussed ALT for censored two-parameter exponential distributions. \cite{yang2005acceleratedusage} provided a thorough explanation of ALT with a greater usage rate. \cite{ding2009accelerated} have introduced the ALT sampling plans for the Weibull distribution under Type-I progressive interval censoring with random removals. \cite{lin2012weibull} have investigated Type-I APHC, for a two-parameter Weibull distribution. \cite{abushal2013estimating} have estimated the Pareto parameters under progressive censoring data for constant-partial ALT. \cite{collins2013accelerated} have discussed accelerated test methods for reliability prediction. \cite{ismail2014stepstress} discussed the ML estimation of Weibull distribution parameters and the acceleration factor. Using two forms of PHC methods, this estimation is performed for step-stress partial ALT models. \cite{sobhi2016exponentiated} have discussed the usage of adaptive Type-II progressive censoring algorithms in reliability and life testing trials under exponentiated Weibull model. \cite{wu2016statistical} discussed the statistical analysis of dependent competing risks model in ALT under PHC using copula function. \cite{wang2021multicomponent} have estimated reliability is a multicomponent system using multilayer ALT. \cite{indmeskine2023review} came up with a review of ALT planning to develop predictive reliability models for electronic components based on design of experiments. \cite{fathi2024inference} came up with the inference on Weibull inverted exponential distribution under progressive first-failure censoring with constant-stress partial ALT. \cite{collins2024halt} have presented a detailed statistical perspective of ALT models. Some other important contributions include \cite{sinha2003modelling}, \cite{monroe2008experimental}, \cite{abdelhamid2009constant}, \cite{seo2009design}, \cite{kim2011life}, \cite{lee2012glm}, \cite{yang2013novel}, \cite{wang2014new}, \cite{ismail2015optimum}, \cite{ghaly2016estimation}, \cite{wang2018constantstress}, \cite{wang2020coffinmanson}, \cite{bai2020statistical}, \cite{liu2022grey}, \cite{wu2022interval}, \cite{feng2024reliability}, and \cite{yao2024inference}.\\
\indent In this paper, we further explore the Weibull ALT model, and propose a novel two-step likelihood-based regression approach to estimate the underlying parameters and stress coefficients, under PHC and APHC schemes. The main contributions of this research are described in the following three points.
\begin{enumerate}
    \item We prove consistency of ML estimators of shape and scale parameters of Weibull accelerated lifetimes under PHC and APHC.
    \item We prove asymptotic normality of ML estimators of shape and scale parameters of Weibull accelerated lifetimes under PHC and APHC.
    \item We introduce a novel two-step estimation procedure, consisting of ML estimation followed by ordinary least squares (OLS) regression with \textit{Murphy-Topel corrections}, to estimate the covariate effects of external stresses on system lifetimes.
\end{enumerate}

\indent The rest of the paper is organized as follows: we first provide the ALT model description in Section \ref{sec2}, where we establish the association between stress variables
and the lifetime distribution characteristics. The effect of external stresses (voltage and temperature) on model parameters is also considered. Then we introduce PHC and APHC schemes under the proposed Weibull ALT model. In Section \ref{sec3}, the ML estimators are developed for the model parameters under the PHC and APHC schemes. We also study some interesting properties of the ML estimators. Motivated by the proposed linear relationship of the Weibull parameters with temperature and voltage, we present regression estimation framework for the coefficients of life-stress equations in Section \ref{sec4}. In Section \ref{sec6}, we provide large scale simulations under
different parameter configurations to analyze the accuracy of our theoretical findings. A data illustration is also provided in Section \ref{sec7} for illustrative purpose. Finally, we provide
some concluding thoughts and future research directions in Section \ref{sec8}.

\section{The Weibull ALT Model}\label{sec2}
\noindent Suppose $X_1, X_2, \ldots, X_n$ are independent and identically distributed (i.i.d.) lifetimes of the test units from the sample of size $n$. Further, suppose that $X_1, X_2, \ldots, X_n$ follows a two-parameter Weibull distribution with shape and scale parameters $\alpha$ and $\lambda$ respectively. Then the probability density function (PDF) and the cumulative distribution function (CDF) of $X_1, X_2, \ldots, X_n$ are given by
\begin{equation}\label{eq1}
f(x;\alpha,\lambda) = \frac{\alpha}{\lambda} \bigg(\frac{x}{\lambda}\bigg)^{\alpha-1} \exp \bigg[{-\bigg(\frac{x}{\lambda}\bigg)^\alpha}\bigg]; \quad x>0,\; \alpha>0,\; \lambda>0,
\end{equation}
and
\begin{equation}\label{eq2}
F(x;\alpha,\lambda) = 1 - \exp \bigg[{-\bigg(\frac{x}{\lambda}\bigg)^\alpha}\bigg],
\end{equation}
respectively. The corresponding hazard rate function is given by $h(t) = \alpha x^{\alpha -1}/\lambda^\alpha.$ Here the shape parameter determines whether the hazard function is increasing, decreasing or constant with respect to time, depending on whether $\alpha>1$, $\alpha=1$ or $\alpha<1$. The scale parameter $\lambda$ is also the characteristic life of the Weibull distribution, i.e., a time point when the probability of failure is independent of the parameters of the distribution. The Weibull distribution has been shown to be very adaptive and valuable in the area of ``Reliability and Life Testing" due to its capacity to mimic various failure patterns. Its adaptability makes it particularly valuable in ALT, where it helps to estimate product reliability under varying stress conditions.\\
\indent Following the combined ALSM and IPLM life-stress relationship proposed in Section \ref{sec1}, we relate the shape and scale parameters of the proposed Weibull model to external stresses via log-linear STF, such that
\begin{equation}\label{eq3}
\lambda = a_0+\frac{a_1}{T}+a_2 \log V,
\end{equation}
\begin{equation}\label{eq4}
\alpha = c_0+\frac{c_1}{T}+c_2 \log V,
\end{equation}
where $T$ and $V$ are temperature and voltage, respectively, $a_i$ and $c_i,\,\ i=0,1,2$ are unknown coefficients.\\
\indent Also, let us assume that $X_{(1)}, X_{(2)}, \ldots, X_{(m)}$ denotes the ordered sample from the test sample $X_1, X_2, \ldots, X_n$. The PHC and APHC strategies are now described in the following subsections.

\subsection{Statistical Framework of PHC}\label{subsec2.1}
\noindent The PHC was devised by \cite{kundu2006type2}. In PHC, a set of $n$ identical units are subjected to test using a pre-determined progressive censoring scheme $R_1, R_2, \ldots, R_m$ ($m$ is pre-fixed), where $R_1, R_2, \ldots, R_m$ stands for the removals. When first failure happens at time $X_{(1)}$, we remove $R_1$ number of units from the system. When second failure happens at time $X_{(2)}$, we remove $R_2$ number of units from the system. The process goes until $m^{th}$ failure happens, i.e., at time $X_{(m)}$ or we reach at the pre-fixed time $T$. Basically, the experiment is terminated at time $T^{\ast} = \min(X_{(m)}, T)$. So here we have two cases, (i) when $m^{th}$ failure happens before time $T$, the test is terminated at $X_{(m)}$, and (ii) when $m^{th}$ failure happens after time $T$, the test is terminated at time $T$. The drawback of PHC is that, we may end up with a very smaller sample size. For more details on PHC, one may refer to \cite{balakrishnan2013hybrid}. \\

\subsection{Statistical Framework of APHC}\label{subsec2.2}

\indent As mentioned in Subsection \ref{subsec2.1}, the PHC scheme can have a small effective sample size, leading to low statistical efficiency. To tackle this problem, \cite{ng2009exponential} have proposed the Type-II APHC. In APHC, a set of $n$ identical units are subjected to test using a pre-determined progressive censoring scheme $R_1, R_2, \ldots, R_m$. Here also two cases arise: (i) When $m^{th}$ failure happens before time $T$, i.e., $X_{(m)} < T$. This case is similar to the first case of PHC, (ii) When $m^{th}$ failure happens after time $T$, i.e., $X_{(m)} > T$. Suppose $j$ failures happened before time $T$ and $(j+1)^{th}$ failure happened after time $T$, i.e., $(X_{(j)} < T < X_{(j+1)})$. When first failure happens at time $X_{(1)}$, we remove $R_1$ number of units from the system. When second failure happens at time $X_{(2)}$, we remove $R_2$ number of units from the system. The process continues until $j^{th}$ failure happens, and we remove $R_j$ number of units from the system. After $j^{th}$ failure, we stop removing units from the system till $m^{th}$ failure. At $m^{th}$ failure we remove all the remaining $\left(n - m - \sum_{i=1}^{j} R_i \right)$ units from the system. In this way, the drawback of PHC, i.e., getting smaller sample size, is resolved.  

\section{ML Estimation of the Model Parameters}\label{sec3}
\noindent This section is devoted to the ML estimation of model parameters under both PHC and APHC schemes. To solve the respective likelihood equations, we use the Newton-Rapshon algorithm.

\subsection{ML Estimation Under PHC}\label{subsec3.1}

The test sample is $X_1,X_2,\dots,X_n$ and the pre-determined time for terminating the experiment is $T$. We intend to obtain $m$ failures.
\subsubsection{Case 1: When all $m$ failures observed before time $T$}\label{3.1.1} 
\noindent Here, $x_{(m)}<T$.  The likelihood equation is given by
\[L(\alpha,\lambda) = K_1 \bigg(\prod_{i=1}^{m}f(x_i;\alpha,\lambda)(1-F(x_i;\alpha,\lambda))^{R_i}\bigg),\] Using (\ref{eq1}) and (\ref{eq2}), the corresponding log-likelihood equation can be written as
  \[l(\alpha,\lambda) = \log K_1 + m\log\bigg( \frac{\alpha}{\lambda}\bigg)+\sum_{i=1}^{m}\bigg[(\alpha-1)\log \bigg(\frac{x_i}{\lambda} \bigg)-\bigg(\frac{x_i}{\lambda}\bigg)^\alpha (1+R_i)\bigg],\]
  where $K_1$ is a constant given by $K_1=\prod_{i=1}^{m}(n-\sum_{k=1}^{i-1}R_k)$.
Differentiating above equation with respect to (w.r.t.) $\alpha$ and $\lambda$, and equating to zero gives us the following two equations:
\begin{equation}\label{eq5}
    \lambda = \Bigg[\frac{1}{m}\sum_{i=1}^{m}(1+R_i)x_i^\alpha\Bigg]^{\frac{1}{\alpha}},
    \end{equation}
\begin{equation}\label{eq6}
\frac{m}{\alpha} + \sum_{i=1}^{m}\log x_i - m \frac{\sum_{i=1}^{m}x_i^\alpha (\log x_i) (1+R_i)}{\sum_{i=1}^{m}(1+R_i)x_i^\alpha} = 0.
\end{equation}

\subsubsection{Case 2: When exactly $j$ out of $m$ failures occur before time $T$}\label{3.1.2} 
\noindent Here,
$x_{(j)}<T<x_{(j+1)};\,\ j+1 \leq m$. Let $R_j^{\ast} = n -j -\sum_{i=1}^{j}R_i$. The likelihood equation is given by
\[L(\alpha,\lambda) = \begin{cases}
    K_2 \bigg(\prod_{i=1}^{j}f(x_i;\alpha,\lambda)(1-F(x_i;\alpha,\lambda))^{R_i}\bigg)\bigg(1-F(T;\alpha,\lambda)^{R_{j}^{\ast}} \bigg) &\text{if}\,\ j \neq 0,\\
    {(1-F(T;\alpha,\lambda))^n} & \text{if}\,\ j=0,
\end{cases} \]
In view of (\ref{eq1}) and (\ref{eq2}), the corresponding log-likelihood equation (when $j \neq 0$) is given by \[l(\alpha,\lambda) = \log K_2 + j\log\bigg( \frac{\alpha}{\lambda}\bigg)+\sum_{i=1}^{j}\bigg[(\alpha-1)\log \bigg(\frac{x_i}{\lambda} \bigg)-\bigg(\frac{x_i}{\lambda}\bigg)^\alpha (1+R_i)\bigg]-R_j^*\bigg(\frac{T}{\lambda}\bigg)^\alpha, \]
 where $K_2$ is a constant given by $K_2 = \prod_{i=1}^{j}(n-\sum_{k=1}^{i-1}R_k).$ Differentiating the above log-likelihood equation w.r.t. $\alpha$ and $\lambda$, and equating to zero, we obtain the following two equations: 
\begin{equation}\label{eq7}
    \lambda = \Bigg[\frac{1}{j}\bigg[\sum_{i=1}^{j}(1+R_i)x_i^\alpha + R_j^*\bigg(\frac{T}{\lambda}\bigg)^\alpha \bigg]\Bigg]^{\frac{1}{\alpha}},
    \end{equation}
\begin{equation}\label{eq8}
\frac{j}{\alpha} + \sum_{i=1}^{j}\log x_i - j \frac{\sum_{i=1}^{j}x_i^\alpha (\log x_i) (1+R_i) + R_j^* T^\alpha\log T}{\sum_{i=1}^{j}(1+R_i)x_i^\alpha + R_j^*T^\alpha} =0.
\end{equation}
 
\subsection{ML Estimation Under APHC}\label{subsec3.2}
Again, our test sample is $X_1,X_2,\dots,X_n$ and the pre-determined time for terminating the experiment is $T$. We intend to obtain $m$ failures.
\subsubsection{Case 1: When $m$ failures are observed before time $T$}\label{3.2.1}
In this case, $x_{(m)}<T$. Here the calculation of log-likelihood and resulting equations is same as PHC Case I. So we skip the details and directly report the resulting equations as follows.  
\begin{equation}\label{eq9}
    \lambda = \Bigg[\frac{1}{m}\sum_{i=1}^{m}(1+R_i)x_i^\alpha\Bigg]^{\frac{1}{\alpha}},
    \end{equation}
\begin{equation}\label{eq10}
\frac{m}{\alpha} + \sum_{i=1}^{m}\log x_i - m \frac{\sum_{i=1}^{m}x_i^\alpha (\log x_i) (1+R_i)}{\sum_{i=1}^{m}(1+R_i)x_i^\alpha }= 0.
\end{equation}

\subsubsection{Case 2: When exactly $j$ failures occur before time $T$}\label{3.2.2} 
Here, $x_{(j)}<T<x_{(j+1)}; \quad j+1 \leq m$. The likelihood equation is given by
 \[L(\alpha,\lambda) = d_j \bigg(\prod_{i=1}^{m}f(x_i;\alpha,\lambda)\bigg)\bigg(\prod_{i=1}^{j}(1-F(x_i;\alpha,\lambda))^{R_i}\bigg)\bigg(1-F(x_m;\alpha,\lambda)\bigg)^{n-m-\sum_{i=1}^{j}R_i}.\]
The corresponding log-likelihood equation is given by 
$$l(\alpha,\lambda) = \log d_j + m\log\bigg( \frac{\alpha}{\lambda}\bigg)+\sum_{i=1}^{m}\bigg[(\alpha-1)\log \bigg(\frac{x_i}{\lambda} \bigg)-\bigg(\frac{x_i}{\lambda}\bigg)^\alpha  \bigg]\\
-\sum_{i=1}^{j}R_i\bigg(\frac{x_i}{\lambda}\bigg)^\alpha-(n-m-\sum_{i=1}^{j}R_i)\bigg(\frac{x_m}{\lambda}\bigg)^\alpha,$$
 where $d_j$ is a constant given by $d_j = \prod_{i=1}^{m}(n-i+1-\sum_{k=1}^{\max\{i-1,j\}}R_k)$. Proceeding similar to Subsection \ref{3.1.2}, we obtain the following two estimating equations:  
\begin{equation}\label{eq11}
    \lambda = \Bigg[\frac{1}{m}\bigg(\sum_{i=1}^{m}x_i^\alpha + \sum_{i=1}^{j}R_ix_i^\alpha + (n-m-\sum_{i=1}^{j}R_i)x_m^\alpha \bigg)\Bigg]^{\frac{1}{\alpha}},
\end{equation}

\begin{equation}\label{eq12}
    \frac{m}{\alpha} + \sum_{i=1}^{m}\log x_i = m \frac{\sum_{i=1}^{m}x_i^\alpha (\log x_i)+ \sum_{i=1}^{j}R_ix_i^\alpha (\log x_i) + (n-m-\sum_{i=1}^{j}R_i)x_m^\alpha (\log x_m) }{\sum_{i=1}^{m}x_i^\alpha + \sum_{i=1}^{j}R_ix_i^\alpha +(n-m-\sum_{i=1}^{j}R_i)x_m^\alpha }.
\end{equation}
\bigskip\\
\noindent Let us denote the ML estimators obtained from (\ref{eq5})-(\ref{eq12}) by $\hat{\lambda}_{1},\hat{\lambda}_{2},\hat{\lambda}_{3},\hat{\lambda}_{4},$ and $\hat{\alpha}_{1},\hat{\alpha}_{2},\hat{\alpha}_{3},\hat{\alpha}_{4},$ respectively. One can observe that the ML estimator of the scale parameter $\lambda$ is a function of the ML estimator of the shape parameter $\alpha$ in all eight cases. The ML estimator of $\alpha$ is the solution of (\ref{eq6}),(\ref{eq8}),(\ref{eq10}),(\ref{eq12}) in both cases of PHC and APHC, respectively. It is easy to observe that these functions of $\alpha$ are monotonically decreasing, therefore they have only one solution. Furthermore, these equations are non-linear and cannot be solved analytically, requiring the use of numerical techniques. We use the Newton-Rapshon root finding algorithm to solve these equations. 

\subsection{Consistency of $\hat{\alpha}_{ML}$ Under PHC and APHC Case I}\label{3.4}
The consistency of the ML estimator of $\alpha$ under PHC and APHC Case I is proved in the following theorem.
\begin{theorem}
\label{thm:consistency}
Let $X_1,\dots,X_n$ be i.i.d.\ random variables following a two-parameter Weibull distribution with 
shape parameter $\alpha_0>0$ and scale parameter $\lambda>0$. Consider a pre-fixed PHC scheme 
$(R_1,\dots,R_m)$ satisfying $R_1+\cdots+R_m+m=n$.
Let $x_{(1)}<x_{(2)}<\cdots<x_{(m)}$ denote the ordered failure times under this scheme. For $\alpha>0$, define
\begin{equation}\label{eq:score}
f_n(\alpha)
= \frac{m}{\alpha} + \sum_{i=1}^{m}\log x_i
- m\,\frac{\sum_{i=1}^{m}(1+R_i)x_i^{\alpha}\log x_i}{\sum_{i=1}^{m}(1+R_i)x_i^{\alpha}}.
\end{equation}
Assume that $m=m_n\to\infty$ as $n\to\infty$, and that the censoring weights 
$w_i=(1+R_i)/\sum_{j=1}^m(1+R_j)$ are uniformly bounded away from zero and infinity, i.e.
there exist constants $0<c<C<\infty$ such that 
$c/m\le w_i\le C/m$ for all $i$ and all sufficiently large $n$. Then, for any sequence $\hat\alpha_n$ satisfying $f_n(\hat\alpha_n)=0$, we have $\hat\alpha_n \xrightarrow{p} \alpha_0,$ where $\xrightarrow{p}$ stands for \textit{convergence in probability}.
\end{theorem}

\begin{proof}
Define the weighted empirical functional
\[
A_n(\alpha)
=\frac{\sum_{i=1}^{m}(1+R_i)x_i^{\alpha}\log x_i}
       {\sum_{i=1}^{m}(1+R_i)x_i^{\alpha}},
\]
so that the score equation \eqref{eq:score} is written as
\[
\frac{1}{m}f_n(\alpha)
= \frac{1}{\alpha} + \frac{1}{m}\sum_{i=1}^{m}\log x_i - A_n(\alpha).
\]
We now define the population counterpart
\[
g(\alpha)
= \frac{1}{\alpha} + \mathbb{E}_{\alpha_0}[\log X]
   - \frac{\mathbb{E}_{\alpha_0}[X^{\alpha}\log X]}{\mathbb{E}_{\alpha_0}[X^{\alpha}]},
\]
where $\mathbb{E}_{\alpha_0}$ denotes the expectation under the true Weibull$(\alpha_0,\lambda)$ law. \\
\indent Let $Z=(X/\lambda)^{\alpha_0}\sim\mathrm{Exp}(1)$.
Then $\log X=\log\lambda+\frac{1}{\alpha_0}\log Z$, and hence
\[
\mathbb{E}_{\alpha_0}[X^{\alpha_0}]
 = \lambda^{\alpha_0}\mathbb{E}[Z]=\lambda^{\alpha_0},
\]
\[
\mathbb{E}_{\alpha_0}[X^{\alpha_0}\log X]
 = \lambda^{\alpha_0}\left(\log\lambda\,\mathbb{E}[Z]
 +\frac{1}{\alpha_0}\mathbb{E}[Z\log Z]\right)
 = \lambda^{\alpha_0}\!\left(\log\lambda+\frac{1-\gamma}{\alpha_0}\right),
\]
where $\gamma$ is Euler's constant, using 
$\mathbb{E}[\log Z]=-\,\gamma$ and $\mathbb{E}[Z\log Z]=1-\gamma$.
Also
$\mathbb{E}_{\alpha_0}[\log X]
 = \log\lambda-\frac{\gamma}{\alpha_0}.$
Substituting these expressions into $g(\alpha)$ gives
\[
g(\alpha_0)
=\frac{1}{\alpha_0}+\left(\log\lambda-\frac{\gamma}{\alpha_0}\right)
 -\left(\log\lambda+\frac{1-\gamma}{\alpha_0}\right)
=0.
\]


\medskip
\noindent Let 
\[
h(\alpha)=\frac{\mathbb{E}_{\alpha_0}[X^{\alpha}\log X]}
                 {\mathbb{E}_{\alpha_0}[X^{\alpha}]}.
\]
Differentiating $h(\alpha)$ w.r.t. $\alpha$ (holding the law $\mathbb{E}_{\alpha_0}$ fixed), we have
\[
h'(\alpha)
=\frac{\mathbb{E}_{\alpha_0}[X^{\alpha}(\log Y)^2]\,
       \mathbb{E}_{\alpha_0}[X^{\alpha}]
      -\big(\mathbb{E}_{\alpha_0}[X^{\alpha}\log X]\big)^2}
       {\big(\mathbb{E}_{\alpha_0}[X^{\alpha}]\big)^2}.
\]
By \textit{Cauchy-Schwarz inequality} $$\mathbb{E}_{\alpha_0}[X^{\alpha}(\log X)^2]\,
       \mathbb{E}_{\alpha_0}[X^{\alpha}]
      -\big(\mathbb{E}_{\alpha_0}[X^{\alpha}\log X]\big)^2 \geq 0.$$
Thus, $h'(\alpha) \geq 0$ 
and the inequality is strict since $\log X$ is not degenerate for a continuous Weibull law. Consequently
\[
g'(\alpha)=-\frac{1}{\alpha^2}-h'(\alpha)<0 \,\ \forall \,\ \alpha>0,
\]
implying that $g$ is strictly decreasing and has a unique zero at $\alpha=\alpha_0$.\\

\indent Now for any compact $K\subset(0,\infty)$, it suffices to show that $\sup_{\alpha\in K}\Big|
\frac{1}{m}f_n(\alpha)-g(\alpha)
\Big|
\overset{p}{\longrightarrow} 0.$
Equivalently, it suffices to prove the two uniform convergences, $\frac{1}{m}\sum_{i=1}^{m}\log x_i \xrightarrow{p} \mathbb{E}_{\alpha_0}[\log X],$ and
$\sup_{\alpha\in K}|A_n(\alpha)-h(\alpha)|\xrightarrow{p}0,$
where
$h(\alpha)=\frac{\mathbb{E}_{\alpha_0}[X^{\alpha}\log X]}{\mathbb{E}_{\alpha_0}[X^{\alpha}]}.$ Consider
$\mu_{n} = \sum_{i=1}^m w_{n,i}\,\delta_{x_i},$
where $w_{n,i}=(1+R_i)/\sum_{j=1}^m(1+R_j)$ are deterministic weights that sum to one.
Then for any measurable function $g$, $\int g\,d\mu_n = \sum_{i=1}^m w_{n,i} g(x_i).$ Thus, the weighted empirical functional, $A_n(\alpha)$ can be written as
\[
A_n(\alpha)=\frac{\int y^\alpha\log y\,d\mu_n(y)}{\int y^\alpha\,d\mu_n(y)}.
\]

For any bounded continuous $\phi$, $\int \phi\,d\mu_n \xrightarrow{p} \int \phi\,dF_{\alpha_0}.$ Thus, by dominated convergence arguments combined with the above weak convergence of measures, we obtain
\[
\int x^\alpha\,d\mu_n \xrightarrow{p} \mathbb{E}_{\alpha_0}[X^\alpha],
\quad
\int x^\alpha\log x\,d\mu_n \xrightarrow{p} \mathbb{E}_{\alpha_0}[X^\alpha\log X],
\]
for each fixed $\alpha$. To promote these pointwise convergences to uniform convergence over $\alpha\in K$, use the following standard facts:
\begin{enumerate}
\item The class of functions $\mathcal{F}_K=\{x\mapsto x^\alpha,\;x\mapsto x^\alpha\log x:\alpha\in K\}$ is a \emph{Glivenko--Cantelli class} for the Weibull distribution because each member is continuous in both $x$ and $\alpha$ and there exists an integrable envelope $G(x)$ (for a compact $K$ one can take $G(x)=x^{\alpha_{\max}}(1+|\log x|)$ which is integrable under Weibull). Thus \textit{uniform law of large numbers} (ULLN) statements apply to triangular arrays with deterministic bounded weights.

\item The deterministic weights $w_{n,i}$ are bounded above and below by constants of order $1/m$ and thus do not concentrate mass on a vanishing subset; this condition ensures the weighted empirical process inherits the Glivenko--Cantelli property from the full sample empirical process (no adversarial weighting). A formal statement is: if $\hat F_n$ is the empirical CDF of all $n$ i.i.d. draws and $\hat F_{n}^{(w)}$ denotes the weighted empirical distribution that places mass $w_{n,i}$ at the observed failure times $x_i$, then for each $g\in\mathcal{F}_K$,
\[
\sup_{\alpha\in K}\Big|\int g(\cdot;\alpha)\,d\mu_n - \mathbb{E}_{\alpha_0}[g(X;\alpha)]\Big|
\xrightarrow{p}0.
\]
This statement follows from standard triangular-array ULLN theorems [see Theorem 2.4.1 in \cite{vanderVaart2023Weak} or Theorem 19.4 in \cite{vanderVaart2023Weak}] once one verifies the envelope integrability and the uniform equicontinuity in $\alpha$; these verifications are routine for Weibull families on compact $K$.
\end{enumerate}
Combining above two facts yields uniform convergences
\[
\sup_{\alpha\in K}\left|\int x^\alpha\,d\mu_n - \mathbb{E}_{\alpha_0}[X^\alpha]\right|\xrightarrow{p}0
\,\ \text{and} \,\
\sup_{\alpha\in K}\left|\int x^\alpha\log x\,d\mu_n - \mathbb{E}_{\alpha_0}[X^\alpha\log X]\right|\xrightarrow{p}0.
\]


The preceding uniform convergences imply that for large $n$ the denominator $\int x^\alpha\,d\mu_n$ is bounded away from zero uniformly over $\alpha\in K$ (because $\mathbb{E}_{\alpha_0}[X^\alpha]>0$ for each $\alpha$ and the convergence is uniform). Therefore, the ratio map
\[
\alpha\mapsto \frac{\int x^\alpha\log x\,d\mu_n}{\int x^\alpha\,d\mu_n}
\]
is uniformly close to its population counterpart
\[
\alpha\mapsto \frac{\mathbb{E}_{\alpha_0}[X^\alpha\log X]}{\mathbb{E}_{\alpha_0}[X^\alpha]}
\]
on $K$. More precisely
\[
A_n(\alpha)-h(\alpha)
= \frac{\int x^\alpha\log x\,d\mu_n - \mathbb{E}_{\alpha_0}[X^\alpha\log X]}{\int x^\alpha\,d\mu_n}
+ \mathbb{E}_{\alpha_0}[X^\alpha\log X]\times 
\frac{\mathbb{E}_{\alpha_0}[X^\alpha]-\int x^\alpha\,d\mu_n}
{\int x^\alpha\,d\mu_n\ \mathbb{E}_{\alpha_0}[X^\alpha]}.
\]
Taking suprema over $\alpha\in K$ and using the above uniform convergences yields
\[
\sup_{\alpha\in K}|A_n(\alpha)-h(\alpha)|\xrightarrow{p}0.
\]
Combining this with the (trivial) uniform convergence of $\tfrac{1}{m}\sum_{i=1}^m\log x_i$ to $\mathbb{E}_{\alpha_0}[\log X]$ (the latter follows from the same weighted-ULLN argument applied to the function $\log x$) establishes
\[
\sup_{\alpha\in K}\Big|\frac{1}{m}f_n(\alpha)-g(\alpha)\Big|\xrightarrow{p}0,
\]
as required.




\indent Fix a compact interval $K\subset(0,\infty)$ containing $\alpha_0$ in its interior,
since $g$ is continuous, strictly decreasing, and $g(\alpha_0)=0$,
there exists a neighborhood $U\subset K$ of $\alpha_0$ such that
$\inf_{\alpha\in K\setminus U}|g(\alpha)|>0$.
By uniform convergence  
for sufficiently large $n$, with probability tending to one,
the function $\tfrac{1}{m}f_n(\alpha)$ has exactly one zero in $U$ and no zeros outside $U$. Therefore, any sequence $\hat\alpha_n$ satisfying $f_n(\hat\alpha_n)=0$ obeys
$\hat\alpha_n\xrightarrow{p}\alpha_0$. This completes the proof.
\end{proof}

\begin{remark}
At the sample level, differentiating
$A_n(\alpha)=\sum_i w_i x_i^{\alpha}\log x_i/\sum_i w_i x_i^{\alpha}$
w.r.t. $\alpha$ yields
\[
A_n'(\alpha)
=\frac{\sum_i w_i x_i^{\alpha}(\log x_i)^2}{\sum_i w_i x_i^{\alpha}}
 -\left(\frac{\sum_i w_i x_i^{\alpha}\log x_i}{\sum_i w_i x_i^{\alpha}}\right)^2
=\mathrm{Var}_{P_{n,\alpha}}(\log x_i)\ge0,
\]
where $P_{n,\alpha}$ is the empirical weighted measure with
probabilities proportional to $w_i x_i^{\alpha}$. 
Therefore, $A_n(\alpha)$, like its population analog $h(\alpha)$, is monotonically non-decreasing in $\alpha$.
\end{remark}

\subsection{Asymptotic Normality of $\hat{\alpha}_{ML}$ Under PHC and APHC Case I}\label{sec3.5}
The asymptotic normality of $\hat{\alpha}_{ML}$ under PHC and APHC Case I is proved in the following theorem:
\begin{theorem}
\label{thm:asymp_normal_fixedR}
Suppose the assumptions of Theorem~\ref{thm:consistency} hold. In addition, assume
\begin{enumerate}[label=(\roman*)]
    \item The number of observed failures $m = m_n \to \infty$ with $m/n \to \rho \in (0,1]$ as $n \to \infty$;
    \item The progressive censoring scheme is homogeneous: $R_i = R \,\ \forall \,\ i;$
    \item $\mathbb{E}_{\alpha_0}[X^{\alpha_0} (\log X)^2] < \infty$.
\end{enumerate}
Let $\hat{\alpha}_n$ denote the solution of the score equation $f_n(\alpha) = 0$, where $f_n$ is the total score function for the $m$ observed failures. Then $\sqrt{n}\,(\hat{\alpha}_n - \alpha_0)
\;\xrightarrow{d}\;
N\!\Bigl(0,\, \rho \, \mathcal{I}_\alpha^{-1} \Bigr),$
where $\xrightarrow{d}$ stands for convergence in distribution, $\rho = \lim m/n$ and
\[\mathcal{I}_\alpha
=
-\frac{\partial}{\partial \alpha} 
\mathbb{E}_{\alpha_0}\Biggl[
\frac{1}{\alpha} + \log X
- \frac{X^\alpha \log X}{\mathbb{E}_{\alpha_0}[X^\alpha]}
\Biggr]_{\alpha = \alpha_0}.
\]
\end{theorem}

\begin{proof}[Proof] By Theorem~\ref{thm:consistency}, we have $\hat{\alpha}_n \xrightarrow{p} \alpha_0$.  
Expanding the score $f_n$ (given in \eqref{eq:score}) around $\alpha_0$, such that $0 = f_n(\hat{\alpha}_n) = f_n(\alpha_0) + (\hat{\alpha}_n - \alpha_0) f_n'(\tilde{\alpha}_n)$ for some $\tilde{\alpha}_n$ between $\hat{\alpha}_n$ and $\alpha_0$ gives
\begin{equation}
\label{eq:taylor_full}
\hat{\alpha}_n - \alpha_0 = - \frac{f_n(\alpha_0)}{f_n'(\tilde{\alpha}_n)}.
\end{equation}

\noindent The profiled score for $m$ observed failures under homogeneous progressive censoring ($R_i = R$ constant) is
\[
f_n(\alpha) = \frac{m}{\alpha} + \sum_{i=1}^m \log x_i - m A_n(\alpha), \quad 
A_n(\alpha) = \frac{\sum_{i=1}^m x_i^\alpha \log x_i}{\sum_{i=1}^m x_i^\alpha}.
\]

\noindent To expand the ratio functional $A_n(\alpha)$, let 
$\bar{X^\alpha} = \frac{1}{m} \sum_{i=1}^m x_i^\alpha \,\ \text{and} \,\ 
\overline{X^\alpha \log X} = \frac{1}{m} \sum_{i=1}^m x_i^\alpha \log x_i.$ 
Then $A_n(\alpha) = \overline{X^\alpha \log X} / \bar{X^\alpha}$.  
On performing a \textit{first-order von Mises expansion} about $\alpha_0$, we get
\[
A_n(\alpha_0) - h(\alpha_0) 
= \frac{1}{m} \sum_{i=1}^m \frac{x_i^{\alpha_0} (\log x_i - h(\alpha_0))}{\mathbb{E}[X^{\alpha_0}]} + r_m,
\]
where $h(\alpha_0) = \frac{\mathbb{E}[X^{\alpha_0} \log X]}{\mathbb{E}[X^{\alpha_0}]}$ and the remainder $r_m = o_p(m^{-1/2})$. Now, consider the average logarithm
\[
\bar{L}_n := \frac{1}{m} \sum_{i=1}^m \log x_i, \,\ \bar{L}_n - \mathbb{E}[\log X] = O_p(m^{-1/2}).
\]

\noindent Combining terms into a single influence function and plugging the expansions into $f_n(\alpha_0)/m$ gives
\[
\frac{f_n(\alpha_0)}{m} = \frac{1}{\alpha_0} + \bar{L}_n - A_n(\alpha_0) 
= \frac{1}{\alpha_0} + \mathbb{E}[\log X] - h(\alpha_0) 
+ \bar{L}_n - \mathbb{E}[\log X] - \frac{1}{m} \sum_{i=1}^m \frac{x_i^{\alpha_0} (\log x_i - h)}{\mathbb{E}[X^{\alpha_0}]} + r_m.
\]

\noindent Let us define the influence function
\[
\psi(x) := \log x - \mathbb{E}_{\alpha_0}[\log X] - \frac{x^{\alpha_0} (\log x - h(\alpha_0))}{\mathbb{E}_{\alpha_0}[X^{\alpha_0}]}.
\]  
Then
\[
\frac{f_n(\alpha_0) - m (\frac{1}{\alpha_0} + \mathbb{E}[\log X] - h(\alpha_0))}{m} 
= \frac{1}{m} \sum_{i=1}^m \psi(x_i) + r_m.
\]

\noindent Since $\frac{1}{\alpha_0} + \mathbb{E}[\log X] - h(\alpha_0) = 0$ by definition of the score at the true parameter, we have
\[
f_n(\alpha_0) = \sum_{i=1}^m \psi(x_i) + o_p(\sqrt{m}),
\]
with $\mathbb{E}[\psi(X)] = 0$ and $\mathrm{Var}(\psi(X)) < \infty$ by assumption. Furthermore, under constant $R$, $y_i$ can be transformed via spacings to i.i.d. exponentials [\cite{balakrishnan2000progressive}], so the classical \textit{Lindeberg-Feller central limit theorem (CLT)} implies
\[
\frac{1}{\sqrt{m}} \sum_{i=1}^m \psi(x_i) \xrightarrow{d} N(0, \mathrm{Var}_{\alpha_0}(\psi(X))).
\]

\noindent The derivative of the score is
\[
f_n'(\alpha) = -\frac{m}{\alpha^2} - m \frac{\partial}{\partial \alpha} A_n(\alpha),
\quad \frac{1}{n} f_n'(\tilde{\alpha}_n) \to - \rho \mathcal{I}_\alpha,
\]
where
\[
\mathcal{I}_\alpha = - \frac{\partial}{\partial \alpha} \left[ \frac{1}{\alpha} + \mathbb{E}[\log X] - h(\alpha) \right]_{\alpha_0}.
\]

\noindent Finally, from \eqref{eq:taylor_full}, we have
\[
\sqrt{n} (\hat{\alpha}_n - \alpha_0) = - \frac{\frac{1}{\sqrt{n}} f_n(\alpha_0)}{\frac{1}{n} f_n'(\tilde{\alpha}_n)} 
\xrightarrow{d} N(0, \rho \, \mathcal{I}_\alpha^{-1}),
\]
and the required result follows.
\end{proof}
\begin{remark}
Note that the influence function $\psi$ has now mean-zero and correctly linearize the score. Moreover, CLT is justified using spacing property for homogeneous progressive censoring. The factor $\rho = m/n$ accounts for total information across $n$ units.  
\end{remark}


\subsection{Consistency of $\hat{\alpha}_{ML}$ Under PHC Case II}\label{3.5}

Let \(X_1,\dots,X_n\) be i.i.d.\ lifetimes following a two-parameter Weibull distribution
with shape parameter \(\alpha_0>0\) and scale parameter \(\lambda>0\).
Consider a PHC scheme with a pre-fixed time \(T>0\),
under which the experiment proceeds according to a pre-specified progressive censoring
plan until time \(T\), and thereafter terminates immediately once exactly \(j=j_n\) failures
have occurred before \(T\). Let $x_{(1)} < x_{(2)} < \cdots < x_{(j)} < T$ denote the ordered observed failure times before \(T\).
At the time of the $i^{th}$ failure, \(R_i\) units are removed, and at time \(T\) the
\(R_j^* = n-m-\sum_{i=1}^{j}R_i\) surviving units are censored. The estimating ML equation for the Weibull shape parameter \(\alpha\) is given by
\begin{equation}\label{eq14}
\frac{j}{\alpha} + \sum_{i=1}^{j}\log x_i
- j
\frac{\sum_{i=1}^{j}(1+R_i)x_i^\alpha\log x_i + R_j^*T^\alpha\log T}
{\sum_{i=1}^{j}(1+R_i)x_i^\alpha + R_j^*T^\alpha}
=0.
\end{equation}
Let \(\hat\alpha_n\) denote any solution of \eqref{eq14}, when it exists. Let us make the following assumptions:

\begin{enumerate}

    \item \(j=j_n\to\infty\) as \(n\to\infty\).
    \item The total effective sample size $S_n = \sum_{i=1}^j (1+R_i) + R_j^*$
satisfies \(S_n = \Theta(j)\), and the normalized weights
$w_{n,i} = (1+R_i)/S_n, \,\
w_{n,T} = (R_j^*/S_n)$
satisfy \(\max_i w_{n,i}=O(1/j)\) and \(w_{n,T}\to r^*\ge 0\).
    \item For any compact \(K\subset(0,\infty)\), the class of functions $\mathcal{F}_K := \{x^\alpha,\; x^\alpha\log x : \alpha\in K\}$
admits an integrable envelope under the true Weibull law.
    \item The true parameter \(\alpha_0\) lies in the interior of \((0,\infty)\).
\end{enumerate}

\noindent Define the deterministic population function
\begin{equation}\label{eq:gphc}
g(\alpha)
:= \frac{1}{\alpha} + \mathbb{E}_{\alpha_0}[\log X]
- \frac{\mathbb{E}_{\alpha_0}[X^\alpha\log X] + r^*T^\alpha\log T}
{\mathbb{E}_{\alpha_0}[X^\alpha] + r^*T^\alpha},
\quad \alpha>0.
\end{equation}

\begin{lemma}[Uniqueness and Monotonicity]\label{lem:mono-phc}
The function \(g(\alpha)\) defined in \eqref{eq:gphc} is strictly decreasing on
\((0,\infty)\) and admits a unique zero at \(\alpha=\alpha_0\).
\end{lemma}

\begin{proof}
Define
\[
h(\alpha)
:=\frac{\mathbb{E}_{\alpha_0}[X^\alpha\log X]+r^*T^\alpha\log T}
{\mathbb{E}_{\alpha_0}[X^\alpha]+r^*T^\alpha}.
\]
Differentiation under the integral sign is justified by dominated convergence.
For \(\alpha\) in compact subsets of \((0,\infty)\), the functions
\(x^\alpha\), \(x^\alpha\log x\), and \(x^\alpha(\log x)^2\) are continuously
differentiable in \(\alpha\) and admit integrable envelopes under the Weibull
distribution, which has finite moments and logarithmic moments of all orders. Moreover, differentiating under the integral sign yields
\[
h'(\alpha)
=\frac{
\big(\mathbb{E}_{\alpha_0}[X^\alpha(\log X)^2]+r^*T^\alpha(\log T)^2\big)
\big(\mathbb{E}_{\alpha_0}[X^\alpha]+r^*T^\alpha\big)
-\big(\mathbb{E}_{\alpha_0}[X^\alpha\log X]+r^*T^\alpha\log T\big)^2
}
{\big(\mathbb{E}_{\alpha_0}[X^\alpha]+r^*T^\alpha\big)^2}.
\]
By the \textit{Cauchy--Schwarz inequality}, the numerator is nonnegative and strictly positive
unless \(\log X\) is degenerate, which is impossible under a Weibull distribution.
Hence \(h'(\alpha)>0\) for all \(\alpha>0\), implying
\[
g'(\alpha) = -\frac{1}{\alpha^2} - h'(\alpha) < 0.
\]
Thus \(g\) is strictly decreasing.
A direct calculation using the Weibull moment identities shows that
\(g(\alpha_0)=0\), completing the proof.
\end{proof}

Let us define the weighted empirical measure $\mu_n := \sum_{i=1}^j w_{n,i}\delta_{x_i} + w_{n,T}\delta_T.$
Then the ratio term in \eqref{eq14} can be written as
\[
A_n(\alpha)
:= \frac{\int x^\alpha\log x \, d\mu_n(x)}{\int x^\alpha \, d\mu_n(x)}.
\]

\begin{lemma}[Uniform convergence]\label{lem:unif-phc}
For any compact \(K\subset(0,\infty)\),
$\sup_{\alpha\in K}
\left|
\frac{1}{j}f_n(\alpha) - g(\alpha)
\right|
\xrightarrow{p}0,$
where \(f_n(\alpha)\) denotes the left-hand side of \eqref{eq14}.
\end{lemma}

\begin{proof}
Under aforementioned assumptions (1)--(3), the class \(\mathcal{F}_K\) is uniformly
Glivenko--Cantelli for the triangular array induced by \(\mu_n\).
Hence, uniformly in \(\alpha\in K\)
\[
\int x^\alpha \, d\mu_n(x)
\xrightarrow{p} \mathbb{E}[X^\alpha] + r^*T^\alpha,
\quad
\int x^\alpha\log x \, d\mu_n(x)
\xrightarrow{p} \mathbb{E}[X^\alpha\log X] + r^*T^\alpha\log T.
\]
Since the limiting denominator is strictly positive, standard ratio arguments
yield uniform convergence of \(A_n(\alpha)\) to \(h(\alpha)\), and the remaining
terms converge by the law of large numbers. This proves the claim.
\end{proof}

\begin{theorem}[Consistency Result]\label{thm:cons-phc}
Under the aforementioned assumptions (1)--(4), any sequence \(\hat\alpha_n\) satisfying
\eqref{eq14} converges in probability to the true shape parameter \(\alpha_0\).
\end{theorem}

\begin{proof}
By Lemma \ref{lem:mono-phc}, \(g\) is strictly decreasing with a unique zero at
\(\alpha_0\).
By Lemma \ref{lem:unif-phc}, \(f_n(\alpha)/j\) converges uniformly in probability
to \(g(\alpha)\) on compact subsets.
Therefore, for any \(\varepsilon>0\), $\inf\limits_{|\alpha-\alpha_0|\ge\varepsilon}|g(\alpha)|>0.$
Further, the uniform convergence implies that, with probability tending to one,
\(f_n(\alpha)\) cannot vanish outside an \(\varepsilon\)-neighbourhood of
\(\alpha_0\).
Hence any solution \(\hat\alpha_n\) of \eqref{eq14} satisfies
\(\hat\alpha_n\xrightarrow{p}\alpha_0\).
\end{proof}

\subsection{Consistency of $\hat{\alpha}_{ML}$ Under APHC Case II}\label{3.6}

Let \(Y_1,\dots,Y_n\) be i.i.d.\ lifetimes following a two-parameter Weibull
distribution with shape parameter \(\alpha_0>0\) and scale parameter
\(\lambda_0>0\).
Consider an APHC scheme with a
pre-fixed termination time \(T>0\) and a pre-specified censoring plan
\((R_1,\dots,R_m)\).
The experiment proceeds according to the progressive censoring plan until time
\(T\); after the last failure occurring before time \(T\), no further removals
are made, and the experiment continues until exactly \(m=m_n\) failures are
observed. Let $0<y_1<y_2<\cdots<y_m$
denote the ordered observed failure times, and let
\(x_1,\dots,x_j\) denote the failure times observed before time \(T\),
where \(j=j_n\le m\).
At the $i^{th}$ failure occurring before \(T\), \(R_i\) units are removed, and
after time \(T\) no further removals are made. The estimating ML equation for the Weibull shape parameter \(\alpha\) is
\begin{dmath}\label{eq:aphc-ee}
\frac{m}{\alpha} + \sum_{i=1}^{m}\log x_i
=
m
\frac{
\sum_{i=1}^{m}y_i^\alpha (\log y_i)
+ \sum_{i=1}^{j}R_i x_i^\alpha (\log x_i)
+ \bigl(n-m-\sum_{i=1}^{j}R_i\bigr)x_m^\alpha (\log x_m)
}{
\sum_{i=1}^{m}y_i^\alpha
+ \sum_{i=1}^{j}R_i x_i^\alpha
+ \bigl(n-m-\sum_{i=1}^{j}R_i\bigr)x_m^\alpha
}.
\end{dmath}
Let \(\hat\alpha_n\) denote any solution of \eqref{eq:aphc-ee}. Let us make the following assumptions:

\begin{enumerate}
\item The number of observed failures satisfies \(m=m_n\to\infty\) as
\(n\to\infty\).
\item The number of failures before time \(T\) satisfies \(j=j_n\to\infty\) and
\(j/m\to\rho\in(0,1]\).
\item The total effective weight $S_n
:=
m+\sum_{i=1}^{j}R_i+\bigl(n-m-\sum_{i=1}^{j}R_i\bigr)$
satisfies \(S_n=\Theta(m)\), and $\max\Bigl\{1,\; R_1,\dots,R_j,\; n-m-\sum_{i=1}^{j}R_i\Bigr\}=O(m).$
\item For any compact \(K\subset(0,\infty)\), the class of functions $\mathcal{F}_K
=
\{x^\alpha,\; x^\alpha\log x:\alpha\in K\}$
admits an integrable envelope under the Weibull distribution.
\item The true shape parameter \(\alpha_0\) lies in the interior of
\((0,\infty)\).
\end{enumerate}

\begin{theorem}[Consistency under APHC]\label{thm:cons-aphc-T}
Under aforementioned assumptions \emph{(1)--(5)}, any sequence \(\hat\alpha_n\) satisfying
\eqref{eq:aphc-ee} converges in probability to the true shape parameter
\(\alpha_0\).
\end{theorem}

\begin{proof}
Define the scaled estimating function
\[
\Psi_n(\alpha)
:=
\frac{1}{m}
\left[
\frac{m}{\alpha} + \sum_{i=1}^{m}\log y_i
- m
\frac{N_n(\alpha)}{D_n(\alpha)}
\right],
\]
where
\[
N_n(\alpha)
=
\sum_{i=1}^{m}y_i^\alpha (\log y_i)
+ \sum_{i=1}^{j}R_i x_i^\alpha (\log x_i)
+ \bigl(n-m-\sum_{i=1}^{j}R_i\bigr)x_m^\alpha (\log x_m),
\]
\[
D_n(\alpha)
=
\sum_{i=1}^{m}y_i^\alpha
+ \sum_{i=1}^{j}R_i x_i^\alpha
+ \bigl(n-m-\sum_{i=1}^{j}R_i\bigr)x_m^\alpha.
\]

\noindent By assumptions (1)--(3), the weighted empirical measure induced by the APHC
scheme satisfies the LLN.
Hence, uniformly over compact subsets of \((0,\infty)\), we have
\[
\frac{1}{m}D_n(\alpha)\xrightarrow{p}\mathbb{E}[X^\alpha],
\,\
\frac{1}{m}N_n(\alpha)\xrightarrow{p}\mathbb{E}[X^\alpha\log X].
\]
Therefore
\[
\Psi_n(\alpha)\xrightarrow{p}
g(\alpha)
:=
\frac{1}{\alpha}
+ \mathbb{E}[\log X]
-
\frac{\mathbb{E}[X^\alpha\log X]}{\mathbb{E}[X^\alpha]}.
\]
Define
\[
h(\alpha)
=
\frac{\mathbb{E}[X^\alpha\log X]}{\mathbb{E}[X^\alpha]}.
\]
Differentiation under the integral sign (justified by dominated convergence
under assumption 4) yields
\[
h'(\alpha)
=
\frac{
\mathbb{E}[X^\alpha(\log X)^2]\mathbb{E}[X^\alpha]
-
\bigl(\mathbb{E}[X^\alpha\log X]\bigr)^2
}
{\bigl(\mathbb{E}[X^\alpha]\bigr)^2}
>0,
\]
where strict positivity follows from the \textit{Cauchy--Schwarz inequality}.
Consequently
\[
g'(\alpha)=-\frac{1}{\alpha^2}-h'(\alpha)<0,
\]
so \(g\) is strictly decreasing on \((0,\infty)\).
A direct calculation using Weibull moment identities shows that
\(g(\alpha_0)=0\), and uniqueness follows.\\
\indent Now by assumptions (2)--(4), the class \(\mathcal{F}_K\) is \textit{Glivenko--Cantelli} for
the weighted triangular array induced by the APHC scheme.
Thus, $\sup_{\alpha\in K}|\Psi_n(\alpha)-g(\alpha)|\xrightarrow{p}0,$
for any compact \(K\subset(0,\infty)\).\\
\indent Furthermore, let \(\varepsilon>0\).
Since \(g\) is strictly decreasing with a unique zero at \(\alpha_0\), we have
\[
\inf_{|\alpha-\alpha_0|\ge\varepsilon}|g(\alpha)|>0.
\]
Uniform convergence implies that, with probability tending to one,
\(\Psi_n(\alpha)\) has no root outside an \(\varepsilon\)-neighbourhood of
\(\alpha_0\).
Therefore
$\hat\alpha_n\xrightarrow{p}\alpha_0,$
completing the proof.
\end{proof}

\section{Estimating the Stress Coefficients Using Regression Approach}\label{sec4}
The full information ML yields two equations for six parameters. It is computationally expensive to optimize for these six variables over the entire search space, and highly dependent on the choice of initial value. Therefore, we use a two-step estimation procedure, where we first evaluate the ML estimators of the shape and scale parameters, and then we substitute the estimated parameters and perform OLS regression to estimate the values of the coefficients. We note that the errors are converging in probability to zero, so it is possible to consistently estimate the coefficients. In general, two-step estimators require a variance correction to account for the estimation of nuisance parameters. Since the coefficients $a_0,a_1,a_2,c_0,c_1,c_2$ are estimated after replacing the shape and scale parameters by the preliminary estimators, the resulting estimator is a two–step (plug-in) estimator. Accordingly, the standard errors (SEs) of the coefficients, and thus the corresponding $t$-statistic values, are obtained using the \cite{murphy_topel_1985}'s variance correction to account for the estimation error in the first step. All reported $t$-values and confidence intervals (CIs) in our simulation findings are also based on the \cite{murphy_topel_1985}'s corrected variance.\\
\indent The estimated parameters are consistent, as proved in Section \ref{sec3}. Since the STF is linear, a natural choice for estimating the coefficients corresponding to the stress variables (of which the parameters are functions) is to perform linear regression. The theoretical framework of error in the regression of variables is applicable and the OLS estimates converge asymptotically to the true parameters. Here, we regress on 50 different predictors and response variable pairs in order to obtain the asymptotically efficient estimates. It is also possible to execute this with fewer pairs of response variables. We follow the following two-step procedure to estimate the stress coefficients $a_0, a_1, a_2, c_0, c_1, c_2$ from (\ref{eq3})-(\ref{eq4}). 
\begin{enumerate}
    \item Using the Newton-Rapshon algorithm, we first obtain the ML estimators $\hat{\alpha}$ and $\hat{\lambda}$ using (\ref{eq5})-(\ref{eq12}). The estimated $\hat{\alpha}$ can be used as a plug-in estimator to calculate $\hat{\lambda}$. We do this for 50 different pairs of temperature and voltage.
     \item Thereafter we take the OLS estimates of the stress coefficients $a_0, a_1, a_2, c_0, c_1, c_2$.
\end{enumerate}
The fitted regression models are 
$$\hat{\alpha_i} = a_0 + \frac{a_1}{T_i}+a_2 \log V_i + \epsilon_i,$$ $$\hat{\lambda_i} = c_0 + \frac{c_1}{T_i}  +c_2 \log V_i +\epsilon_i,$$
for every $i \in [50]$. We note that $\hat{\alpha}$ is consistent, and so is $\hat{\lambda}$ due to the \textit{continuous mapping theorem}. Thus, $\epsilon_i \rightarrow 0$ in probability. This ensures the asymptotic consistency of the OLS estimators of $a_0,a_1,a_2,c_0,c_1,c_2$. Moreover, if we take the average of the OLS estimates after multiple iterations, by CLT, $\bar{\epsilon_i}$ converges to a normal random variable in distribution. We used Newton-Raphson Method with 25 iterations to numerically solve for $\alpha$. We observed that the equation involving $\alpha$ has only one real root, as the function is monotonically decreasing.

\subsection{Estimating Function of $\alpha$ has a Unique Root}\label{4.1}
We intend to prove that the ML estimator $\hat{\alpha}$ obtained under both cases of PHC and APHC in Section \ref{sec3} has a unique root.
\subsubsection{PHC and APHC Case 1}\label{4.1.1}
The estimating equation of $\alpha$ under PHC and APHC in Section \ref{sec3} takes the form $$f(\alpha) = \frac{m}{\alpha} + \sum_{i=1}^{m}\log x_i - m \frac{\sum_{i=1}^{m}x_i^\alpha (\log x_i) (1+R_i)}{\sum_{i=1}^{m}(1+R_i)x_i^\alpha}. $$ Differentiating w.r.t. $\alpha$ we obtain $$f'(\alpha) = -\frac{m}{\alpha^2} - m \frac{(\sum_{i=1}^{m}x_i^\alpha (\log x_i)^2 (1+R_i))(\sum_{i=1}^{m}(1+R_i)x_i^\alpha)-(\sum_{i=1}^{m}x_i^\alpha (\log x_i)(1+R_i))^2}{(\sum_{i=1}^{m}(1+R_i)x_i^\alpha)^2}.$$ We note that by \textit{Cauchy-Schwarz inequality} \[\bigg(\sum_{i=1}^{m}x_i^\alpha (\log x_i)^2 (1+R_i)\bigg)\bigg(\sum_{i=1}^{m}(1+R_i)x_i^\alpha\bigg)-\bigg(\sum_{i=1}^{m}x_i^\alpha (\log x_i)(1+R_i)\bigg)^2 \geq 0.\] Thus, $f(\alpha)$ is monotonically decreasing as $-\frac{m}{\alpha^2}<0$. Moreover, as $\alpha \rightarrow 0$, $f(\alpha) \rightarrow \infty$, and as $\alpha \rightarrow \infty$, $f(\alpha) \rightarrow \sum_{i=1}^{m-1}\log (\frac{x_{(i)}}{x_{(m)}}) < 0$ as $\frac{x_{(i)}}{x_{(m)}} < 1$, $x_{(i)}$ denoting the $i^{th}$ ordered statistic.
As $f$ is continuous, by \textit{Intermediate Value Theorem} (IVT), since $f$ has a root, and it is monotonically decreasing, it must have exactly one root.

\subsubsection{PHC Case 2}\label{4.1.2} 
In PHC Case 2, the estimating equation is of the form $$f(\alpha) = \frac{j}{\alpha} + \sum_{i=1}^{j}\log x_i - j \frac{\sum_{i=1}^{j}x_i^\alpha (\log x_i) (1+R_i) + R_j^* T^\alpha\log T}{\sum_{i=1}^{j}(1+R_i)x_i^\alpha + R_j^*T^\alpha}.$$  We take the derivative w.r.t. $\alpha$ to obtain 
\begin{align*}
    f'(\alpha) = -\frac{j}{\alpha^2} &- j \frac{\bigg(\sum_{i=1}^{j}(1+R_i)x_i^\alpha + R_j^*T^\alpha\bigg)\bigg(\sum_{i=1}^{j}x_i^\alpha (\log x_i)^2 (1+R_i) + R_j^* T^\alpha(\log T)^2\bigg)}{\bigg(\sum_{i=1}^{j}(1+R_i)x_i^\alpha + R_j^*T^\alpha\bigg)^2}\\
    &-j \frac{ \bigg(\sum_{i=1}^{j}x_i^\alpha (\log x_i) (1+R_i) + R_j^* T^\alpha(\log T)\bigg)^2}{\bigg(\sum_{i=1}^{j}(1+R_i)x_i^\alpha + R_j^*T^\alpha\bigg)^2}.
    \end{align*}
Again, by applying \textit{Cauchy-Schwarz inequality} and making similar observations as in Case 1, one can show that $f(\alpha)$ has a unique root. 

\subsubsection{APHC Case 2}\label{4.1.3}
The estimating equation under APHC Case 2 is of the form \begin{dmath}
    f(\alpha) = \frac{m}{\alpha} + \sum_{i=1}^{m}\log x_i - m \frac{\sum_{i=1}^{m}x_i^\alpha (\log x_i)+ \sum_{i=1}^{j}R_ix_i^\alpha (\log x_i) + (n-m-\sum_{i=1}^{j}R_i)x_m^\alpha (\log x_m) }{\sum_{i=1}^{m}x_i^\alpha + \sum_{i=1}^{j}R_ix_i^\alpha +(n-m-\sum_{i=1}^{j}R_i)x_m^\alpha }.
\end{dmath}
Taking derivative, we obtain
\begin{align*}
f'(\alpha) = -\frac{m}{\alpha^2}     
&- m 
\frac{
\left(
\sum_{i=1}^{m} x_i^{\alpha} 
+ \sum_{i=1}^{j} R_i x_i^{\alpha} 
+ (n - m - \sum_{i=1}^{j} R_i) x_m^{\alpha}
\right)
}{
\left(
\sum_{i=1}^{m} x_i^{\alpha}
+ \sum_{i=1}^{j} R_i x_i^{\alpha}
+ (n - m - \sum_{i=1}^{j} R_i) x_m^{\alpha}
\right)^2
}\\
&\times \frac{\left(
\sum_{i=1}^{m} x_i^{\alpha} (\log x_i)^2 
+ \sum_{i=1}^{j} R_i x_i^{\alpha} (\log x_i)^2 
+ (n - m - \sum_{i=1}^{j} R_i) x_m^{\alpha} (\log x_m)^2
\right)}{\left(
\sum_{i=1}^{m} x_i^{\alpha}
+ \sum_{i=1}^{j} R_i x_i^{\alpha}
+ (n - m - \sum_{i=1}^{j} R_i) x_m^{\alpha}
\right)^2}
\end{align*}
Again, this expression can be shown to be negative using \textit{Cauchy-Schwarz inequality} similar to the aforementioned cases. The existence of unique root follows from the decreasing nature of the function and by IVT, due to continuity of the function.

\section{Analyses from Simulations}\label{sec6}

We first generate progressively type II right censored samples from independent standard Uniform random variables according to \cite{balakrishnan2000progressive}. The sample generation steps are as follows:
\begin{itemize}
\item $m$ i.i.d. $\text{Uniform}(0,1)$ samples are generated, where $m$ denotes the number of failures (pre-determined).
\item For each $U_i$ following $\text{Uniform}(0,1)$, define $Z_i = -\log (1-U_i)$. Then, $Z_i$ are i.i.d. $\text{Exp}(1)$.
\item Define $X_1 = \frac{Z_1}{n}$, $X_2 = X_1 + \frac{Z_2}{n-R_1-1}$. In general, define $X_k = X_{k-1}+\frac{Z_k}{n-\sum_{i=1}^{k-1}R_i-k+1}$ for $k = 2,3,\dots,m$. 
\item Define $x_i = F^{-1}(1-e^{-X_i})$ for $i=1,2,\dots,m$, where $F$ is the CDF of the $\text{Weibull}(\alpha,\lambda)$ distribution. Thereafter, $x_i;\,\ i=1,2,\dots,m$ is the progressively hybrid right censored sample from Weibull distribution with appropriate shape and scale parameters.
\end{itemize}
In order to generate PHC samples, we take the $x_i's$ generated above, and sort $X_1,X_2,\dots,X_m,T$ in ascending order. In this way, we can determine $J$ such that $J = \max\{i:x_i\leq T\}$. If $J=m$, we proceed with case I, else with case II. Furthermore, to generate the APHC samples, we follow the approach of \cite{ng2009exponential}. The steps are as follows: we determine $J = \max\{i:x_i\leq T\}$. If $J < m$, we do the following:
\begin{itemize}
    \item Discard the samples $X_{(j+2)},X_{(j+3)},\dots,X_{(m)}$.
    \item Generate the first $m-j-1$ order statistics from the truncated distribution $\frac{f(x)}{1-F(x_{(j+1)})}$, with sample size $(n-\sum_{i=1}^{j}R_i-j-1)$ and label them as $X_{(j+2)},X_{(j+3)},\dots,X_{(m)}$.
\end{itemize}

\indent We obtain results for 15 theoretically generated datasets under PHC and APHC with different \textit{total number of samples}, $n$, \textit{number of effective samples}, $m$, and \textit{censoring schemes}, $R$ with the following choices of parameters: $T = 2.73$, $a_0 = 10$, $a_1 = 140$, $a_2 = -3$, $c_0 = 7$, $c_1 = 125$, $c_2 = -2$. The details for each dataset are provided in Table \ref{datasets}.\\
\indent Although we have obtained results for all 15 combinations but due to paucity of space, we are reporting only three combinations (datasets 3, 9, and 15). For these three combinations under PHC, the true values and ML estimates of the Weibull parameters, along with the SEs and stress variables (temperature, $T$ and voltage, $V$) are provided in Tables \ref{MLPHC3}, \ref{MLPHC9} and \ref{MLPHC15}, respectively. The histograms depicting the empirical distribution of the estimated shape parameter under PHC for datasets 3, 9, and 15 are provided in Figures \ref{PHC1HAl1}, \ref{PHC1HAl2} and \ref{PHC1HAl3}. Likewise Figures \ref{PHC1HLam1}, \ref{PHC1HLam2} and \ref{PHC1HLam15} exhibit the histograms depicting the empirical distribution of the scale parameter under PHC. Moreover, estimates for the coefficients of stress variables (constant, inverse temperature and log voltage) based on the OLS regression, along with SEs, $t$-statistic values, $p$-values, and CIs are reported in Tables \ref{OLSPHC3}, \ref{OLSPHC9} and \ref{OLSPHC15}, respectively. \\
\indent For the same three combinations (datasets 3, 9, and 15), the true values and ML estimates of the Weibull parameters, along with the SEs and stress variables (temperature, $T$ and voltage, $V$) under APHC are provided in Tables \ref{MLAPHC3}, \ref{MLAPHC9} and \ref{MLAPHC15}, respectively. Figures \ref{APHC1HAl3}, \ref{APHC1HAl9} and \ref{APHC1HAl15} show the histograms depicting the empirical distribution of the estimated shape parameter under APHC for datasets 3, 9, and 15. Similarly the histograms depicting the empirical distribution of the scale parameter under APHC are available in Figures \ref{APHC1HLam1}, \ref{APHC1HLam2} and \ref{APHC1HLam15}. Moreover, estimates for the coefficients of stress variables (constant, inverse temperature and log voltage) based on OLS regression, along with SEs, $t$-statistic values, $p$-values, and CIs are reported in Tables \ref{OLSAPHC3}, \ref{OLSAPHC9} and \ref{OLSAPHC15}, respectively.\\
\indent Note that in the regression results (Tables \ref{OLSPHC3}, \ref{OLSPHC9}, \ref{OLSPHC15}, \ref{OLSAPHC3}, \ref{OLSAPHC9} and \ref{OLSAPHC15}), panel A reports estimates for the Weibull scale parameter, while Panel B reports estimates for the shape parameter. The 95\% CIs are reported in brackets. All coefficients are statistically significant at the 1\% level.\\
\indent Some key findings from the simulation results are as follows:
\begin{enumerate}   
\item The ML estimator for both shape and scale parameters converge to the true parameter values, as the effective sample size increases. This indicates that the ML estimator is consistent, as proved theoretically. 
\item The regression coefficients are estimated precisely over a large number of datasets, under both censoring schemes. The CIs are narrow over a large number of different datasets, indicating high information content and low sensitivity to model perturbations. 
\item The coefficient for inverse temperature is positive and significant at all levels, indicating that a rise in inverse temperature leads to increase in both shape (steepening the hazard function) and scale (increasing mean lifetimes) parameters. 
\item The coefficient for log volume is negative and significant at all levels, indicating that a rise in log volume leads to decrease in both shape (flattening the hazard function) and scale (reducing mean lifetimes) parameters. 
\end{enumerate}

\section{Data Illustration}\label{sec7}
In this section, we present a dual-stress ALT dataset, obtained from the Python library, named as \textit{Reliability} [\cite{reliability2022}] to illustrate the proposed model. This dataset represents that the ALT experiment conducted at two different temperature - voltage pairs. Total 12 failures were observed from the test with zero right censored observations. Actual values of the temperature - voltage pairs were (348K, 3V), (348K, 5V), and (378K, 5V) and the corresponding observed failure times were [620, 632, 685, 822], [380, 416, 460, 596], and [216, 146,  332, 400], respectively. The data information and its Python implementation can be accessed at \url{https://pypi.org/project/reliability/} or \url{https://reliability.readthedocs.io/en/latest/API\%20reference.html}.\\ 
\indent In order to check the suitability of the data for our ALT model under PHC and APHC, we first performed an \textit{Anderson - Darling} (AD) goodness-of-fit test to determine the distribution of the system lifetimes. We believe that our data follows a two-parameter Weibull distribution, therefore we set our null hypothesis that the sampled data is from a two parameter Weibull distribution. The outcome of the AD test supports our claim. The corresponding $p$-values with estimated parameter values are available in Table \ref{tab:ad_test}. \\
\indent We then created artificial datasets taking lifetime, $T=700$ hours and censoring scheme, $R = [1,0,0]$. Then, the PHC lifetime data (in hours) for the given temperature - voltage pairs, (348K, 3V), (348K, 5V), (378K, 5V) are $[620,632,685]$, $[380,460,596],$ and $[146,332,400]$. The APHC lifetime data in hours for the cutoff $T=700$ hours and $R = [1,0,0]$ are same as the PHC lifetime data. We estimate the shape and scale parameters of the Weibull Distribution under both types of censoring schemes. The results are provided in Table \ref{tab:phc_estimates}.\\
\indent At the second stage, the estimates of parameters at each stress level obtained under the given censoring were regressed on covariates using a deterministic least-squares fit. As the lifetime data were available under three temperature - voltage pairs only, the fitted model is exactly identified and therefore does not permit estimation of sampling variability for the regression coefficients. Consequently, the second-stage regression is interpreted as a structural projection rather than a stochastic regression model, which is consistent with the theoretical two-step estimation framework. Therefore the inference is based on the first-stage estimates, while the second-stage fit serves to characterize the functional dependence on the stress variables. The corresponding results are provided in Tables \ref{tab:coefficients1} and \ref{tab:coefficients2}.
\par This empirical data illustration demonstrates the utility of the proposed two-stress ALT framework under PHC and APHC within a real reliability setting. The first-stage analysis supports the suitability of the Weibull model for describing the lifetime distribution under accelerated conditions. The subsequent deterministic regression step provides a coherent structural characterization of the dependence of the Weibull parameters on temperature and voltage, thereby enabling extrapolation to use conditions. Although the limited number of stress combinations precludes formal second-stage inference, the resulting regression relationships remain consistent with both the theoretical two-step estimation framework and established engineering understanding of stress–life behavior. From a methodological perspective, this example highlights that the proposed procedure is not only mathematically tractable but also directly implementable in industrial reliability studies, particularly in settings where ALT is necessarily constrained by cost, time or design considerations. In this sense, the data analysis reinforces the relevance of the proposed methods for practical reliability modeling and computational statistical analysis.

\section{Concluding Remarks}\label{sec8}
\indent In this paper, our focus was on the study of two ALT schemes, namely, PHC and APHC for Weibull life time distribution. The shape and scale parameters for the product/system lifetime distribution were assumed to be stress-dependent on temperature and voltage. For both PHC and APHC ALT models, we estimated the distribution parameters and model coefficients using a novel likelihood-based regression approach. Besides the two-step procedure, the theoretical proofs of the consistency and asymptotic normality of the ML estimators under PHC and APHC are the unique additions to the ALT and censoring literature. The extensive simulation study supported the proposed theoretical contributions for a wide range of parameter configurations and the lifetime data illustration further reiterated our study. In near future, we would like to extend this work for other lifetime models and censoring schemes. 

\section*{Data Availability}
The link to the real data illustration is provided in Section 6 of the manuscript.

\section*{Conflict of Interest}
\noindent The authors have no conflict of interest to report.

\bibliographystyle{apalike}
\bibliography{ref}

\begin{table}[H]
\centering
\caption{PHC and APHC Datasets}
\label{datasets}
\small
\renewcommand{\arraystretch}{1.2}

\begin{tabularx}{\textwidth}{c c c >{\raggedright\arraybackslash}X}
\hline
\textbf{S. No.} & $\boldsymbol{n}$ & $\boldsymbol{m}$ & $\boldsymbol{R}$ \\
\hline

1 & 58 & 25 &
\begin{tabular}[t]{@{}l@{}}
$[1,1,1,1,1,1,1,1,1,1,1,1,1,1,1,1,1,1,1,1,1,1,1,9]$
\end{tabular} \\

2 & 75 & 25 &
\begin{tabular}[t]{@{}l@{}}
$[2,2,2,2,2,2,2,2,2,2,2,2,2,2,2,2,2,2,2,2,2,2,2,2,2]$
\end{tabular} \\

3 & 58 & 25 &
\begin{tabular}[t]{@{}l@{}}
$[16,0,0,0,0,0,0,0,1,1,1,1,1,1,1,1,1,1,1,1,1,1,1,1,1]$
\end{tabular} \\

4 & 80 & 40 &
\begin{tabular}[t]{@{}l@{}}
$[1,1,1,1,1,1,1,1,1,1,1,1,1,1,1,1,1,1,1,1,1,1,1,1,1,1,1,1,1,1,1,1,1,1,1,1,1,1]$
\end{tabular} \\

5 & 120 & 40 &
\begin{tabular}[t]{@{}l@{}}
$[2,2,2,2,2,2,2,2,2,2,2,2,2,2,2,2,2,2,2,2,2,2,2,2,2,2,2,2,2,2,2,2,2,2,2,2,2,2,$\\
$2,2]$
\end{tabular} \\

6 & 80 & 40 &
\begin{tabular}[t]{@{}l@{}}
$[5,6,0,0,0,0,0,0,0,0,0,1,1,1,1,1,1,1,1,1,1,1,1,1,1,1,1,1,1,1,1,1,1,1,1,1,1,1,$\\
$1,1]$
\end{tabular} \\

7 & 100 & 50 &
\begin{tabular}[t]{@{}l@{}}
$[1,1,1,1,1,1,1,1,1,1,1,1,1,1,1,1,$
$1,1,1,1,1,1,1,1,1,1,1,1,1,1,1,1,1,1,1,1,1,1,$\\
$1,1,1,1,1,1,1,1]$
\end{tabular} \\

8 & 150 & 50 &
\begin{tabular}[t]{@{}l@{}}
$[2,2,2,2,2,2,2,2,2,2,2,2,2,2,2,2,$
$2,2,2,2,2,2,2,2,2,2,2,2,2,2,2,2,2,2,2,2,2,2,$\\
$2,2,2,2,2,2,2,2]$
\end{tabular} \\

9 & 100 & 50 &
\begin{tabular}[t]{@{}l@{}}
$[0,2,0,2,0,2,0,2,0,2,0,2,0,2,0,2,$
$0,2,0,2,0,2,0,2,0,2,0,2,0,2,0,2,0,2,0,2,0,2,$\\
$0,2,0,2,0,2,0,2]$
\end{tabular} \\

10 & 150 & 75 &
\begin{tabular}[t]{@{}l@{}}
$[1,1,1,1,1,1,1,1,1,1,1,1,1,1,1,1,$
$1,1,1,1,1,1,1,1,1,1,1,1,1,1,1,1,1,1,1,1,1,1,$\\
$1,1,1,1,1,1,1,1,1,1,$
$1,1,1,1,1,1,1,1,1,1,1,1,1,1,1]$
\end{tabular} \\

11 & 225 & 75 &
\begin{tabular}[t]{@{}l@{}}
$[2,2,2,2,2,2,2,2,2,2,2,2,2,2,2,2,$
$2,2,2,2,2,2,2,2,2,2,2,2,2,2,2,2,2,2,2,2,2,2,$\\
$2,2,2,2,2,2,2,2,2,2,$
$2,2,2,2,2,2,2,2,2,2,2,2,2,2,2]$
\end{tabular} \\

12 & 150 & 75 &
\begin{tabular}[t]{@{}l@{}}
$[0,0,0,0,0,0,0,0,0,0,0,0,0,0,0,0,$
$0,0,0,0,0,0,0,0,0,0,0,0,0,0,0,0,$
$0,0,0,0,0,0,$\\
$0,0,0,0,3,3,3,3,3,3,3,3,3,3,3,3,3,3,3,3,$
$3,3,3,3,3,3,3,3,3,3,3,3]$
\end{tabular} \\

13 & 200 & 100 &
\begin{tabular}[t]{@{}l@{}}
$[1,1,1,1,1,1,1,1,1,1,1,1,1,1,1,1,$
$1,1,1,1,1,1,1,1,1,1,1,1,1,1,1,1,1,1,1,1,1,1,$\\
$1,1,1,1,1,1,1,1,1,1,$
$1,1,1,1,1,1,1,1,1,1,1,1,1,1,1,1,1,1,1,1,1,1,$
$1,1,1,1,1,1]$
\end{tabular} \\

14 & 300 & 100 &
\begin{tabular}[t]{@{}l@{}}
$[2,2,2,2,2,2,2,2,2,2,2,2,2,2,2,2,$
$2,2,2,2,2,2,2,2,2,2,2,2,2,2,2,2,2,2,2,2,2,2,$\\
$2,2,2,2,2,2,2,2,2,2,$
$2,2,2,2,2,2,2,2,2,2,2,2,2,2,2,2,2,2,2,2,2,$
$2,2,2,2,2,2,2]$
\end{tabular} \\

15 & 200 & 100 &
\begin{tabular}[t]{@{}l@{}}
$[0,0,0,0,0,0,0,0,0,0,0,0,0,0,0,0,$
$0,0,0,0,0,0,0,0,0,0,0,0,0,0,0,0,0,0,0,0,0,0,$\\
$0,0,0,0,0,0,0,0,0,0,$
$0,0,0,0,0,0,0,0,0,0,0,0,0,0,$
$50,50]$
\end{tabular} \\

\hline
\end{tabularx}
\end{table}

\begin{table}[H]
\centering
\caption{True values and ML estimates for dataset 3 under PHC}
\label{MLPHC3}

\begin{tabular}{ccccccccc}
\toprule
Index & $T$ & $V$ & $\alpha$ & $\lambda$ & $\hat{\alpha}_{\text{ML}}$ & $\hat{\lambda}_{\text{ML}}$ & $\text{SE}(\hat{\alpha}_{\text{ML}})$ & $\text{SE}(\hat{\lambda}_{\text{ML}})$ \\
\midrule
1 & 319.6469 & 12.4475 & 2.3480 & 2.8734 & 2.4551 & 2.8768 & 0.4742 & 0.3219 \\
2 & 278.6139 & 20.9161 & 1.3676 & 1.3809 & 1.4365 & 1.3748 & 0.2405 & 0.2096 \\
3 & 272.6851 & 18.2367 & 1.6515 & 1.8031 & 1.7370 & 1.7918 & 0.2937 & 0.2250 \\
4 & 305.1315 & 17.5408 & 1.6806 & 1.8652 & 1.7678 & 1.8475 & 0.3118 & 0.2487 \\
5 & 321.9469 & 15.1132 & 1.9571 & 2.2882 & 2.0239 & 2.2838 & 0.3517 & 0.2525 \\
6 & 292.3106 & 14.6494 & 2.0588 & 2.4257 & 2.1744 & 2.4042 & 0.3986 & 0.2687 \\
7 & 348.0764 & 16.0043 & 1.8134 & 2.0836 & 1.8949 & 2.0783 & 0.3288 & 0.2545 \\
8 & 318.4830 & 19.1756 & 1.4852 & 1.5787 & 1.5656 & 1.5686 & 0.2695 & 0.2247 \\
9 & 298.0932 & 21.5055 & 1.2827 & 1.2647 & 1.3505 & 1.2557 & 0.2282 & 0.2112 \\
10 & 289.2118 & 17.1251 & 1.7511 & 1.9624 & 1.8281 & 1.9504 & 0.3195 & 0.2316 \\
\bottomrule
\end{tabular}

\end{table}

\begin{table}[H]
\centering
\caption{Regression results for Weibull parameters under PHC for dataset 3}
\label{OLSPHC3}
\begin{tabular}{lccccc}
\hline
Variable & Coefficient & SE & $t$-statistic & $p$-value & 95\% CI \\
\hline
\multicolumn{6}{l}{\textbf{Panel A: Scale Parameter}} \\
\hline
Constant   
& 7.2322  & 0.026 & 278.485  & $<10^{-4}$ & [7.180,\; 7.284] \\
Inv. Temp. 
& 139.2532 & 5.890 & 23.642   & $<10^{-4}$ & [127.404,\; 151.103] \\
Log Vol.   
& -2.0685 & 0.008 & -257.408 & $<10^{-4}$ & [-2.085,\; -2.052] \\
\hline
\multicolumn{6}{l}{\textbf{Panel B: Shape Parameter}} \\
\hline
Constant   
& 10.0984 & 0.022 & 468.287  & $<10^{-4}$ & [10.055,\; 10.142] \\
Inv. Temp. 
& 143.2864 & 4.891 & 29.296   & $<10^{-4}$ & [133.447,\; 153.126] \\
Log Vol.   
& -3.0421 & 0.007 & -455.887 & $<10^{-4}$ & [-3.056,\; -3.029] \\
\hline
\end{tabular}
\end{table}

\begin{table}[htbp]
\centering
\caption{True values and ML estimates for dataset 9 under PHC}
\label{MLPHC9}
\begin{tabular}{ccccccccc}
\toprule
Index & $T$ & $V$ & $\alpha$ & $\lambda$ & $\hat{\alpha}_{\text{ML}}$ & $\hat{\lambda}_{\text{ML}}$ & $\text{SE}(\hat{\alpha}_{\text{ML}})$ & $\text{SE}(\hat{\lambda}_{\text{ML}})$ \\
\midrule

1 & 319.6469 & 12.4475 & 2.3480 & 2.8734 & 2.4029 & 2.8738 & 0.3197 & 0.2144 \\
2 & 278.6139 & 20.9161 & 1.3676 & 1.3809 & 1.4062 & 1.3715 & 0.1601 & 0.1491 \\
3 & 272.6851 & 18.2367 & 1.6515 & 1.8031 & 1.6912 & 1.7918 & 0.2059 & 0.1553 \\
4 & 305.1315 & 17.5408 & 1.6806 & 1.8652 & 1.7166 & 1.8641 & 0.2046 & 0.1657 \\
5 & 321.9469 & 15.1132 & 1.9571 & 2.2882 & 1.9956 & 2.2882 & 0.2479 & 0.1850 \\
6 & 292.3106 & 14.6494 & 2.0588 & 2.4257 & 2.1066 & 2.4186 & 0.2658 & 0.1842 \\
7 & 348.0764 & 16.0043 & 1.8134 & 2.0836 & 1.8628 & 2.0792 & 0.2292 & 0.1758 \\
8 & 318.4830 & 19.1756 & 1.4852 & 1.5787 & 1.5223 & 1.5693 & 0.1773 & 0.1518 \\
9 & 298.0932 & 21.5055 & 1.2827 & 1.2647 & 1.3294 & 1.2594 & 0.1589 & 0.1437 \\
10 & 289.2118 & 17.1251 & 1.7511 & 1.9624 & 1.7896 & 1.9479 & 0.2101 & 0.1649 \\
\hline
\end{tabular}
\end{table}

\begin{table}[H]
\centering
\caption{Regression results for Weibull parameters under PHC for dataset 9}
\label{OLSPHC9}
\begin{tabular}{lccccc}
\hline
Variable & Coefficient & SE & $t$-statistic & $p$-value & 95\% CI \\
\hline
\multicolumn{6}{l}{\textbf{Panel A: Scale Parameter}} \\
\hline
Constant   
& 7.1539  & 0.016 & 454.403  & $<10^{-4}$ & [7.122,\; 7.186] \\
Inv. Temp. 
& 130.6625 & 3.571 & 36.592   & $<10^{-4}$ & [123.479,\; 137.846] \\
Log Vol.   
& -2.0453 & 0.005 & -419.838 & $<10^{-4}$ & [-2.055,\; -2.036] \\
\hline
\multicolumn{6}{l}{\textbf{Panel B: Shape Parameter}} \\
\hline
Constant   
& 10.0715 & 0.014 & 696.167  & $<10^{-4}$ & [10.042,\; 10.101] \\
Inv. Temp. 
& 137.1551 & 3.281 & 41.800   & $<10^{-4}$ & [130.554,\; 143.756] \\
Log Vol.   
& -3.0241 & 0.004 & -675.525 & $<10^{-4}$ & [-3.033,\; -3.015] \\
\hline
\end{tabular}
\end{table}

\begin{table}[htbp]
\centering
\caption{True values and ML estimates for dataset 15 under PHC}
\label{MLPHC15}
\begin{tabular}{ccccccccc}
\toprule
Index & $T$ & $V$ & $\alpha$ & $\lambda$ & $\hat{\alpha}_{\text{ML}}$ & $\hat{\lambda}_{\text{ML}}$ & $\text{SE}(\hat{\alpha}_{\text{ML}})$ & $\text{SE}(\hat{\lambda}_{\text{ML}})$ \\
\midrule

1 & 319.6469 & 12.4475 & 2.3480 & 2.8734 & 2.3829 & 2.8659 & 0.22010 & 0.14109 \\
2 & 278.6139 & 20.9161 & 1.3676 & 1.3809 & 1.3919 & 1.3761 & 0.13104 & 0.11036 \\
3 & 272.6851 & 18.2367 & 1.6515 & 1.8031 & 1.6802 & 1.7957 & 0.15471 & 0.11807 \\
4 & 305.1315 & 17.5408 & 1.6806 & 1.8652 & 1.7186 & 1.8556 & 0.15742 & 0.12033 \\
5 & 321.9469 & 15.1132 & 1.9571 & 2.2882 & 1.9978 & 2.2771 & 0.18686 & 0.12854 \\
6 & 292.3106 & 14.6494 & 2.0588 & 2.4257 & 2.0825 & 2.4183 & 0.19449 & 0.13305 \\
7 & 348.0764 & 16.0043 & 1.8134 & 2.0836 & 1.8421 & 2.0776 & 0.17727 & 0.13009 \\
8 & 318.4830 & 19.1756 & 1.4852 & 1.5787 & 1.5121 & 1.5734 & 0.14405 & 0.11977 \\
9 & 298.0932 & 21.5055 & 1.2827 & 1.2647 & 1.3075 & 1.2504 & 0.12405 & 0.11123 \\
10 & 289.2118 & 17.1251 & 1.7511 & 1.9624 & 1.7853 & 1.9539 & 0.16743 & 0.12381 \\
\hline
\end{tabular}
\end{table}

\begin{table}[H]
\centering
\caption{Regression results for Weibull parameters under PHC for dataset 15}
\label{OLSPHC15}
\begin{tabular}{lccccc}
\hline
Variable & Coefficient & SE & $t$-statistic & $p$-value & 95\% CI \\
\hline
\multicolumn{6}{l}{\textbf{Panel A: Scale Parameter}} \\
\hline
Constant   
& 7.0939 & 0.013 & 532.783 & $<10^{-4}$ & [7.067,\; 7.121] \\
Inv. Temp. 
& 125.5734 & 3.020 & 41.582 & $<10^{-4}$ & [119.498,\; 131.649] \\
Log Vol.   
& -2.0230 & 0.004 & -490.995 & $<10^{-4}$ & [-2.031,\; -2.015] \\
\hline
\multicolumn{6}{l}{\textbf{Panel B: Shape Parameter}} \\
\hline
Constant   
& 9.9914 & 0.009 & 1106.893 & $<10^{-4}$ & [9.973,\; 10.010] \\
Inv. Temp. 
& 140.6379 & 2.047 & 68.694 & $<10^{-4}$ & [136.519,\; 144.757] \\
Log Vol.   
& -3.0005 & 0.003 & -1074.216 & $<10^{-4}$ & [-3.006,\; -2.995] \\
\hline
\end{tabular}
\end{table}

\begin{table}[htbp]
\centering
\caption{True values and ML estimates for dataset 3 under APHC}
\label{MLAPHC3}
\begin{tabular}{ccccccccc}
\toprule
Index & $T$ & $V$ & $\alpha$ & $\lambda$ & $\hat{\alpha}_{\text{ML}}$ & $\hat{\lambda}_{\text{ML}}$ & $\text{SE}(\hat{\alpha}_{\text{ML}})$ & $\text{SE}(\hat{\lambda}_{\text{ML}})$ \\
\midrule
1 & 319.6469 & 12.4475 & 2.3480 & 2.8734 & 2.4601 & 2.8414 & 0.4029 & 0.2502 \\
2 & 278.6139 & 20.9161 & 1.3676 & 1.3809 & 1.4315 & 1.3722 & 0.2359 & 0.2054 \\
3 & 272.6851 & 18.2367 & 1.6515 & 1.8031 & 1.7668 & 1.7782 & 0.2987 & 0.2253 \\
4 & 305.1315 & 17.5408 & 1.6806 & 1.8652 & 1.7849 & 1.8559 & 0.3039 & 0.2366 \\
5 & 321.9469 & 15.1132 & 1.9571 & 2.2882 & 2.0807 & 2.2501 & 0.3398 & 0.2408 \\
6 & 292.3106 & 14.6494 & 2.0588 & 2.4257 & 2.1519 & 2.4035 & 0.3351 & 0.2471 \\
7 & 348.0764 & 16.0043 & 1.8134 & 2.0836 & 1.9132 & 2.0566 & 0.3319 & 0.2276 \\
8 & 318.4830 & 19.1756 & 1.4852 & 1.5787 & 1.5743 & 1.5611 & 0.2561 & 0.2241 \\
9 & 298.0932 & 21.5055 & 1.2827 & 1.2647 & 1.3536 & 1.2443 & 0.2197 & 0.2026 \\
10 & 289.2118 & 17.1251 & 1.7511 & 1.9624 & 1.8444 & 1.9369 & 0.3027 & 0.2170 \\

\hline
\end{tabular}
\end{table}

\begin{table}[H]
\centering
\caption{Regression results for Weibull parameters under APHC for dataset 3}
\label{OLSAPHC3}
\begin{tabular}{lcccccc}
\hline
Variable & Coefficient & SE & $t$-statistic & $p$-value & 95\% CI \\
\midrule
\multicolumn{6}{l}{\textbf{Panel A: Shape parameter}} \\
\midrule
Constant     & 7.3438  & 0.023 & 318.018 & $<10^{-4}$ & [7.297, \; 7.390] \\
Inv. Temp.   & 130.5699 & 5.238 & 24.929  & $<10^{-4}$ & [120.033, \; 141.107] \\
Log Voltage  & -2.0936 & 0.007 & -292.986 & $<10^{-4}$ & [-2.108, \; -2.079] \\
\midrule
\multicolumn{6}{l}{\textbf{Panel B: Scale parameter}} \\
\midrule
Constant     & 9.9097  & 0.017 & 596.600 & $<10^{-4}$ & [9.876, \; 9.943] \\
Inv. Temp.   & 145.1120 & 3.767 & 38.518  & $<10^{-4}$ & [137.533, \; 152.691] \\
Log Voltage  & -2.9823 & 0.005 & -580.236 & $<10^{-4}$ & [-2.993, \; -2.972] \\
\bottomrule
\end{tabular}
\end{table}

\begin{table}[htbp]
\centering
\caption{True values and ML estimates for dataset 9 under APHC}
\label{MLAPHC9}
\begin{tabular}{ccccccccc}
\toprule
Index & $T$ & $V$ & $\alpha$ & $\lambda$ & $\hat{\alpha}_{\text{ML}}$ & $\hat{\lambda}_{\text{ML}}$ & $\text{SE}(\hat{\alpha}_{\text{ML}})$ & $\text{SE}(\hat{\lambda}_{\text{ML}})$ \\
\midrule
1 & 319.6469 & 12.4475 & 2.3480 & 2.8734 & 2.4204 & 2.8594 & 0.2813 & 0.1785 \\
2 & 278.6139 & 20.9161 & 1.3676 & 1.3809 & 1.4127 & 1.3692 & 0.1666 & 0.1470 \\
3 & 272.6851 & 18.2367 & 1.6515 & 1.8031 & 1.7118 & 1.7907 & 0.2023 & 0.1514 \\
4 & 305.1315 & 17.5408 & 1.6806 & 1.8652 & 1.7336 & 1.8472 & 0.2006 & 0.1528 \\
5 & 321.9469 & 15.1132 & 1.9571 & 2.2882 & 2.0189 & 2.2761 & 0.2308 & 0.1721 \\
6 & 292.3106 & 14.6494 & 2.0588 & 2.4257 & 2.1391 & 2.4055 & 0.2473 & 0.1683 \\
7 & 348.0764 & 16.0043 & 1.8134 & 2.0836 & 1.8729 & 2.0676 & 0.2116 & 0.1676 \\
8 & 318.4830 & 19.1756 & 1.4852 & 1.5787 & 1.5409 & 1.5654 & 0.1824 & 0.1530 \\
9 & 298.0932 & 21.5055 & 1.2827 & 1.2647 & 1.3121 & 1.2558 & 0.1488 & 0.1430 \\
10 & 289.2118 & 17.1251 & 1.7511 & 1.9624 & 1.7967 & 1.9453 & 0.2069 & 0.1635 \\

\hline
\end{tabular}
\end{table}

\begin{table}[H]
\centering
\caption{Regression results for Weibull parameters under APHC for dataset 9}
\label{OLSAPHC9}
\begin{tabular}{lccccc}
\hline
Variable & Coefficient & SE & $t$-statistic & $p$-value & 95\% CI \\
\hline
\multicolumn{6}{l}{\textbf{Panel A: Scale Parameter}} \\
\hline
Constant   
& 7.2054  & 0.018 & 403.927  & $<10^{-4}$ & [7.170,\; 7.241] \\
Inv. Temp. 
& 130.4373 & 4.046 & 32.239   & $<10^{-4}$ & [122.298,\; 138.577] \\
Log Vol.   
& -2.0606 & 0.006 & -373.301 & $<10^{-4}$ & [-2.072,\; -2.049] \\
\hline
\multicolumn{6}{l}{\textbf{Panel B: Shape Parameter}} \\
\hline
Constant   
& 9.9497  & 0.013 & 740.748  & $<10^{-4}$ & [9.923,\; 9.977] \\
Inv. Temp. 
& 140.3486 & 3.046 & 46.069   & $<10^{-4}$ & [134.220,\; 146.477] \\
Log Vol.   
& -2.9875 & 0.004 & -718.765 & $<10^{-4}$ & [-2.996,\; -2.979] \\
\hline
\end{tabular}
\end{table}

\begin{table}[htbp]
\centering
\caption{True values and ML estimates for dataset 15 under APHC}
\label{MLAPHC15}
\begin{tabular}{ccccccccc}
\toprule
Index & $T$ & $V$ & $\alpha$ & $\lambda$ & $\hat{\alpha}_{\text{ML}}$ & $\hat{\lambda}_{\text{ML}}$ & $\text{SE}(\hat{\alpha}_{\text{ML}})$ & $\text{SE}(\hat{\lambda}_{\text{ML}})$ \\
\midrule
1 & 319.6469 & 12.4475 & 2.3480 & 2.8734 & 2.3865 & 2.8665 & 0.2278 & 0.1349 \\
2 & 278.6139 & 20.9161 & 1.3676 & 1.3809 & 1.3944 & 1.3723 & 0.1310 & 0.1137 \\
3 & 272.6851 & 18.2367 & 1.6515 & 1.8031 & 1.6850 & 1.7919 & 0.1577 & 0.1203 \\
4 & 305.1315 & 17.5408 & 1.6806 & 1.8652 & 1.6959 & 1.8610 & 0.1570 & 0.1283 \\
5 & 321.9469 & 15.1132 & 1.9571 & 2.2882 & 1.9917 & 2.2710 & 0.1886 & 0.1297 \\
6 & 292.3106 & 14.6494 & 2.0588 & 2.4257 & 2.0987 & 2.4219 & 0.1981 & 0.1271 \\
7 & 348.0764 & 16.0043 & 1.8134 & 2.0836 & 1.8417 & 2.0752 & 0.1785 & 0.1295 \\
8 & 318.4830 & 19.1756 & 1.4852 & 1.5787 & 1.5115 & 1.5726 & 0.1423 & 0.1172 \\
9 & 298.0932 & 21.5055 & 1.2827 & 1.2647 & 1.3091 & 1.2540 & 0.1249 & 0.1102 \\
10 & 289.2118 & 17.1251 & 1.7511 & 1.9624 & 1.7731 & 1.9644 & 0.1683 & 0.1256 \\

\hline
\end{tabular}
\end{table}

\begin{table}[H]
\centering
\caption{Regression results for Weibull parameters under APHC for dataset 15}
\label{OLSAPHC15}
\begin{tabular}{lccccc}
\hline
Variable & Coefficient & SE & $t$-statistic & $p$-value & 95\% CI \\
\hline
\multicolumn{6}{l}{\textbf{Panel A: Scale Parameter}} \\
\hline
Constant   
& 7.1226  & 0.014 & 522.051  & $<10^{-4}$ & [7.095,\; 7.150] \\
Inv. Temp. 
& 129.8130 & 3.094 & 41.950   & $<10^{-4}$ & [123.588,\; 136.038] \\
Log Vol.   
& -2.0377 & 0.004 & -482.661 & $<10^{-4}$ & [-2.046,\; -2.029] \\
\hline
\multicolumn{6}{l}{\textbf{Panel B: Shape Parameter}} \\
\hline
Constant   
& 9.9701  & 0.011 & 916.860  & $<10^{-4}$ & [9.948,\; 9.992] \\
Inv. Temp. 
& 142.1009 & 2.466 & 57.615   & $<10^{-4}$ & [137.139,\; 147.063] \\
Log Vol.   
& -2.9948 & 0.003 & -890.019 & $<10^{-4}$ & [-3.002,\; -2.988] \\
\hline
\end{tabular}
\end{table}

\begin{table}[H]
\centering
\caption{AD Test Results}
\label{tab:ad_test}
\begin{tabular}{c|c|>{\centering\arraybackslash}p{4cm}|>{\centering\arraybackslash}p{2.5cm}|>{\centering\arraybackslash}p{2.5cm}}
\hline
\multirow{2}{*}{Dataset} & 
\multirow{2}{*}{$p$-value from AD test} & 
\multirow{2}{*}{Decision} & 
\multicolumn{2}{c}{Estimated Parameters} \\
\cline{4-5}
 &  &  & Shape & Scale \\
\hline
1 & 0.2207 & Insufficient evidence to reject null at $5\%$ level of significance. Data follows two-parameter Weibull. & 8.5582 & 727.4088 \\
\hline
2 & 0.4408 & Insufficient evidence to reject null at $5\%$ level of significance. Data follows two-parameter Weibull. & 5.8113 & 498.6429 \\
\hline
3 & 0.7930 & Insufficient evidence to reject null at $5\%$ level of significance. Data follows two-parameter Weibull. & 3.1457 & 307.1342 \\
\hline
\end{tabular}
\end{table}

\begin{table}[htbp]
\centering
\caption{Estimated Parameters under PHC and APHC}
\label{tab:phc_estimates}
\begin{tabular}{c|c|c}
\hline
Dataset & Estimated Shape & Estimated Scale \\
\hline
1 & 25.2936 & 662.5425 \\
2 & 6.3254 & 522.8471 \\
3 & 3.3928 & 332.1593 \\
\hline
\end{tabular}
\end{table}

\begin{table}[H]
\centering
\caption{Estimated Coefficients for Shape Parameter}
\label{tab:coefficients1}
\begin{tabular}{c|c}
\hline
 & Estimated Coefficients \\
\hline
Intercept & 87.7641 \\
\hline
Inverse Temperature & -0.0978 \\
\hline
Voltage & -9.4841 \\
\hline
\end{tabular}
\end{table} 

\begin{table}[H]
\centering
\caption{Estimated Coefficients for Scale Parameter}
\label{tab:coefficients2}
\begin{tabular}{c|c}
\hline
 & Estimated Coefficients \\
\hline
Intercept & 3084.0641 \\
\hline
Inverse Temperature & -6.3563 \\
\hline
Voltage & -69.8477 \\
\hline
\end{tabular}
\end{table}

\begin{figure}[H]
    \centering
    \subfigure[]{\includegraphics[width=0.15\textwidth]{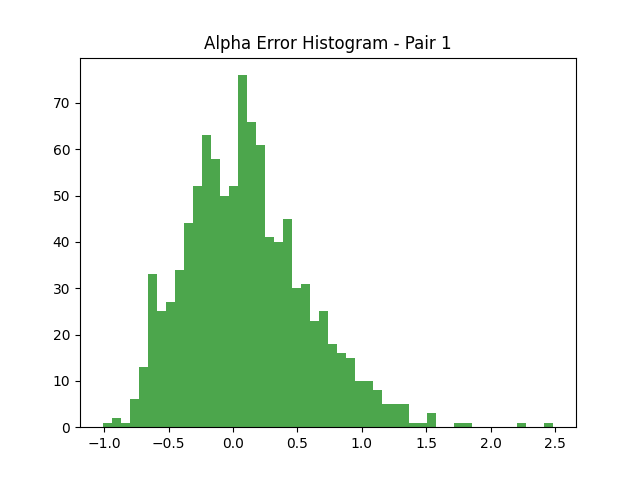}}
    \subfigure[]{\includegraphics[width=0.15\textwidth]{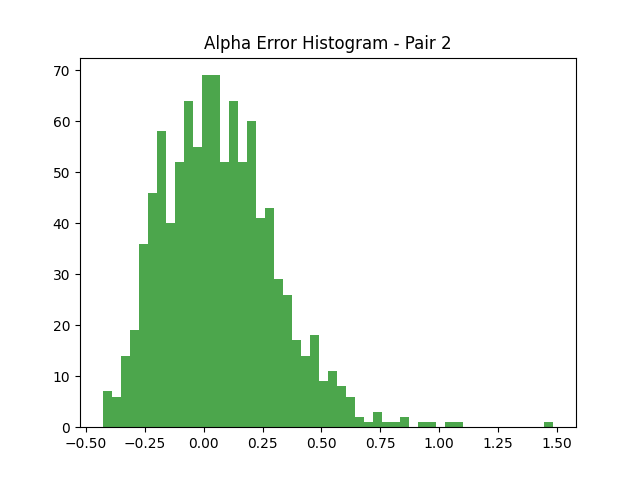}}
      \subfigure[]{\includegraphics[width=0.15\textwidth]{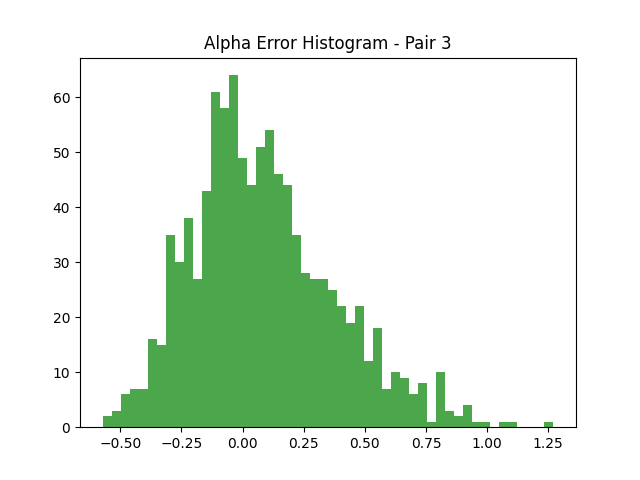}}
        \subfigure[]{\includegraphics[width=0.15\textwidth]{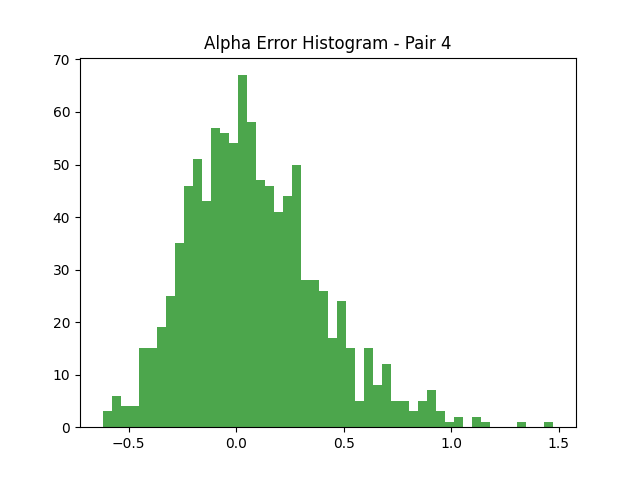}}
          \subfigure[]{\includegraphics[width=0.15\textwidth]{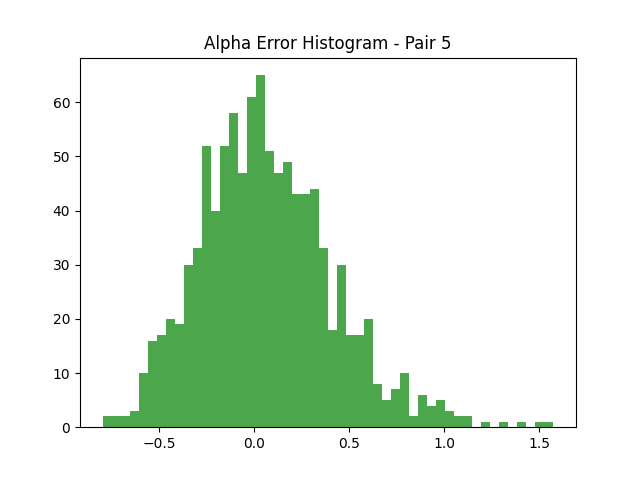}}
            \subfigure[]{\includegraphics[width=0.15\textwidth]{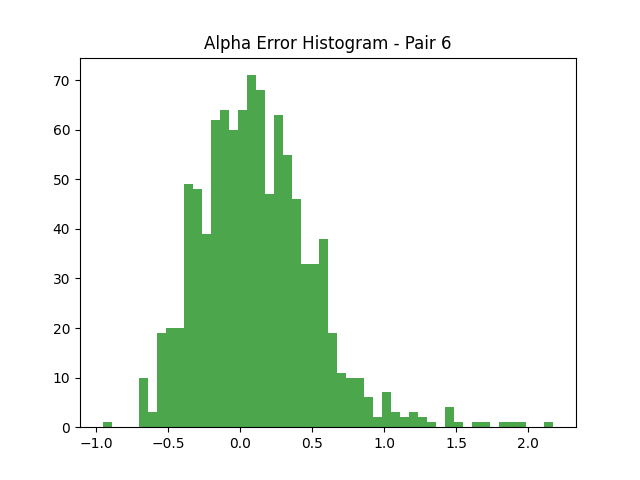}}
              \subfigure[]{\includegraphics[width=0.15\textwidth]{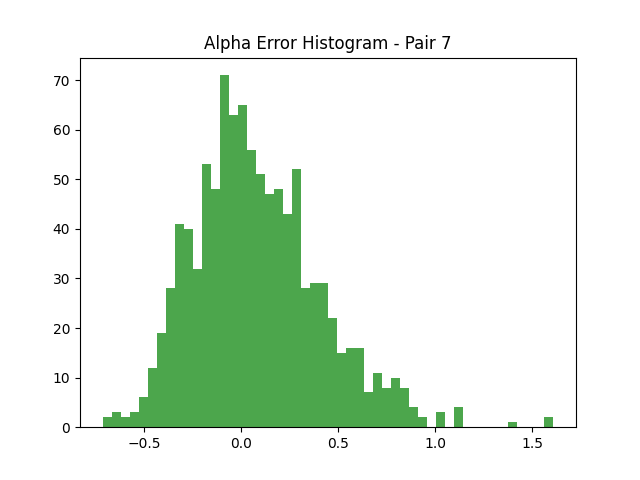}}
                \subfigure[]{\includegraphics[width=0.15\textwidth]{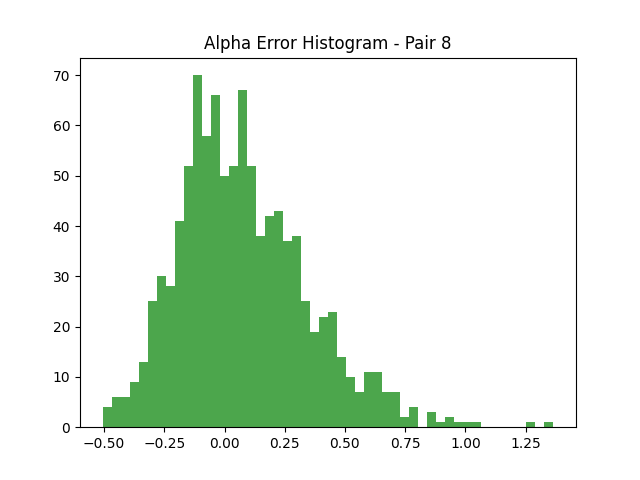}}
                  \subfigure[]{\includegraphics[width=0.15\textwidth]{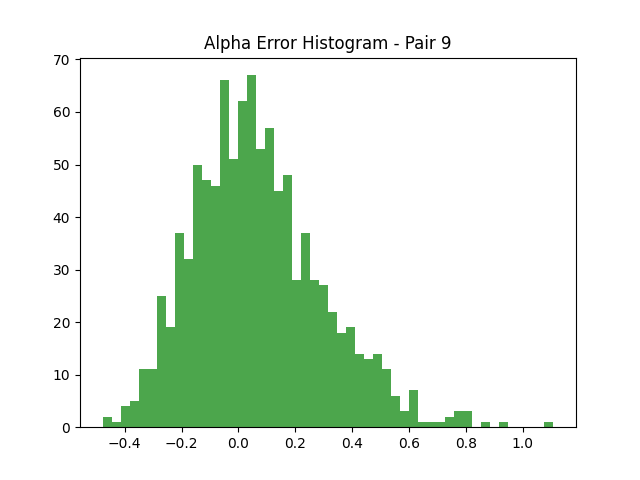}}
                    \subfigure[]{\includegraphics[width=0.15\textwidth]{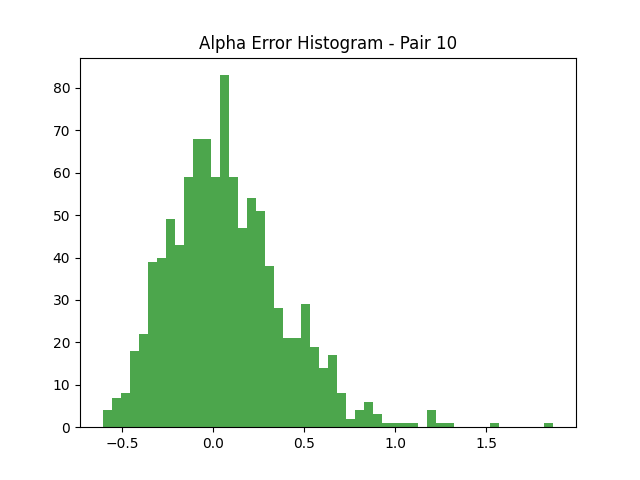}}
                    \caption{Dataset 3: PHC Shape}
    \label{PHC1HAl1}
\end{figure}

\begin{figure}[H]
    \centering
    \subfigure[]{\includegraphics[width=0.15\textwidth]{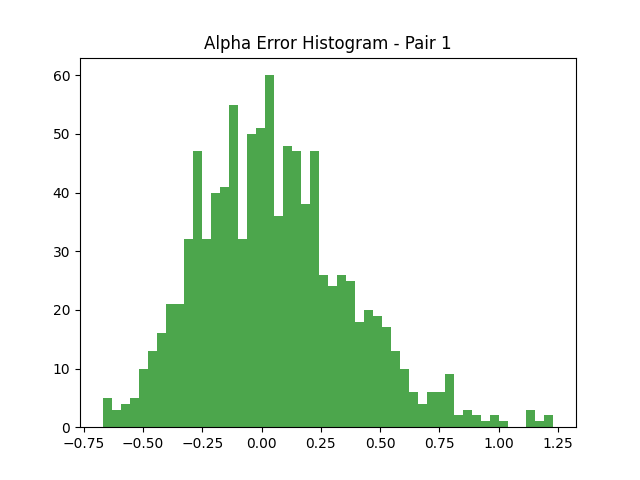}}
    \subfigure[]{\includegraphics[width=0.15\textwidth]{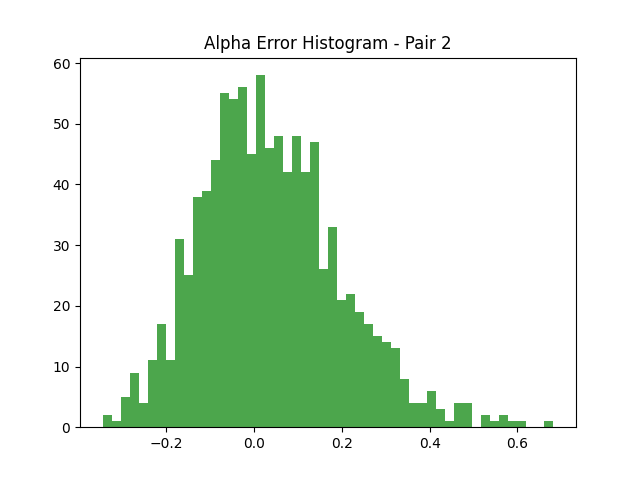}}
      \subfigure[]{\includegraphics[width=0.15\textwidth]{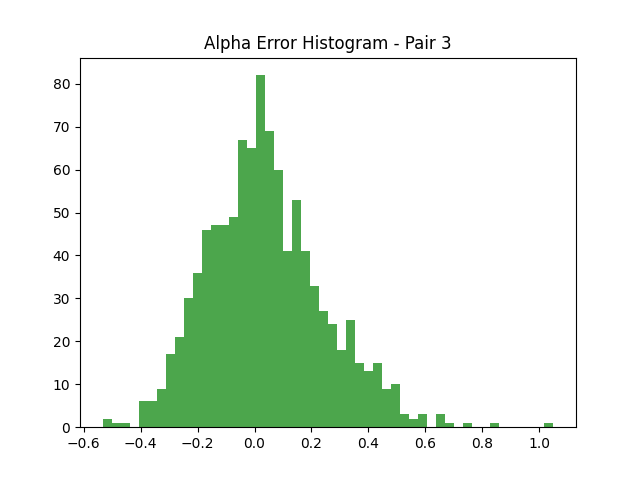}}
        \subfigure[]{\includegraphics[width=0.15\textwidth]{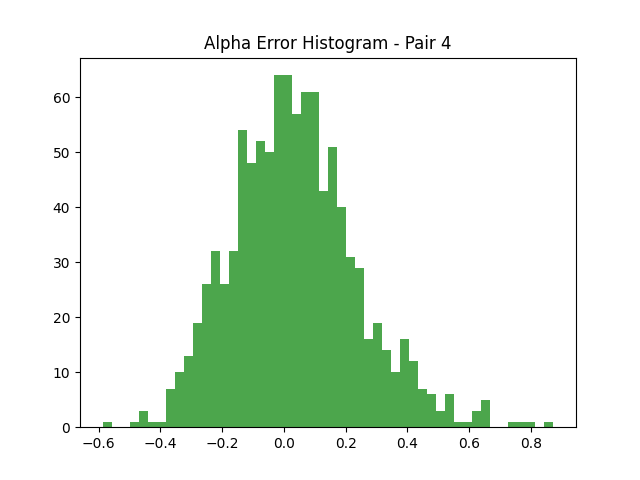}}
          \subfigure[]{\includegraphics[width=0.15\textwidth]{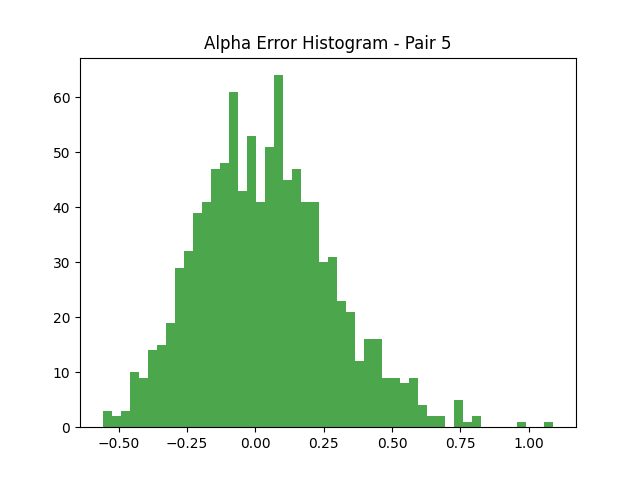}}
            \subfigure[]{\includegraphics[width=0.15\textwidth]{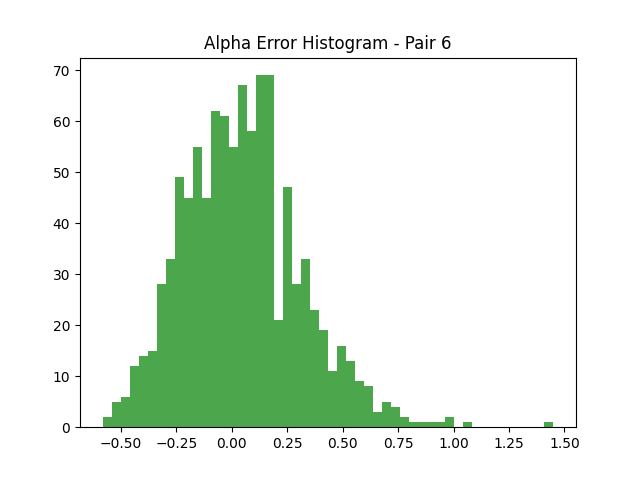}}
              \subfigure[]{\includegraphics[width=0.15\textwidth]{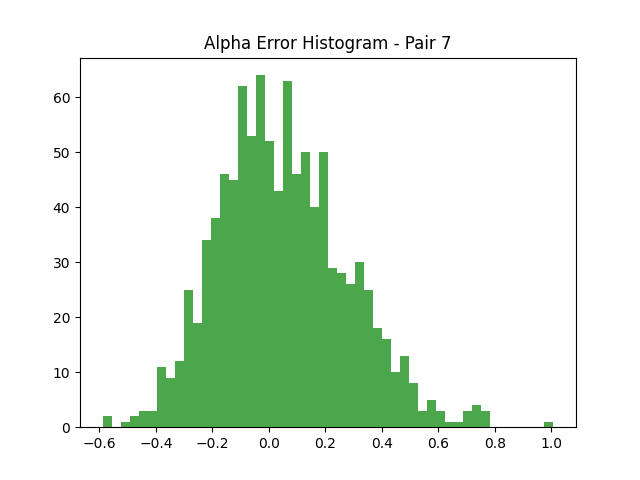}}
                \subfigure[]{\includegraphics[width=0.15\textwidth]{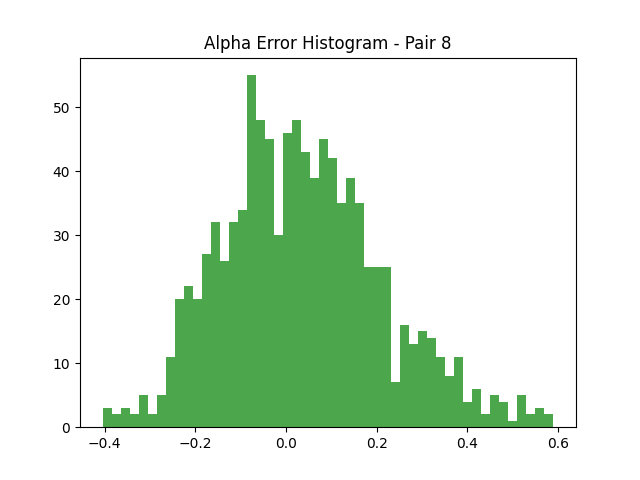}}
                  \subfigure[]{\includegraphics[width=0.15\textwidth]{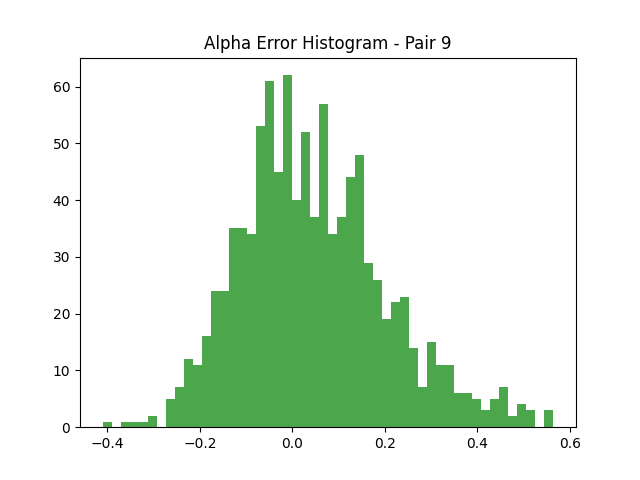}}
                    \subfigure[]{\includegraphics[width=0.15\textwidth]{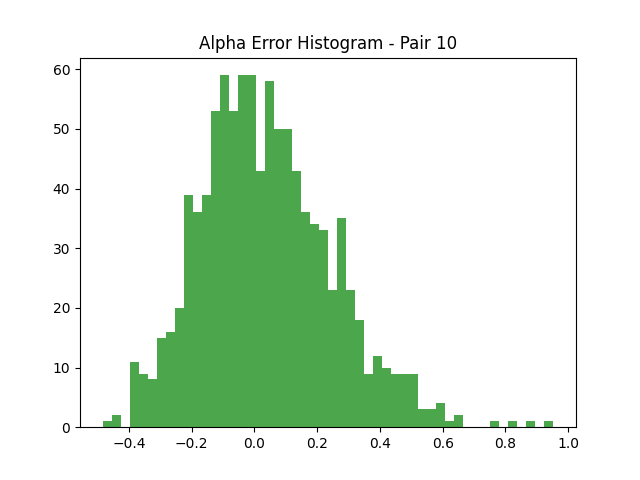}}
                    \caption{Dataset 9: PHC Shape}
    \label{PHC1HAl2}
\end{figure}

\begin{figure}[H]
    \centering
    \subfigure[]{\includegraphics[width=0.15\textwidth]{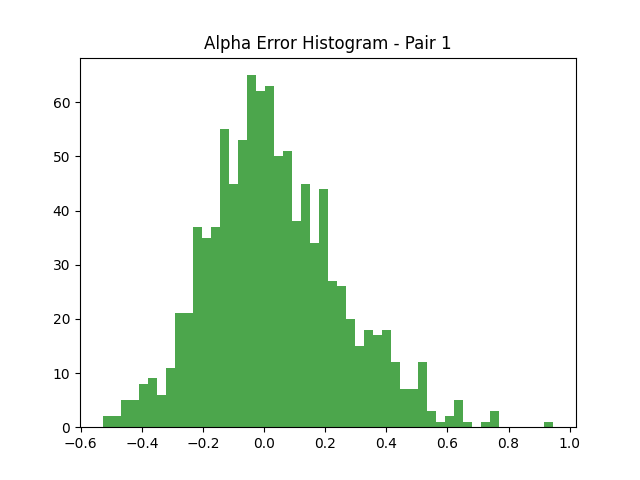}}
    \subfigure[]{\includegraphics[width=0.15\textwidth]{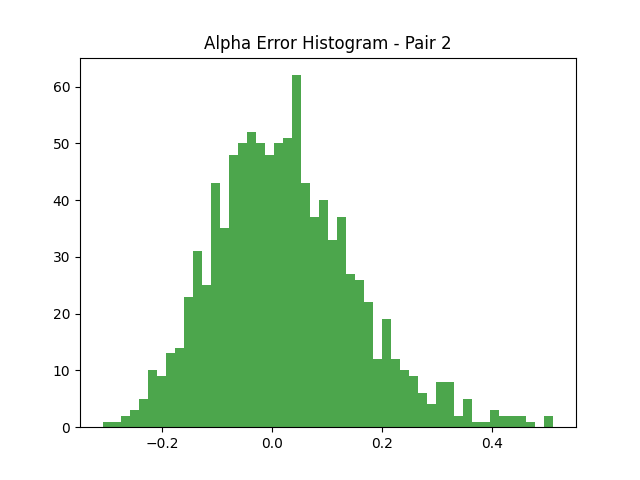}}
      \subfigure[]{\includegraphics[width=0.15\textwidth]{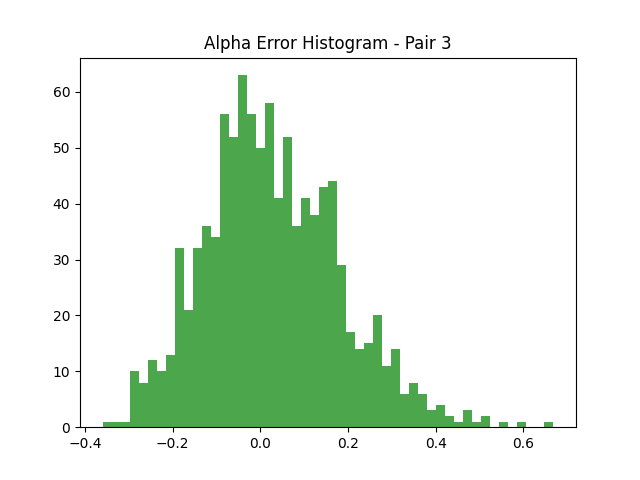}}
        \subfigure[]{\includegraphics[width=0.15\textwidth]{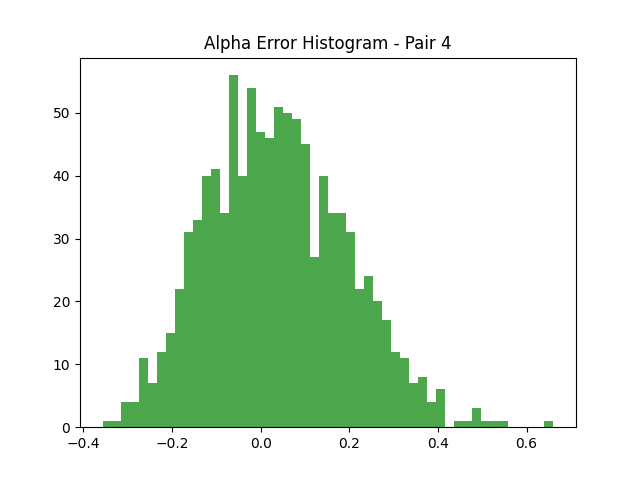}}
          \subfigure[]{\includegraphics[width=0.15\textwidth]{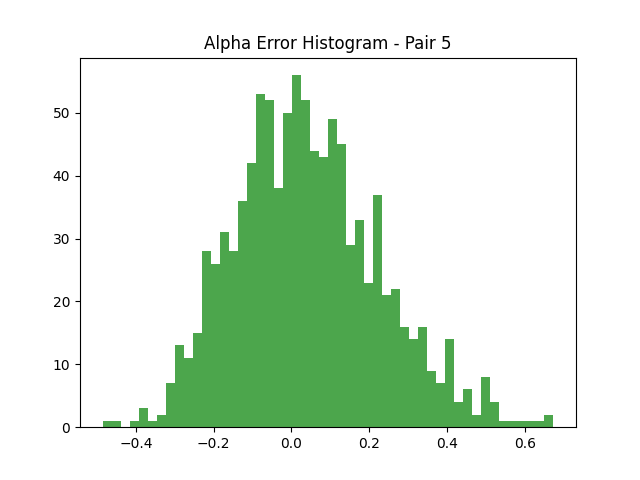}}
            \subfigure[]{\includegraphics[width=0.15\textwidth]{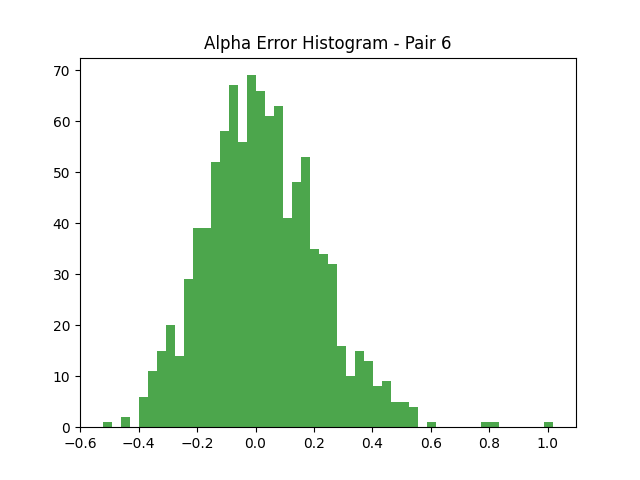}}
              \subfigure[]{\includegraphics[width=0.15\textwidth]{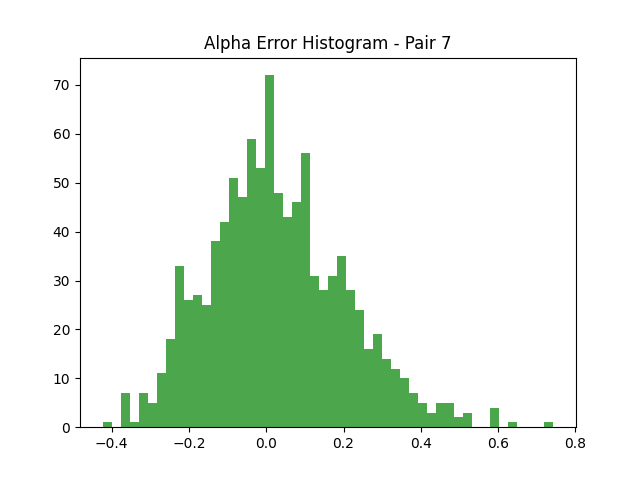}}
                \subfigure[]{\includegraphics[width=0.15\textwidth]{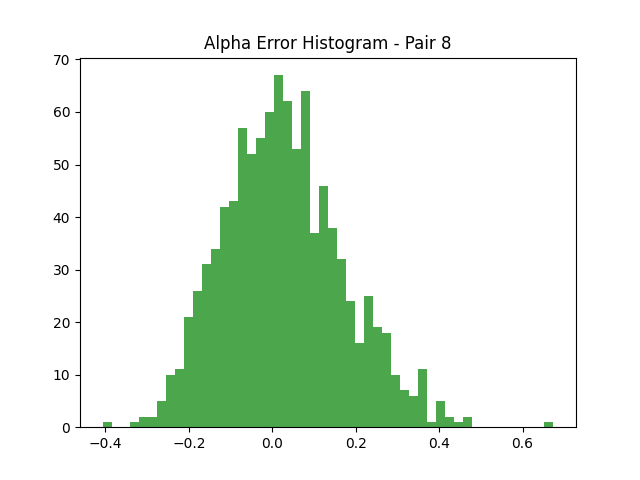}}
                  \subfigure[]{\includegraphics[width=0.15\textwidth]{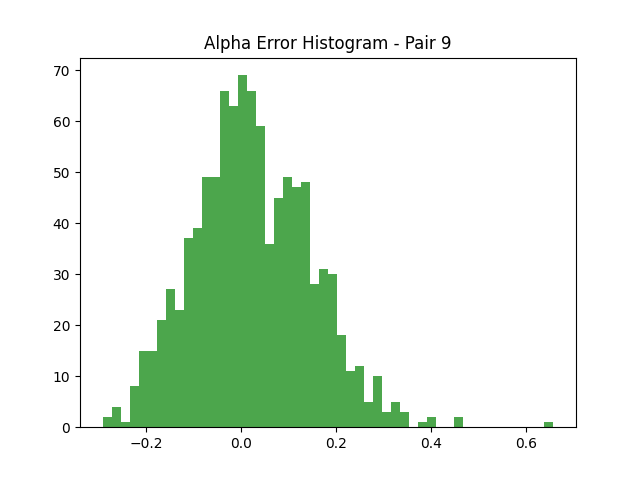}}
                    \subfigure[]{\includegraphics[width=0.15\textwidth]{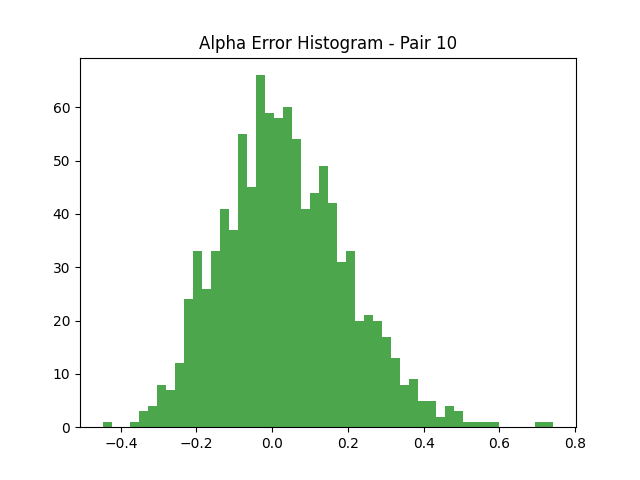}}
                    \caption{Dataset 15: PHC Shape}
    \label{PHC1HAl3}
\end{figure}

\begin{figure}[H]
    \centering
    \subfigure[]{\includegraphics[width=0.15\textwidth]{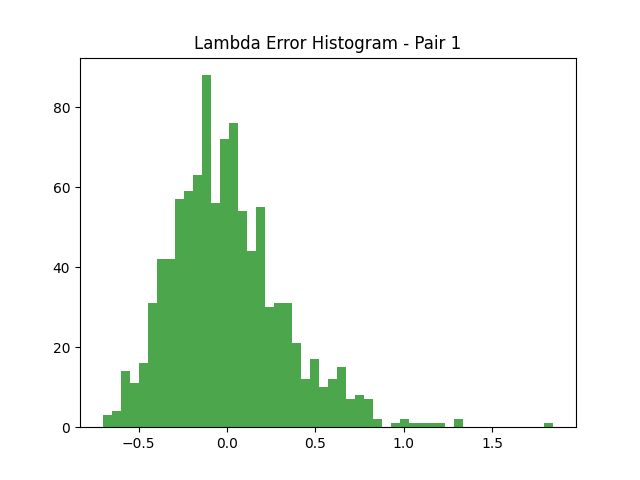}}
    \subfigure[]{\includegraphics[width=0.15\textwidth]{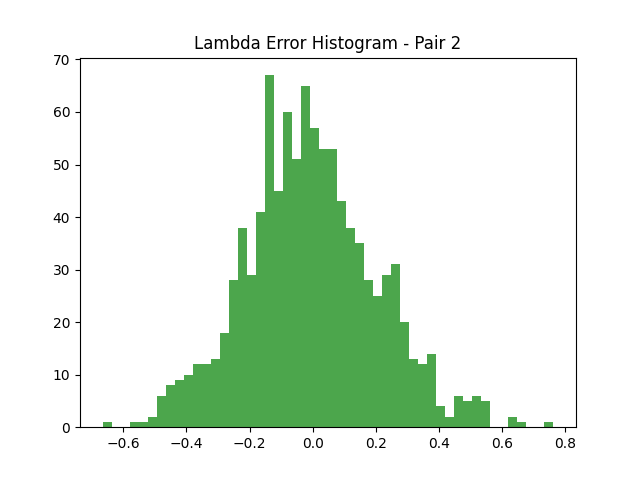}}
      \subfigure[]{\includegraphics[width=0.15\textwidth]{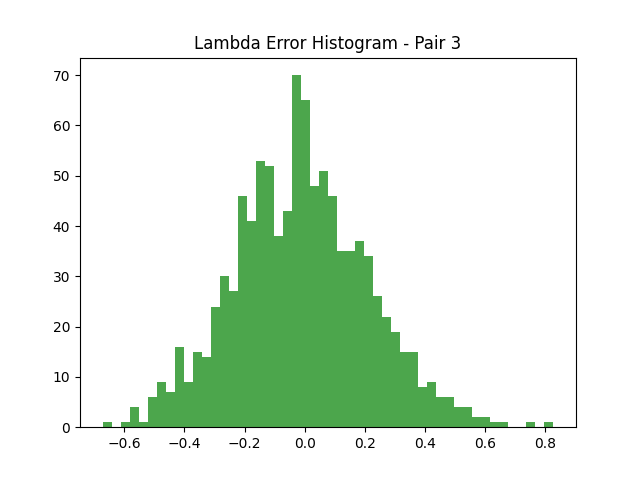}}
        \subfigure[]{\includegraphics[width=0.15\textwidth]{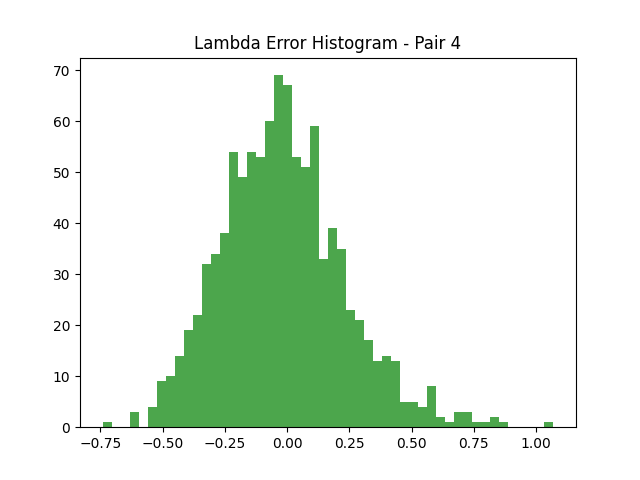}}
          \subfigure[]{\includegraphics[width=0.15\textwidth]{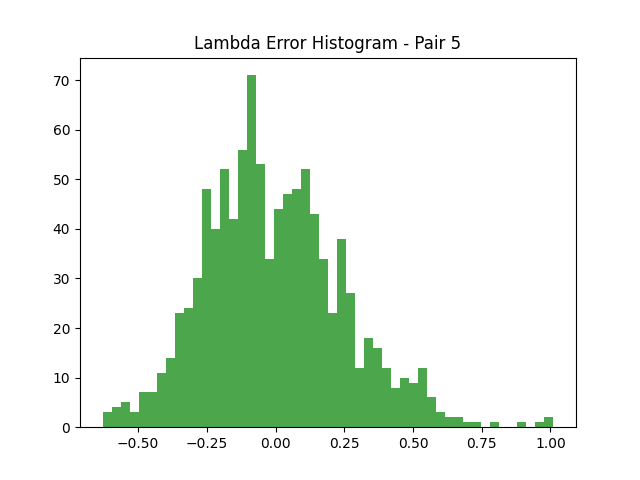}}
            \subfigure[]{\includegraphics[width=0.15\textwidth]{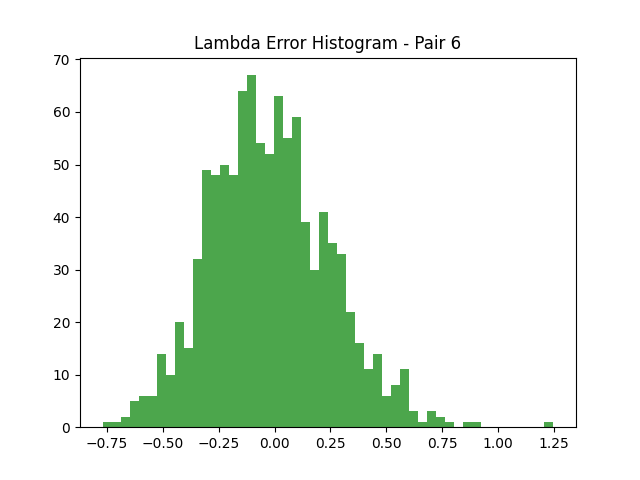}}
              \subfigure[]{\includegraphics[width=0.15\textwidth]{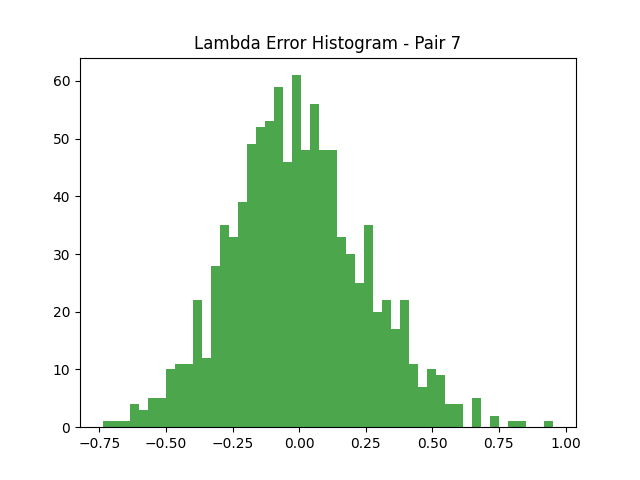}}
                \subfigure[]{\includegraphics[width=0.15\textwidth]{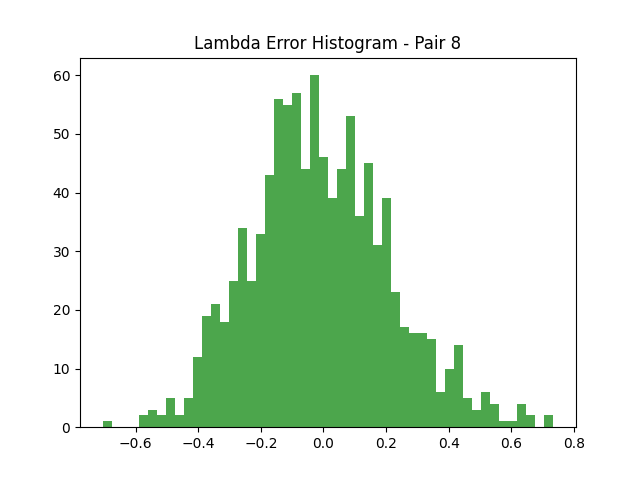}}
                  \subfigure[]{\includegraphics[width=0.15\textwidth]{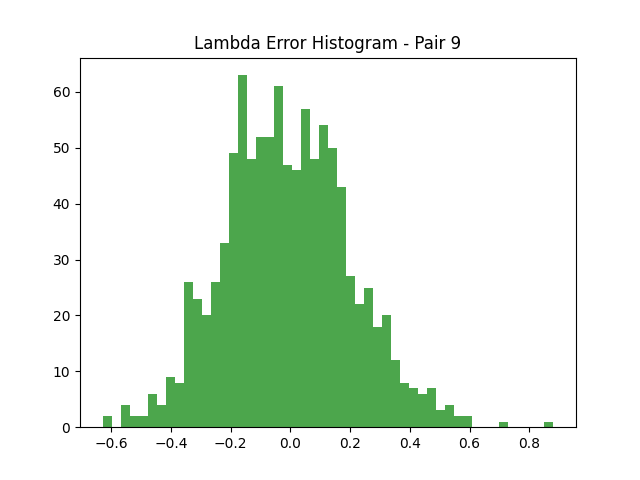}}
                    \subfigure[]{\includegraphics[width=0.15\textwidth]{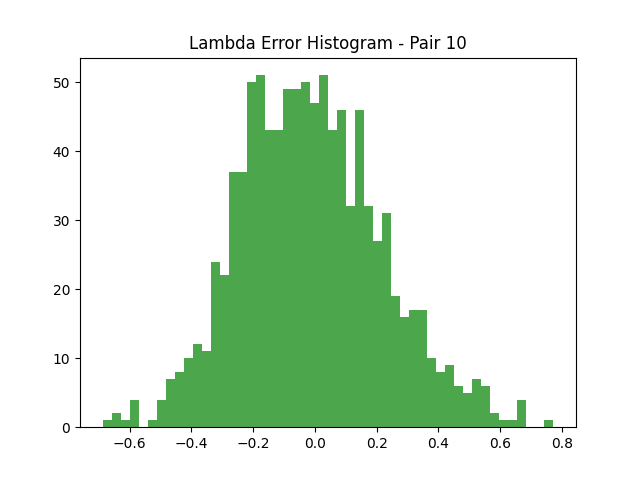}}
                    \caption{Dataset 3: PHC Scale}
    \label{PHC1HLam1}
\end{figure}

\begin{figure}[H]
    \centering
    \subfigure[]{\includegraphics[width=0.15\textwidth]{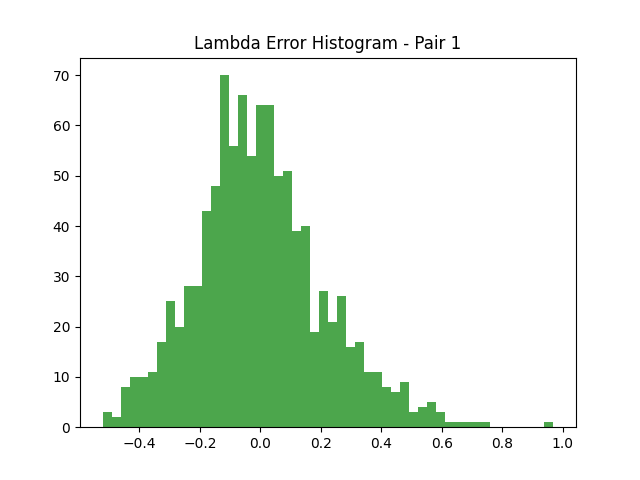}}
    \subfigure[]{\includegraphics[width=0.15\textwidth]{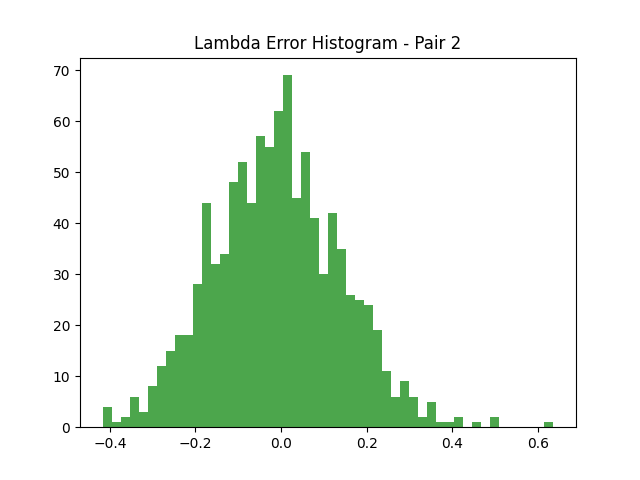}}
      \subfigure[]{\includegraphics[width=0.15\textwidth]{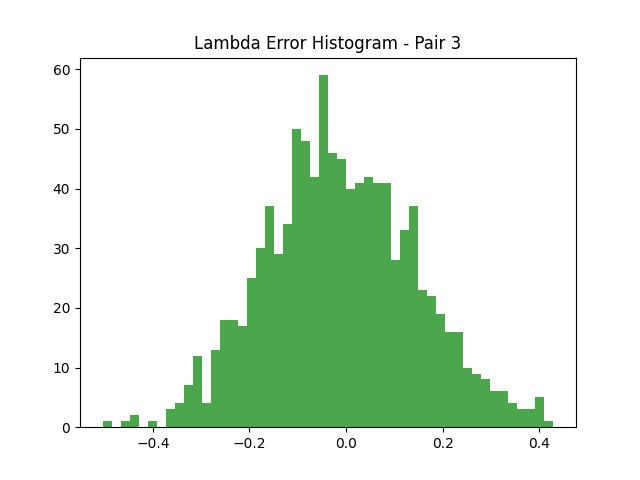}}
        \subfigure[]{\includegraphics[width=0.15\textwidth]{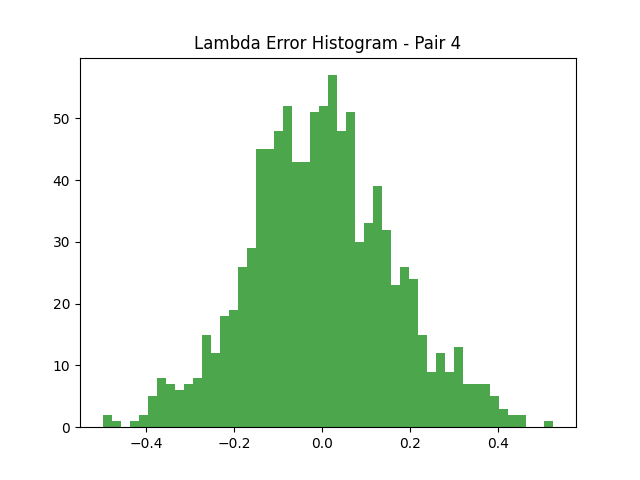}}
          \subfigure[]{\includegraphics[width=0.15\textwidth]{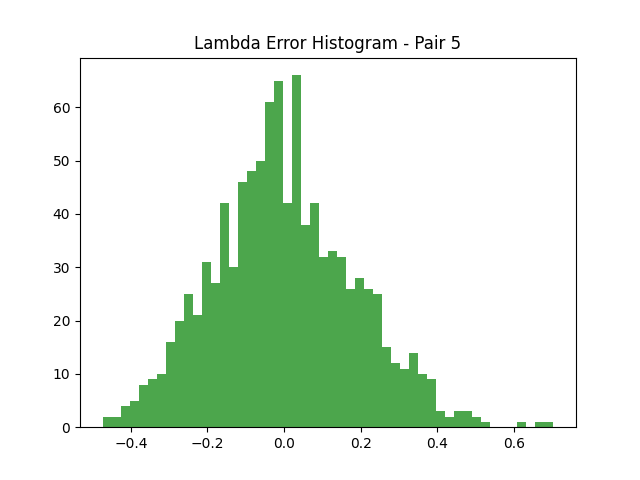}}
            \subfigure[]{\includegraphics[width=0.15\textwidth]{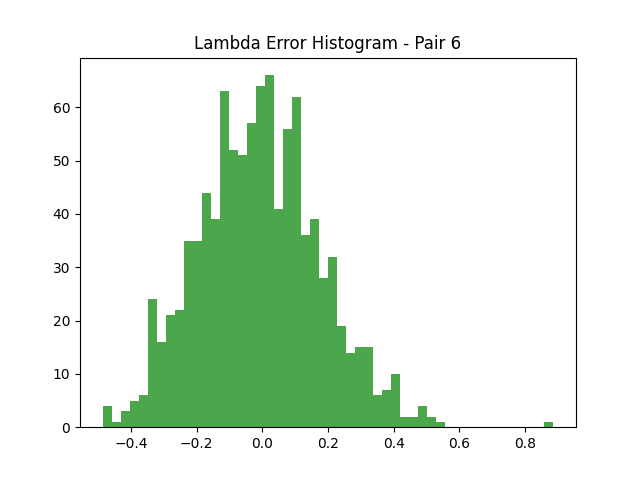}}
              \subfigure[]{\includegraphics[width=0.15\textwidth]{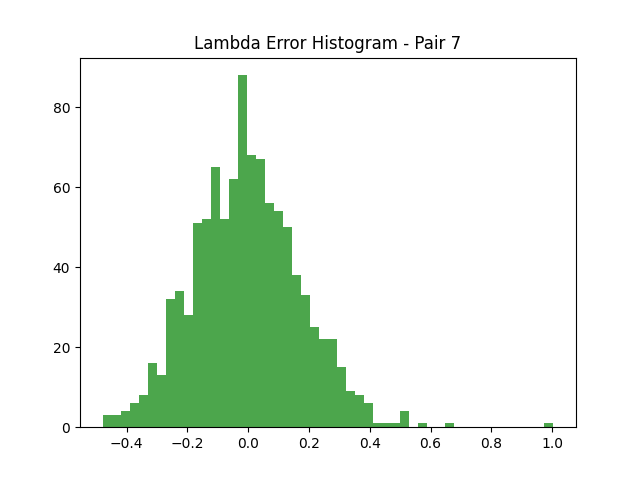}}
                \subfigure[]{\includegraphics[width=0.15\textwidth]{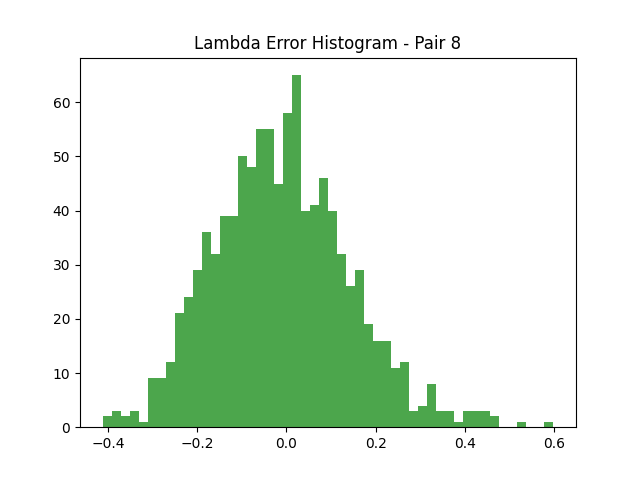}}
                  \subfigure[]{\includegraphics[width=0.15\textwidth]{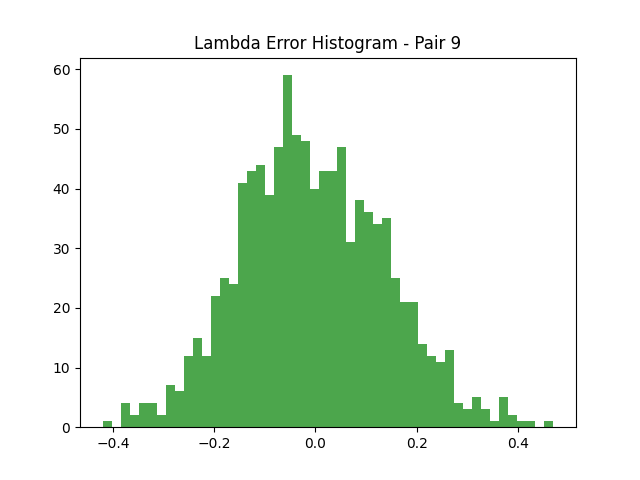}}
                    \subfigure[]{\includegraphics[width=0.15\textwidth]{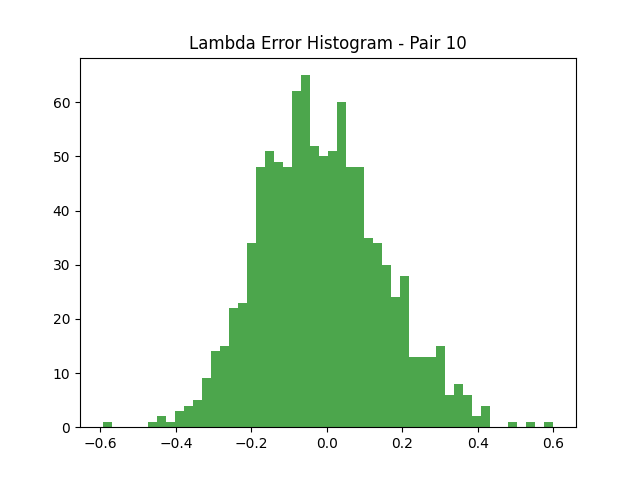}}
                    \caption{Dataset 9: PHC Scale}
    \label{PHC1HLam2}
\end{figure}

\begin{figure}[H]
    \centering
    \subfigure[]{\includegraphics[width=0.15\textwidth]{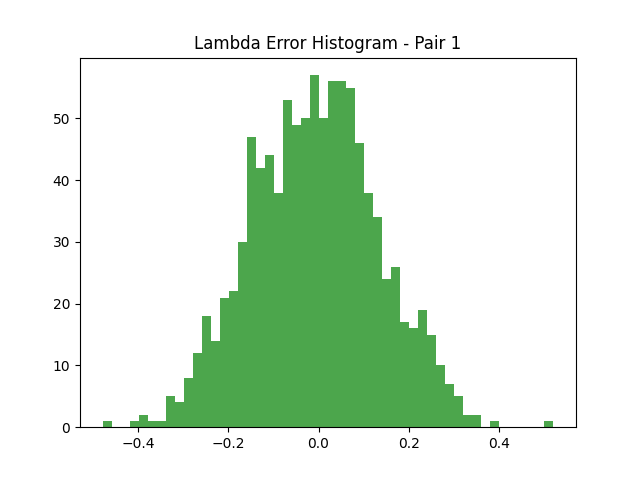}}
    \subfigure[]{\includegraphics[width=0.15\textwidth]{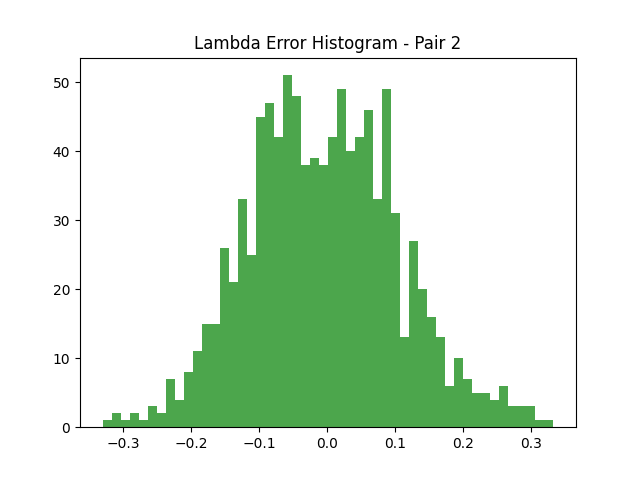}}
      \subfigure[]{\includegraphics[width=0.15\textwidth]{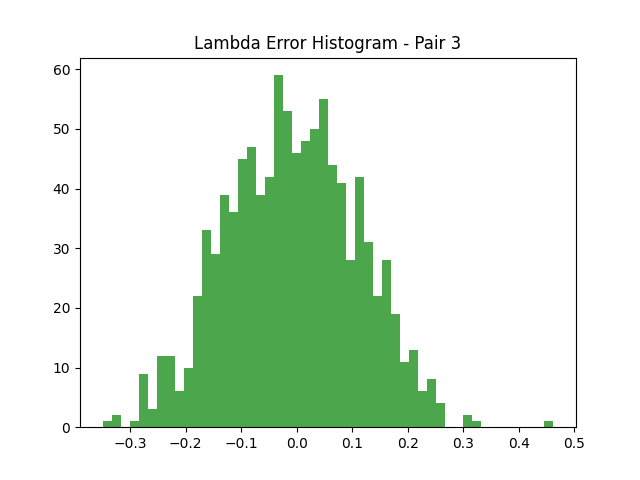}}
        \subfigure[]{\includegraphics[width=0.15\textwidth]{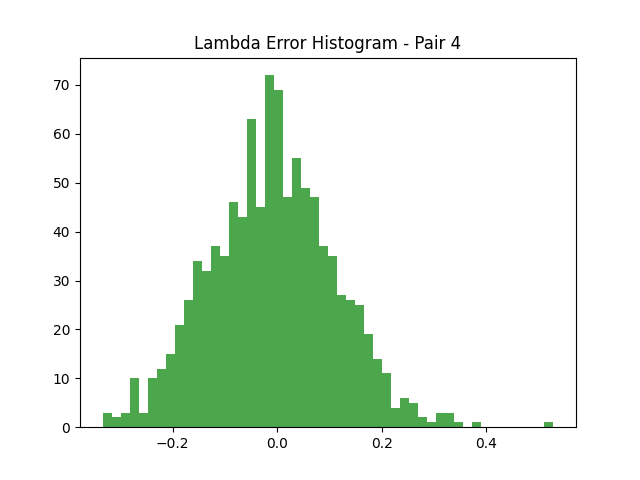}}
          \subfigure[]{\includegraphics[width=0.15\textwidth]{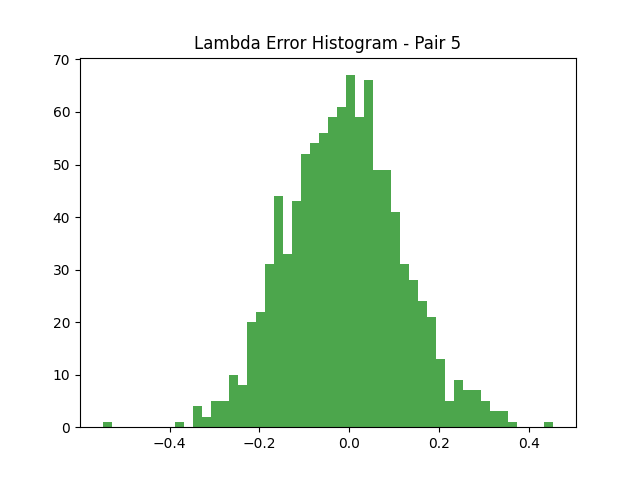}}
            \subfigure[]{\includegraphics[width=0.15\textwidth]{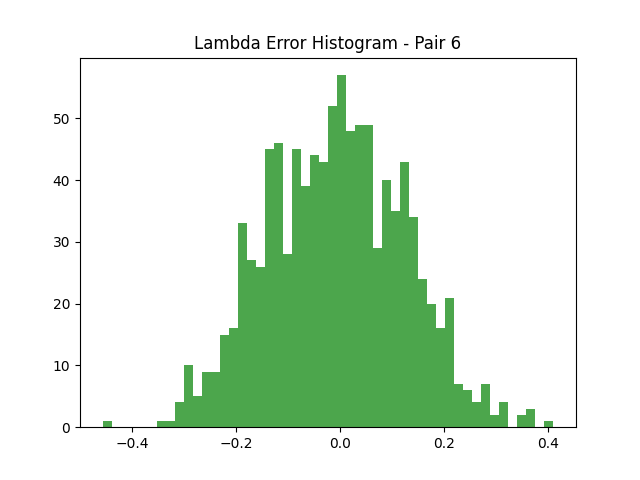}}
              \subfigure[]{\includegraphics[width=0.15\textwidth]{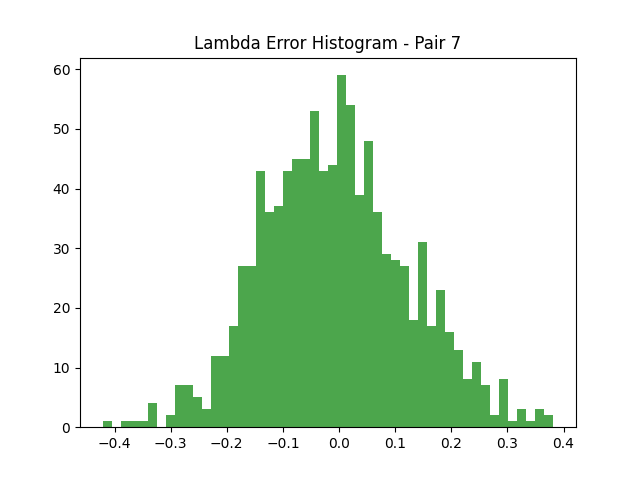}}
                \subfigure[]{\includegraphics[width=0.15\textwidth]{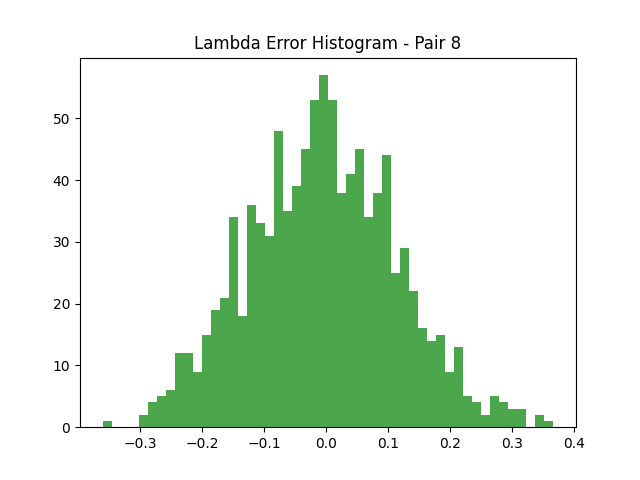}}
                  \subfigure[]{\includegraphics[width=0.15\textwidth]{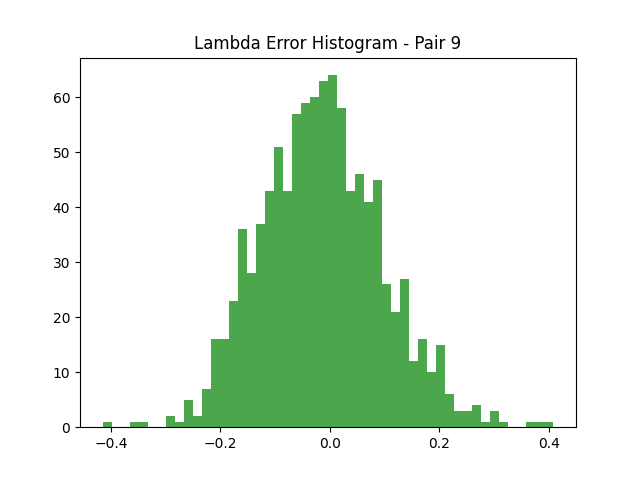}}
                    \subfigure[]{\includegraphics[width=0.15\textwidth]{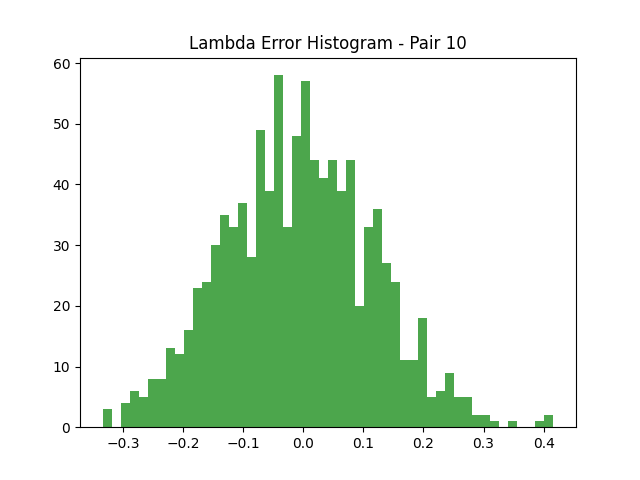}}
                    \caption{Dataset 15: PHC Scale}
    \label{PHC1HLam15}
\end{figure}

\begin{figure}[H]
    \centering
    \subfigure[]{\includegraphics[width=0.15\textwidth]{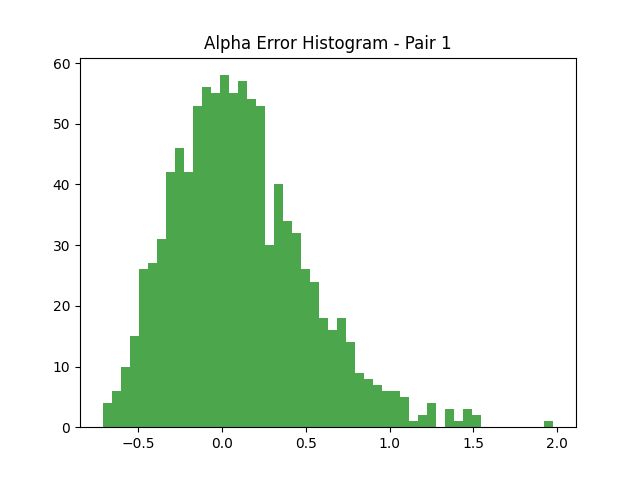}}
    \subfigure[]{\includegraphics[width=0.15\textwidth]{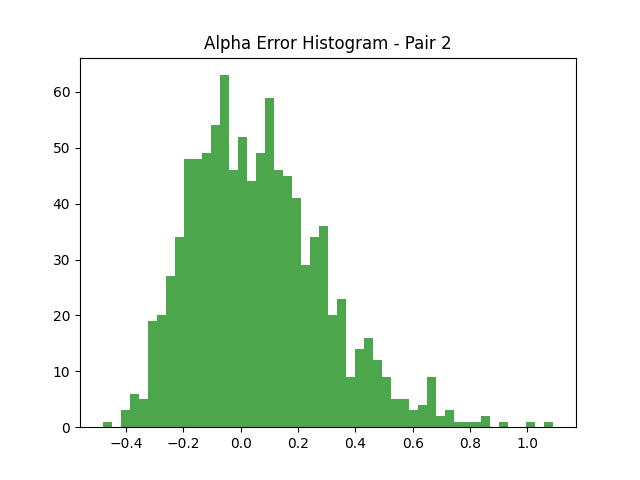}}
      \subfigure[]{\includegraphics[width=0.15\textwidth]{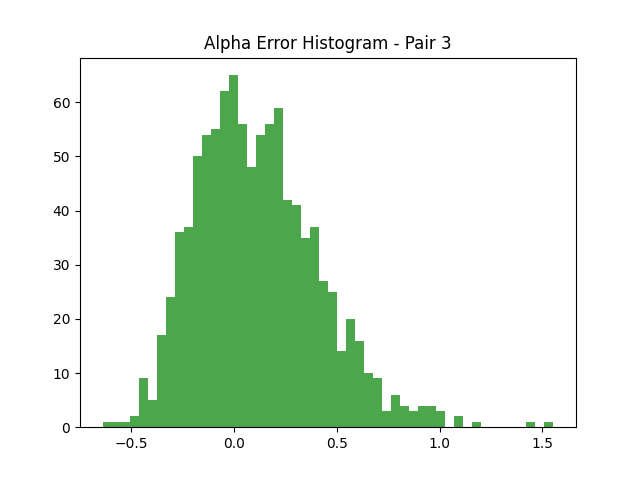}}
        \subfigure[]{\includegraphics[width=0.15\textwidth]{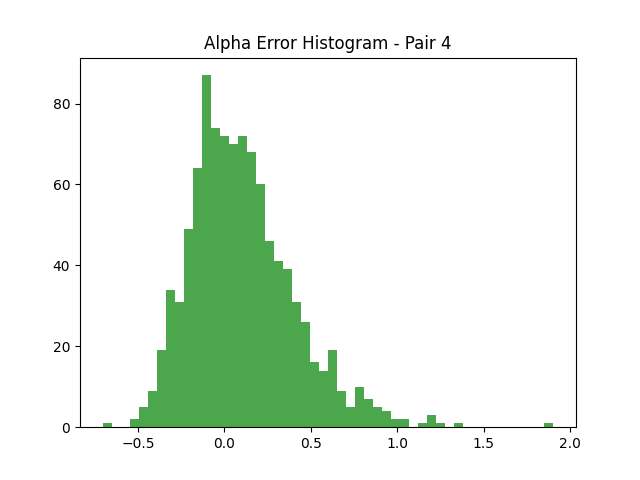}}
          \subfigure[]{\includegraphics[width=0.15\textwidth]{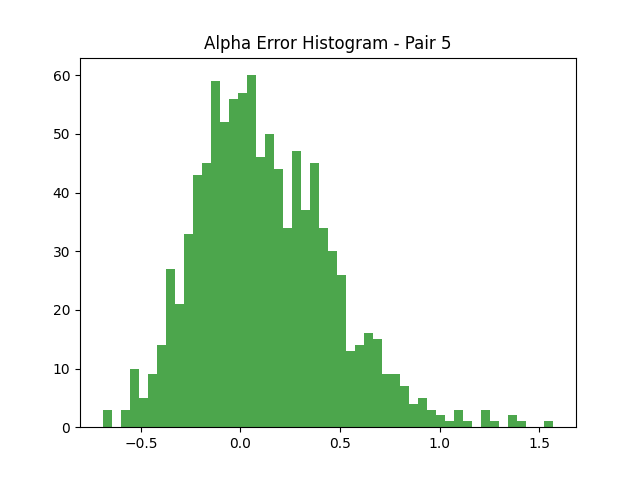}}
            \subfigure[]{\includegraphics[width=0.15\textwidth]{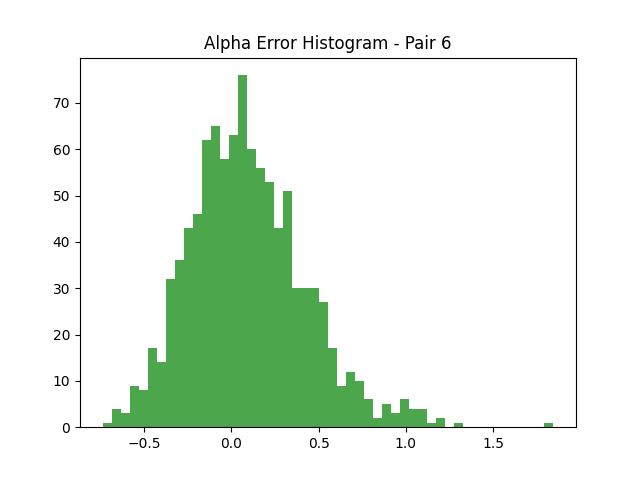}}
              \subfigure[]{\includegraphics[width=0.15\textwidth]{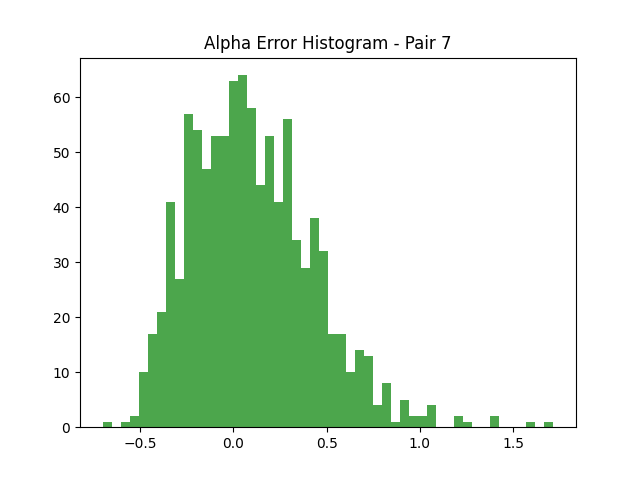}}
                \subfigure[]{\includegraphics[width=0.15\textwidth]{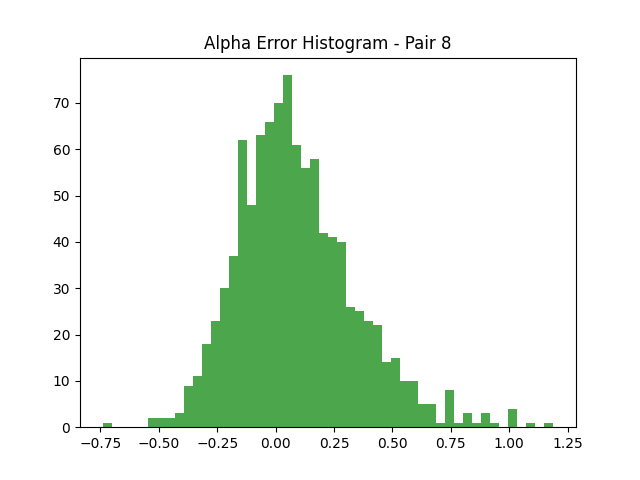}}
                  \subfigure[]{\includegraphics[width=0.15\textwidth]{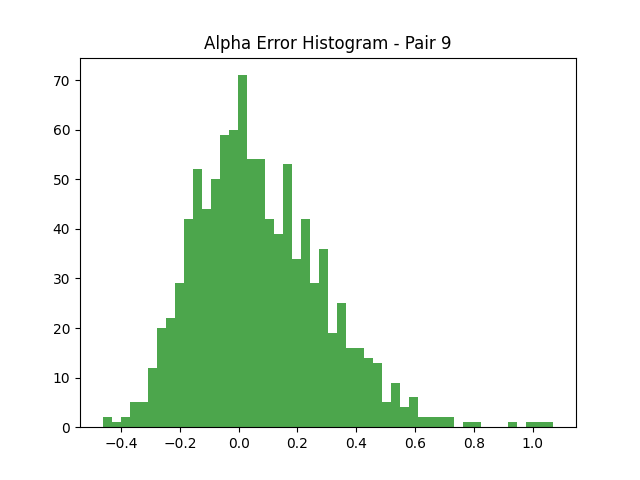}}
                    \subfigure[]{\includegraphics[width=0.15\textwidth]{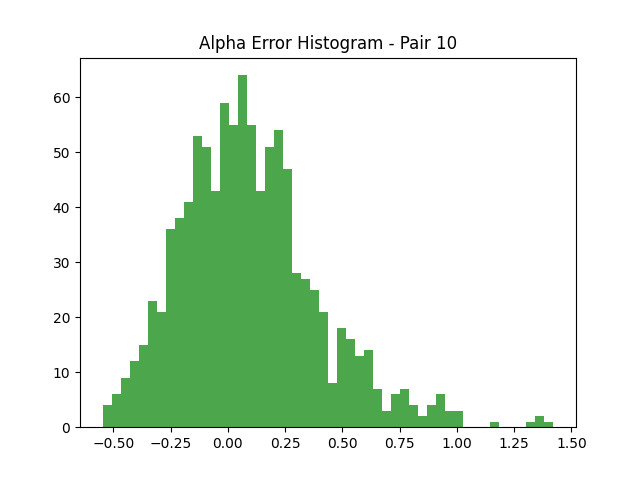}}
                    \caption{Dataset 3: APHC Shape}
    \label{APHC1HAl3}
\end{figure}

\begin{figure}[H]
    \centering
    \subfigure[]{\includegraphics[width=0.15\textwidth]{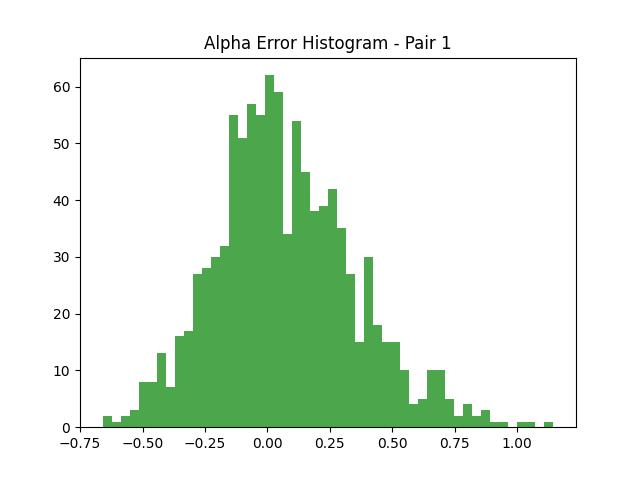}}
    \subfigure[]{\includegraphics[width=0.15\textwidth]{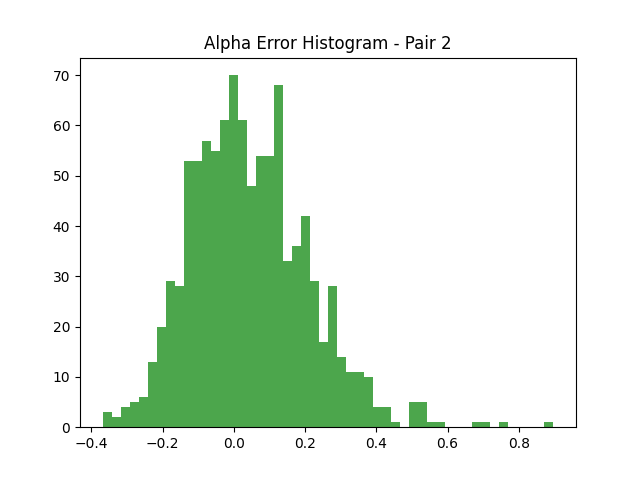}}
      \subfigure[]{\includegraphics[width=0.15\textwidth]{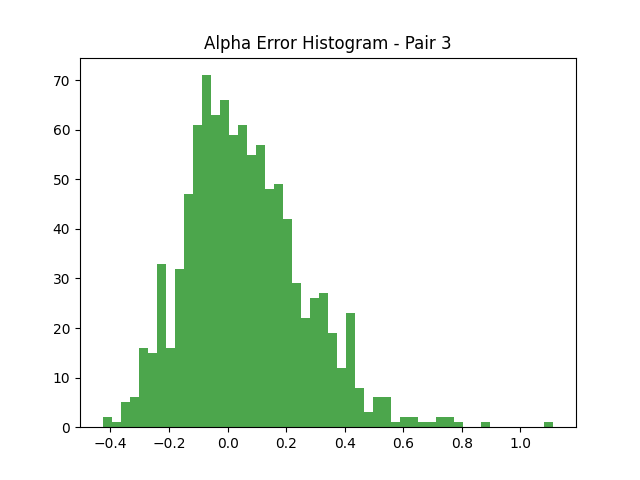}}
        \subfigure[]{\includegraphics[width=0.15\textwidth]{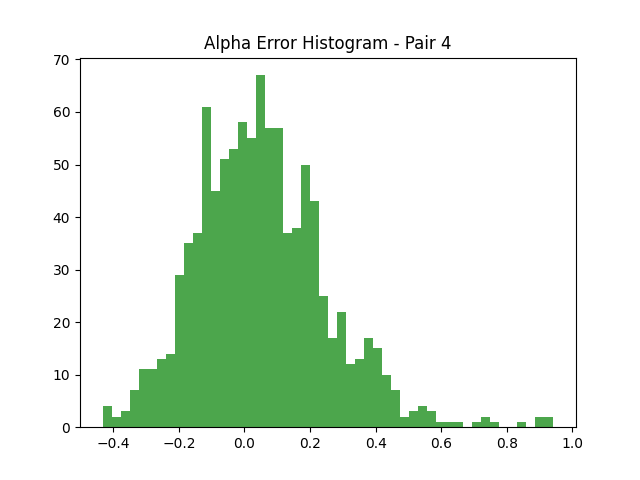}}
          \subfigure[]{\includegraphics[width=0.15\textwidth]{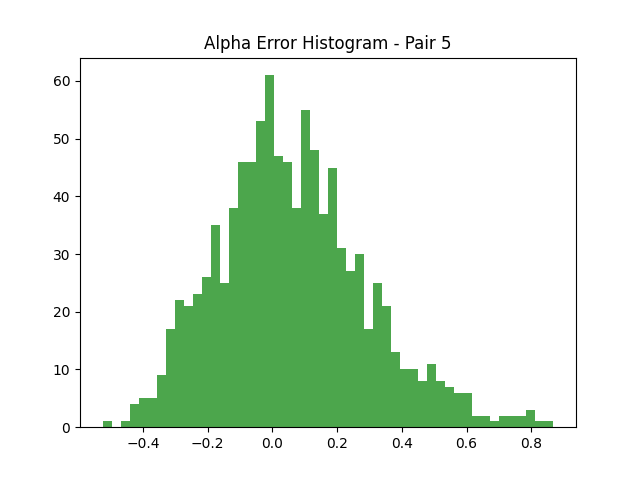}}
            \subfigure[]{\includegraphics[width=0.15\textwidth]{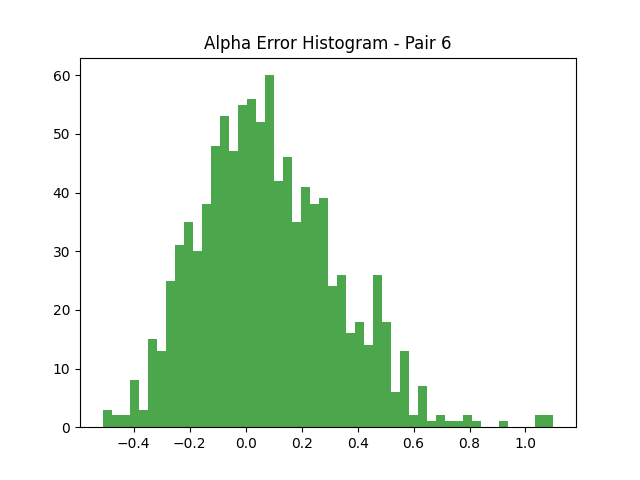}}
              \subfigure[]{\includegraphics[width=0.15\textwidth]{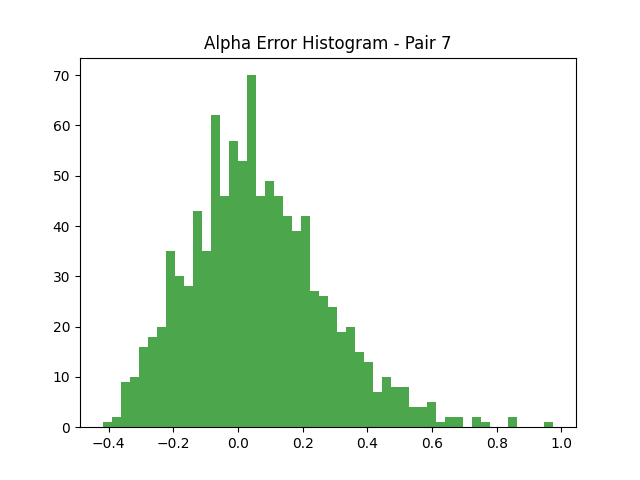}}
                \subfigure[]{\includegraphics[width=0.15\textwidth]{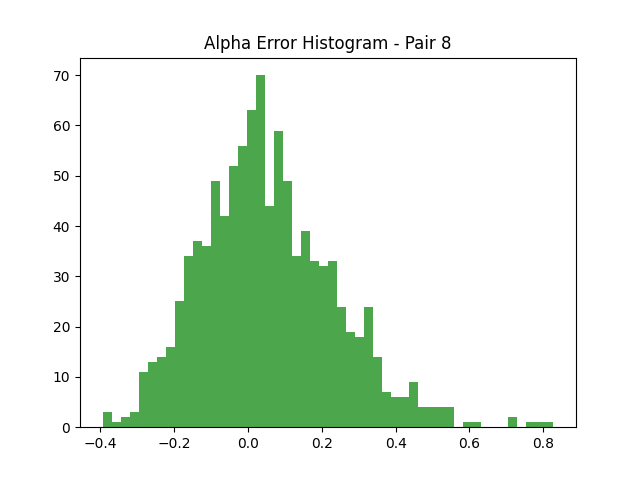}}
                  \subfigure[]{\includegraphics[width=0.15\textwidth]{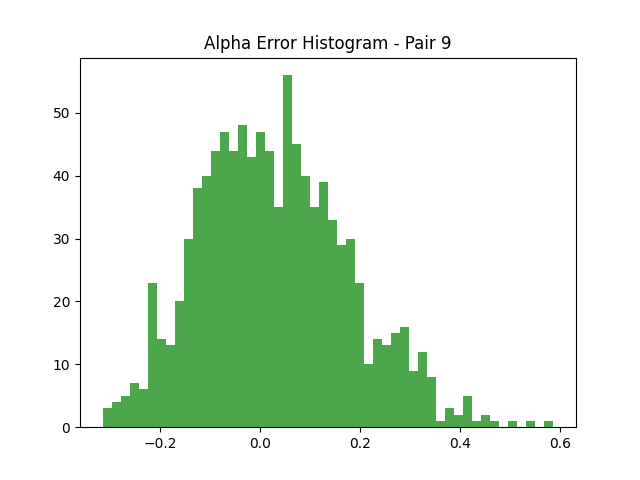}}
                    \subfigure[]{\includegraphics[width=0.15\textwidth]{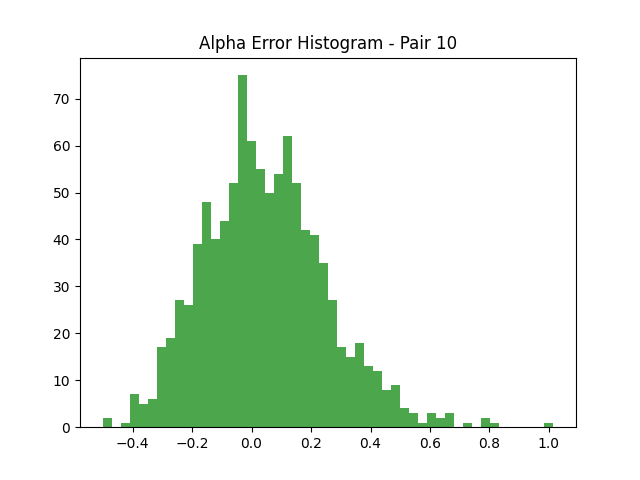}}
                    \caption{Dataset 9: APHC Shape}
    \label{APHC1HAl9}
\end{figure}

\begin{figure}[H]
    \centering
    \subfigure[]{\includegraphics[width=0.15\textwidth]{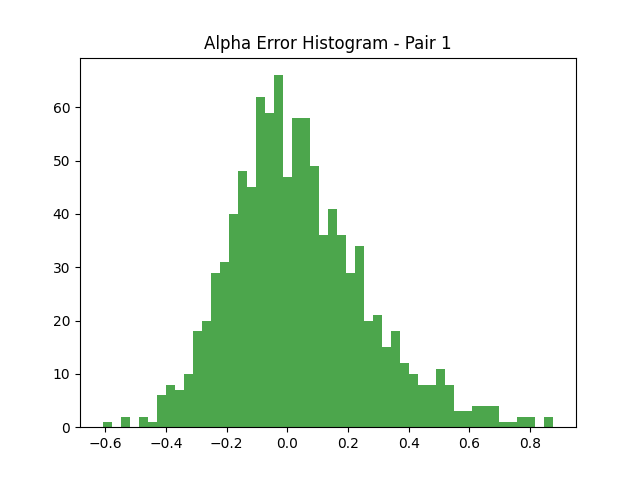}}
    \subfigure[]{\includegraphics[width=0.15\textwidth]{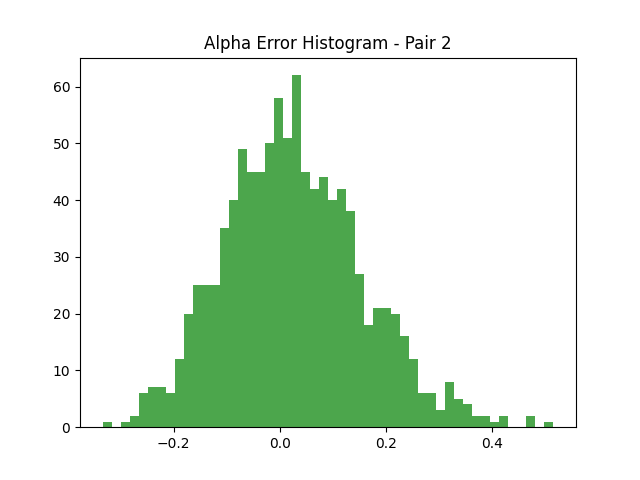}}
      \subfigure[]{\includegraphics[width=0.15\textwidth]{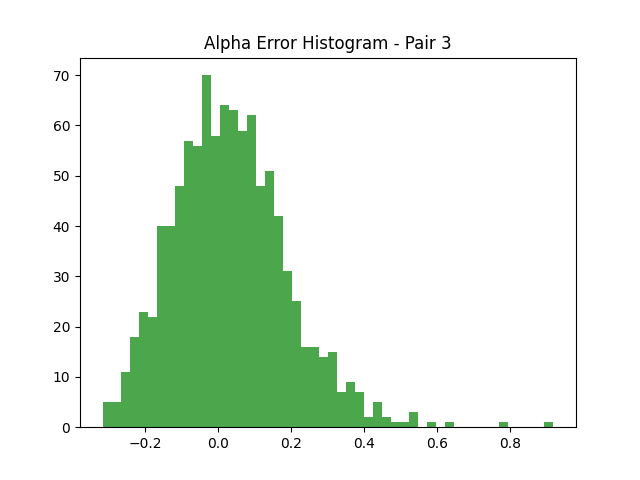}}
        \subfigure[]{\includegraphics[width=0.15\textwidth]{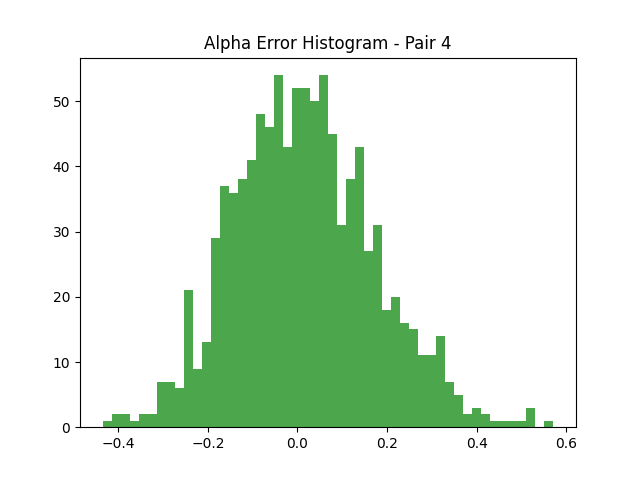}}
          \subfigure[]{\includegraphics[width=0.15\textwidth]{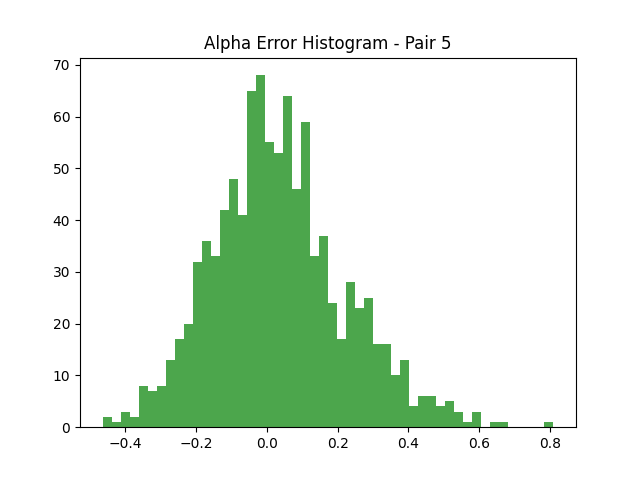}}
            \subfigure[]{\includegraphics[width=0.15\textwidth]{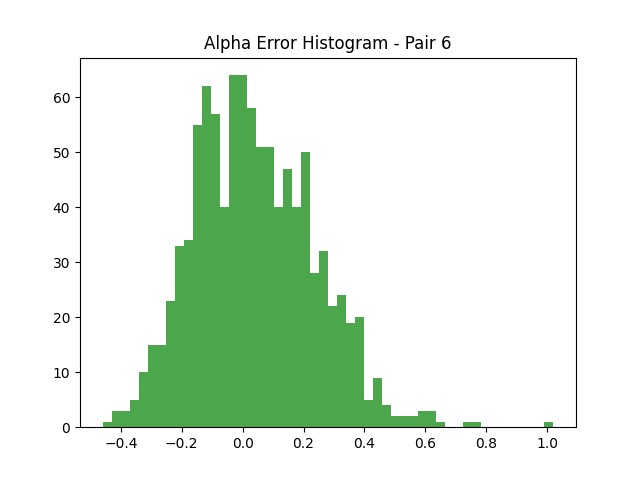}}
              \subfigure[]{\includegraphics[width=0.15\textwidth]{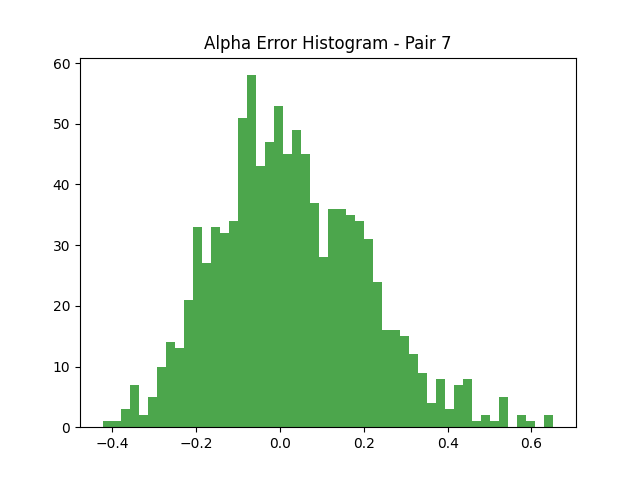}}
                \subfigure[]{\includegraphics[width=0.15\textwidth]{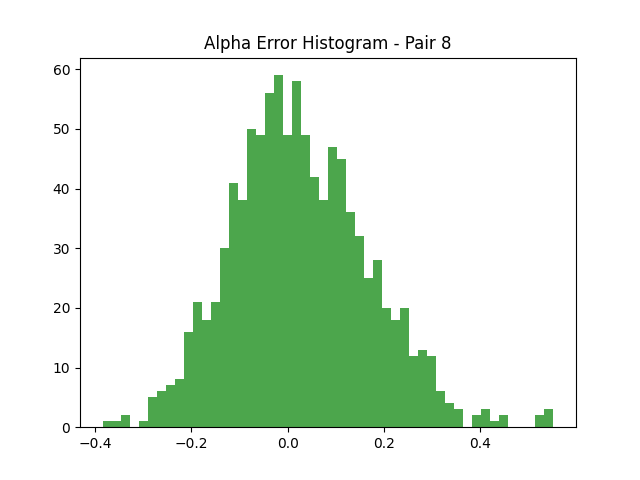}}
                  \subfigure[]{\includegraphics[width=0.15\textwidth]{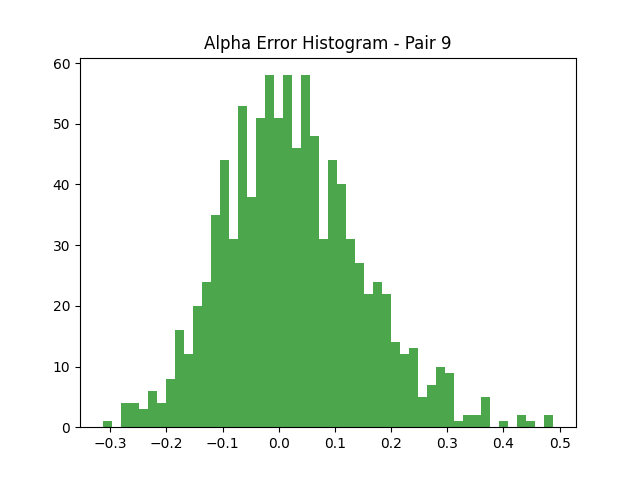}}
                    \subfigure[]{\includegraphics[width=0.15\textwidth]{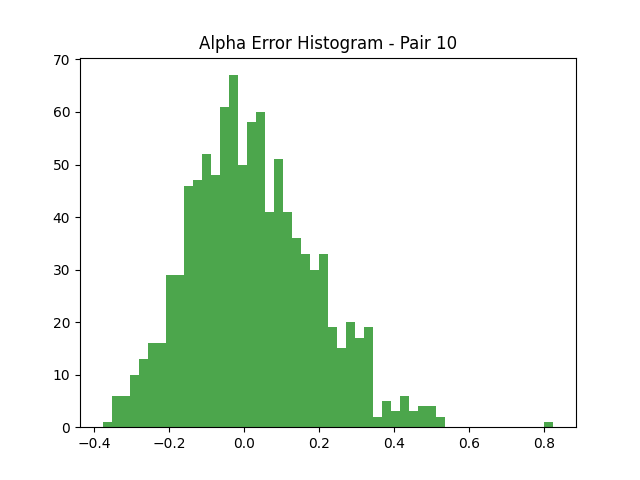}}
                    \caption{Dataset 15: APHC Shape}
    \label{APHC1HAl15}
\end{figure}

\begin{figure}[H]
    \centering
    \subfigure[]{\includegraphics[width=0.15\textwidth]{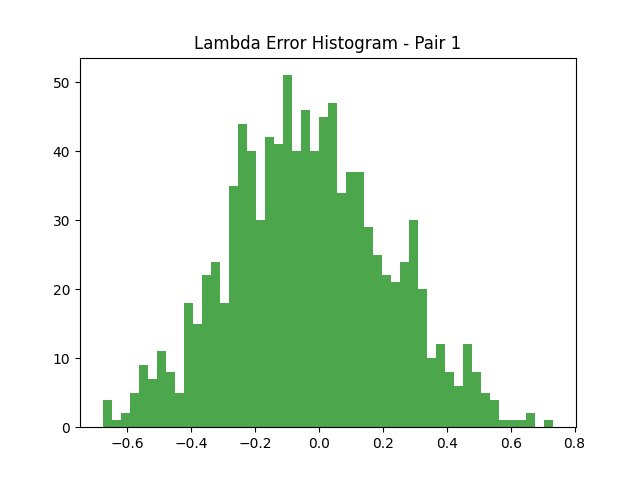}}
    \subfigure[]{\includegraphics[width=0.15\textwidth]{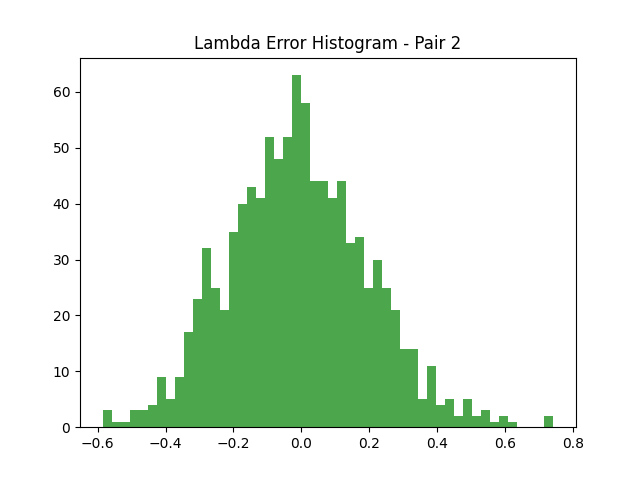}}
      \subfigure[]{\includegraphics[width=0.15\textwidth]{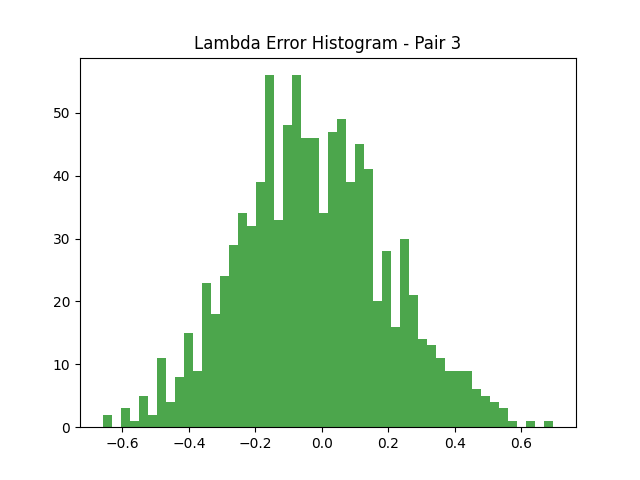}}
        \subfigure[]{\includegraphics[width=0.15\textwidth]{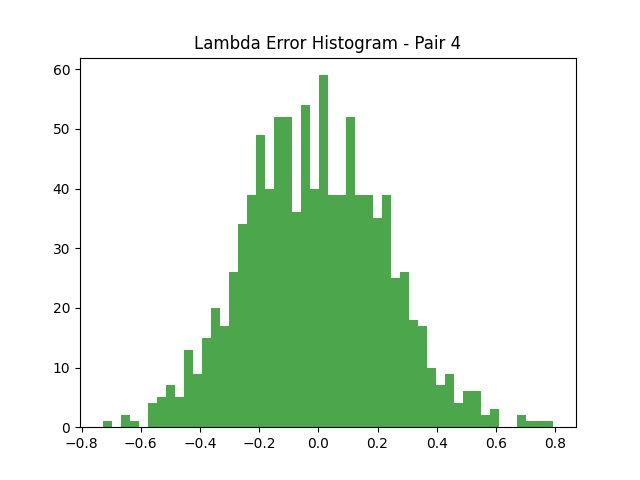}}
          \subfigure[]{\includegraphics[width=0.15\textwidth]{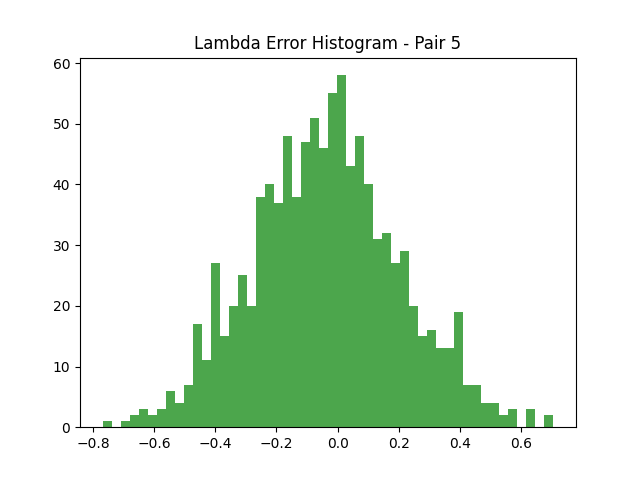}}
            \subfigure[]{\includegraphics[width=0.15\textwidth]{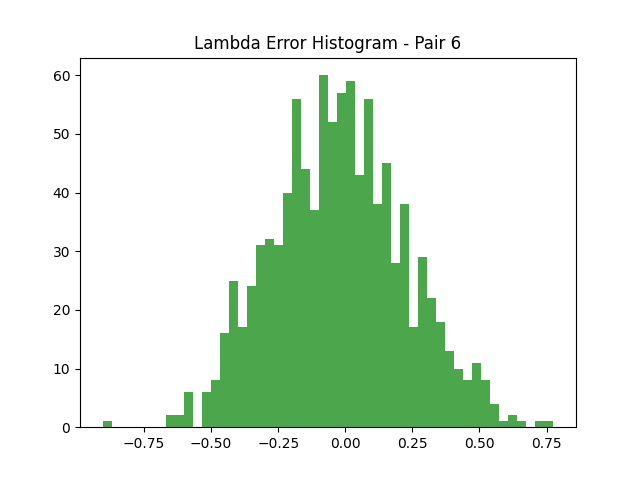}}
              \subfigure[]{\includegraphics[width=0.15\textwidth]{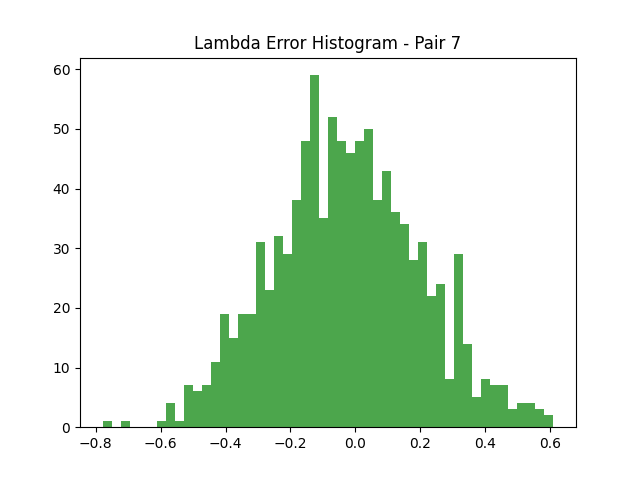}}
                \subfigure[]{\includegraphics[width=0.15\textwidth]{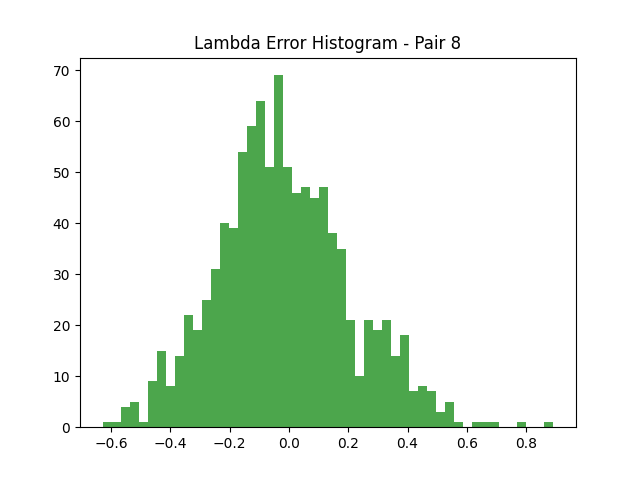}}
                  \subfigure[]{\includegraphics[width=0.15\textwidth]{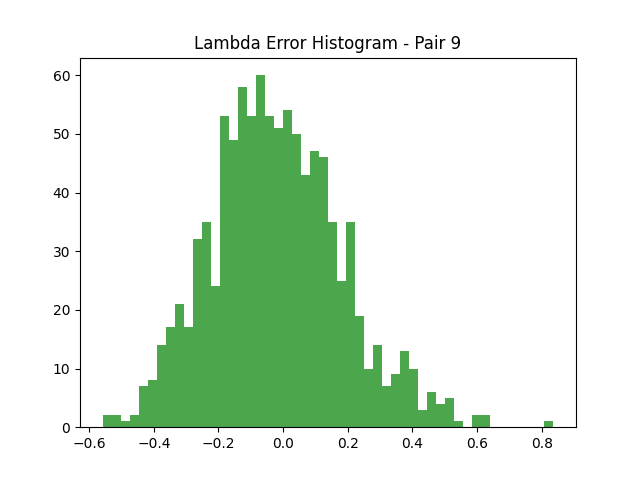}}
                    \subfigure[]{\includegraphics[width=0.15\textwidth]{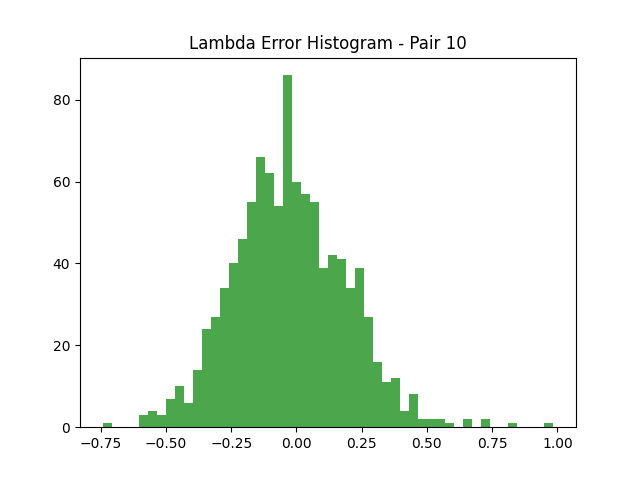}}
                    \caption{Dataset 3: APHC Scale}
    \label{APHC1HLam1}
\end{figure}

\begin{figure}[H]
    \centering
    \subfigure[]{\includegraphics[width=0.15\textwidth]{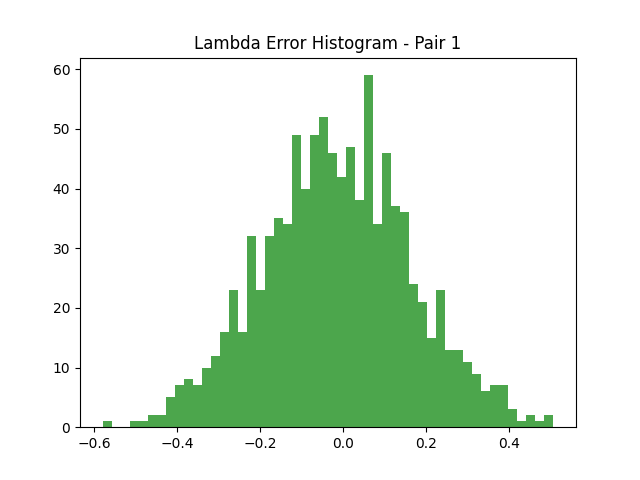}}
    \subfigure[]{\includegraphics[width=0.15\textwidth]{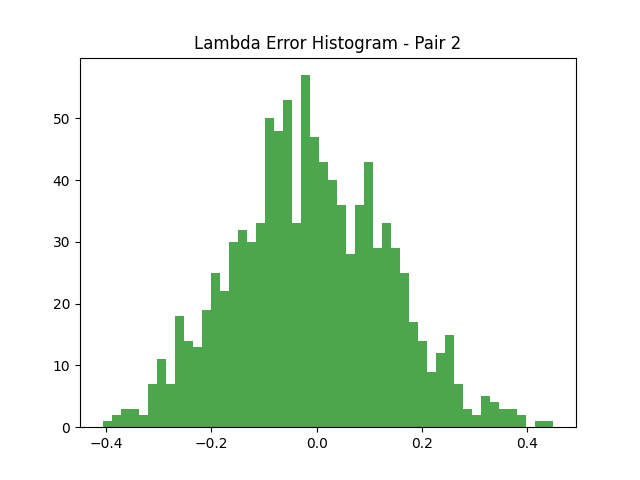}}
      \subfigure[]{\includegraphics[width=0.15\textwidth]{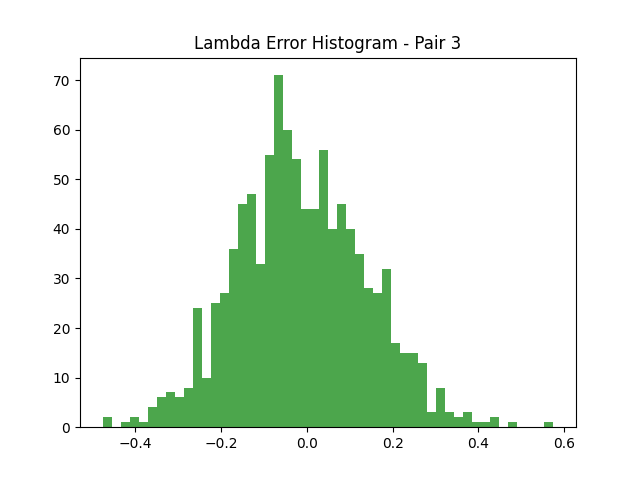}}
        \subfigure[]{\includegraphics[width=0.15\textwidth]{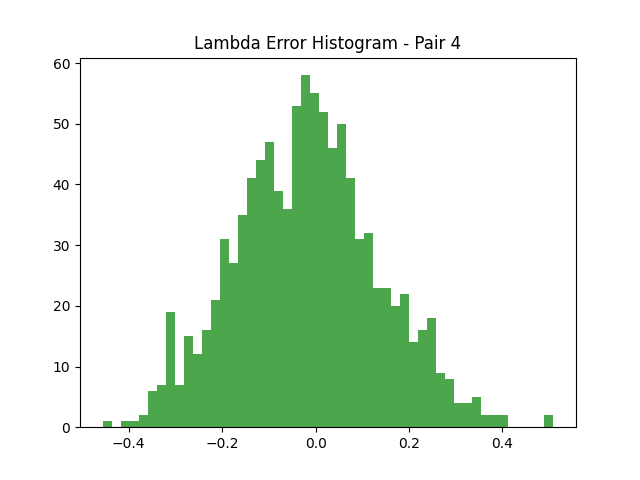}}
          \subfigure[]{\includegraphics[width=0.15\textwidth]{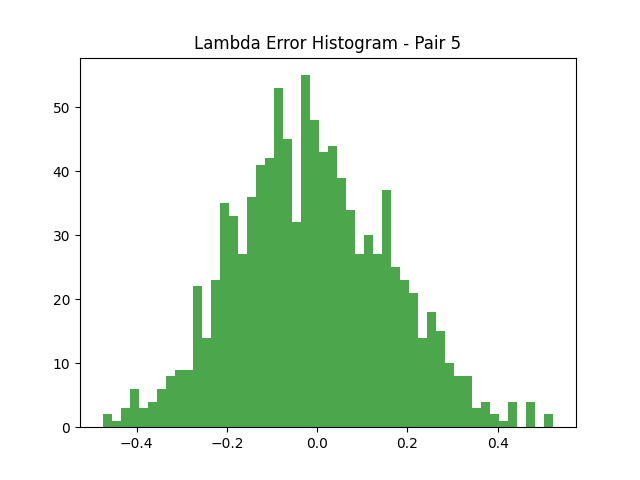}}
            \subfigure[]{\includegraphics[width=0.15\textwidth]{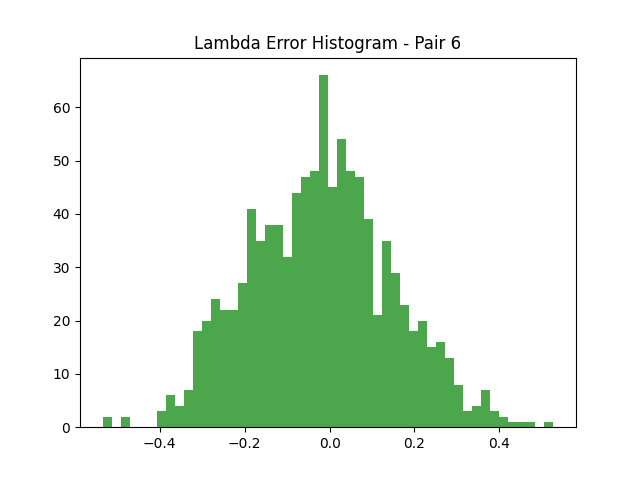}}
              \subfigure[]{\includegraphics[width=0.15\textwidth]{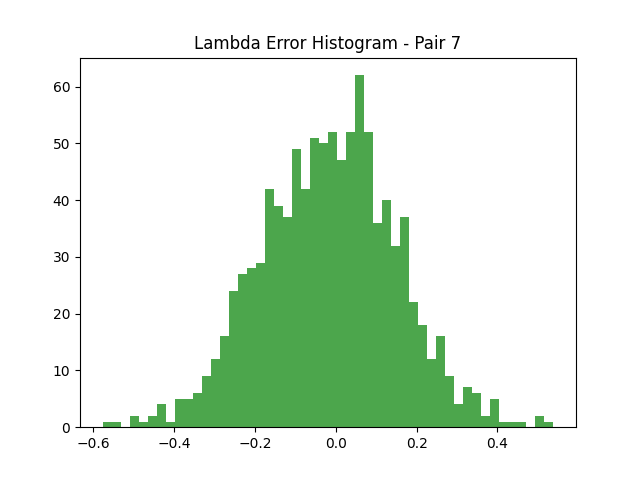}}
                \subfigure[]{\includegraphics[width=0.15\textwidth]{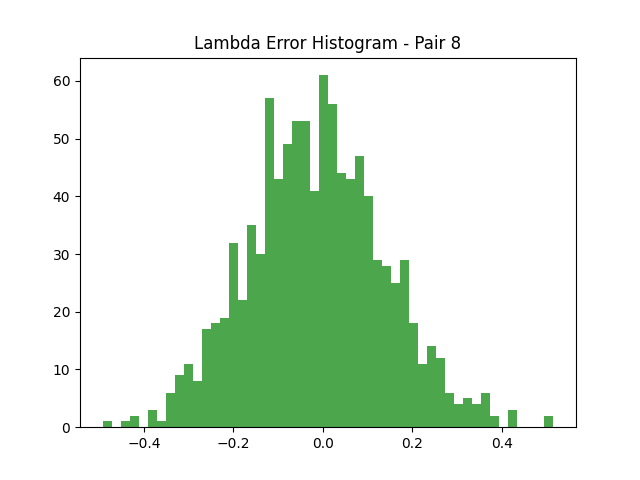}}
                  \subfigure[]{\includegraphics[width=0.15\textwidth]{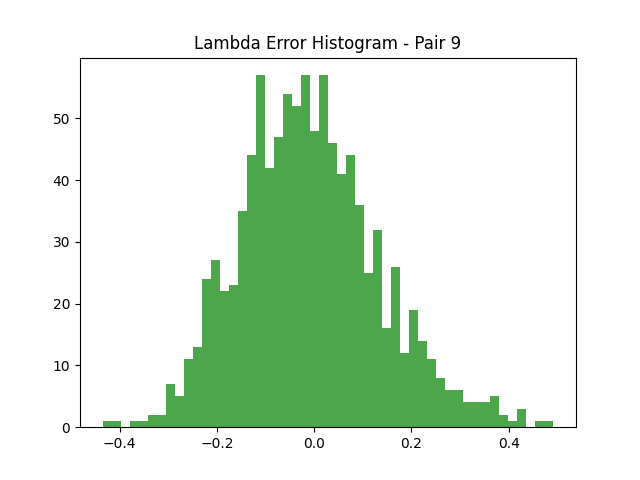}}
                    \subfigure[]{\includegraphics[width=0.15\textwidth]{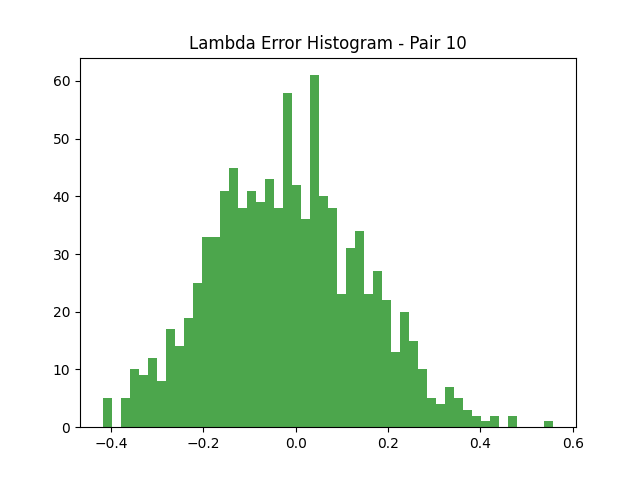}}
                    \caption{Dataset 9: APHC Scale}
    \label{APHC1HLam2}
\end{figure}

\begin{figure}[H]
    \centering
    \subfigure[]{\includegraphics[width=0.15\textwidth]{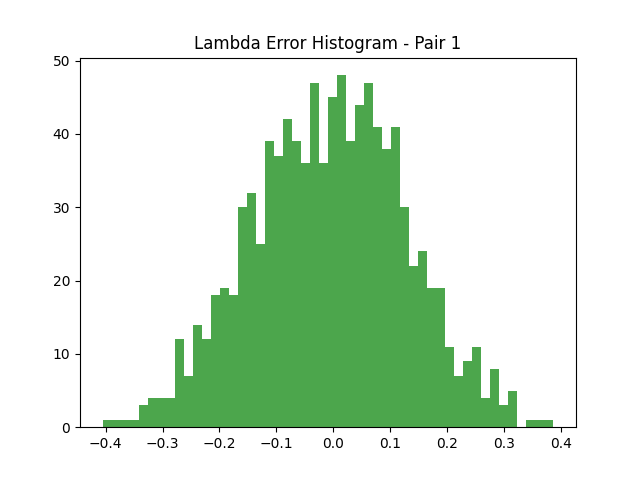}}
    \subfigure[]{\includegraphics[width=0.15\textwidth]{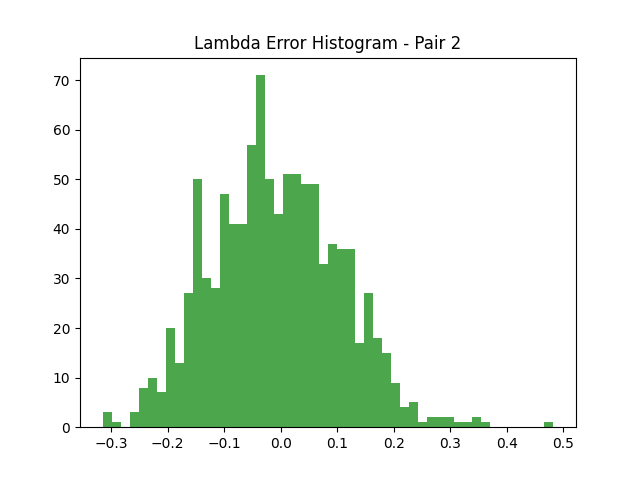}}
      \subfigure[]{\includegraphics[width=0.15\textwidth]{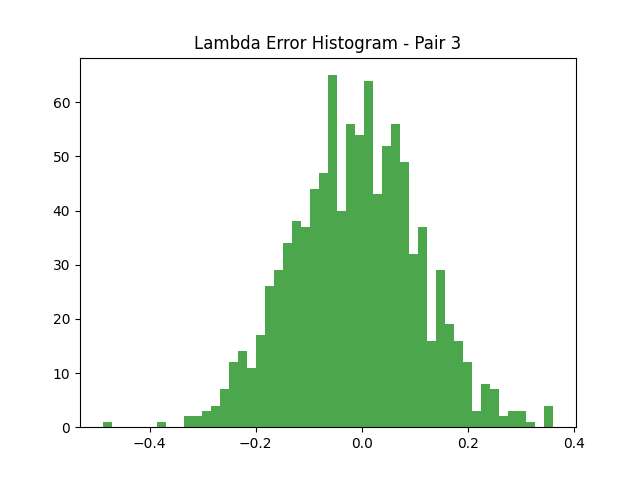}}
        \subfigure[]{\includegraphics[width=0.15\textwidth]{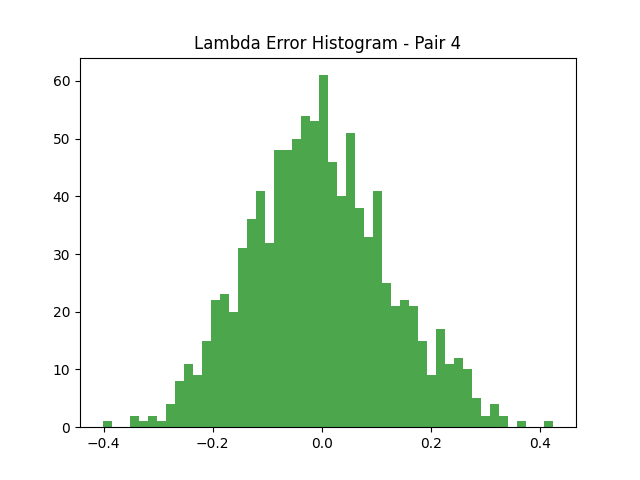}}
          \subfigure[]{\includegraphics[width=0.15\textwidth]{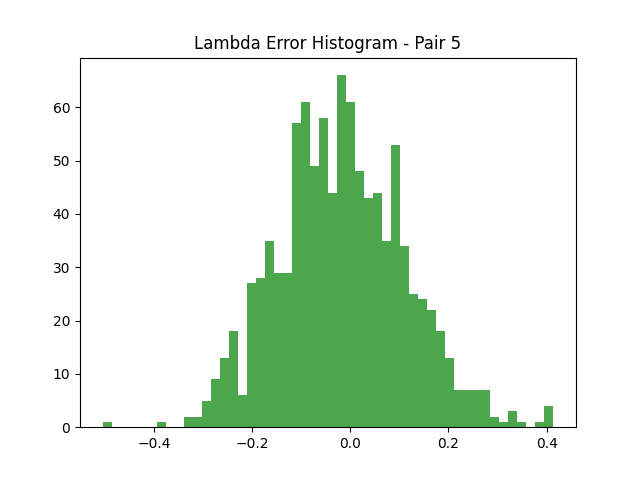}}
            \subfigure[]{\includegraphics[width=0.15\textwidth]{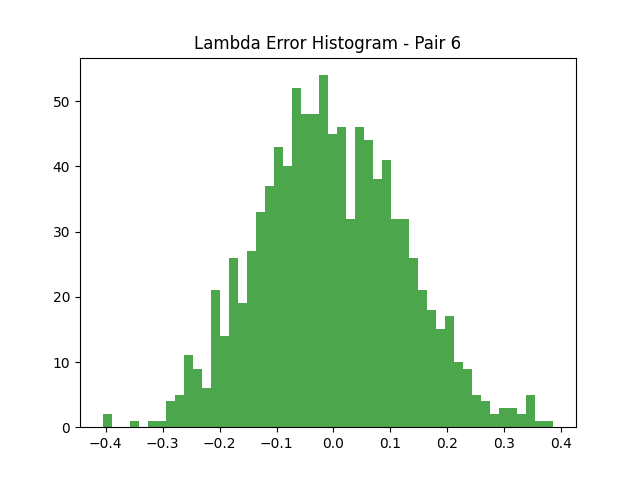}}
              \subfigure[]{\includegraphics[width=0.15\textwidth]{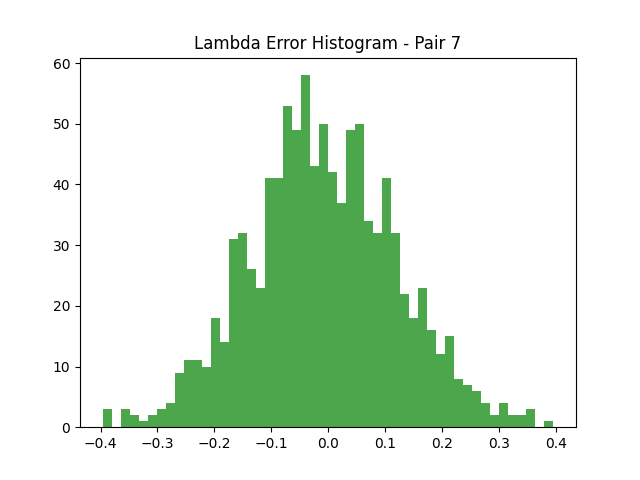}}
                \subfigure[]{\includegraphics[width=0.15\textwidth]{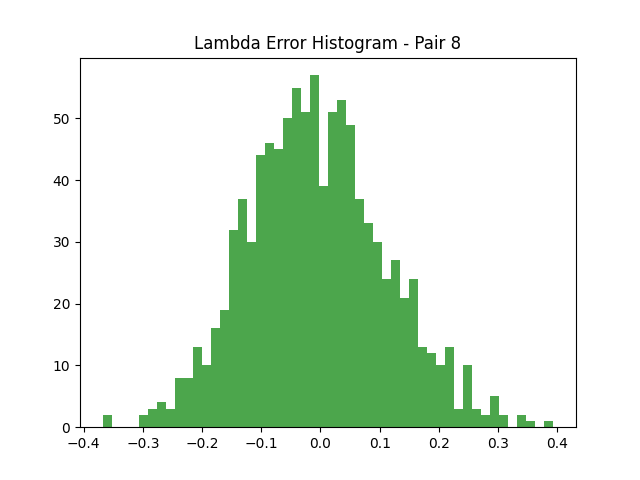}}
                  \subfigure[]{\includegraphics[width=0.15\textwidth]{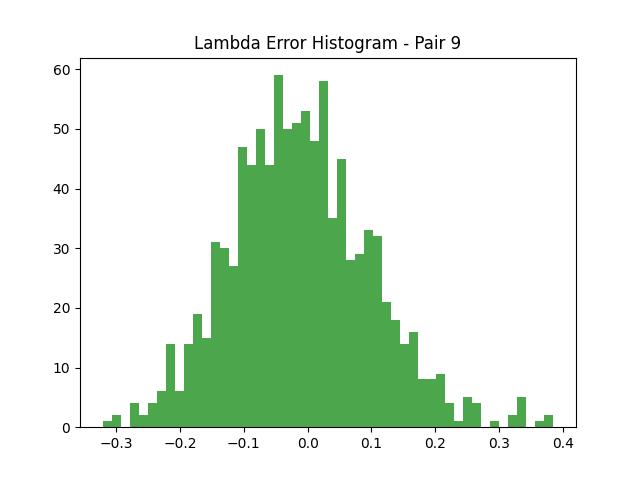}}
                    \subfigure[]{\includegraphics[width=0.15\textwidth]{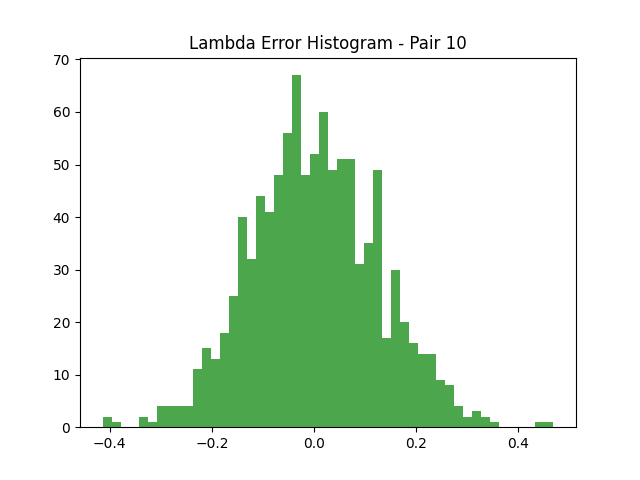}}
                    \caption{Dataset 15: APHC Scale}
    \label{APHC1HLam15}
\end{figure}

\end{document}